\RequirePackage{silence}
\documentclass[11pt]{article}

\usepackage{titlesec}

\usepackage{xcolor}
\usepackage{amsfonts,amsmath,amssymb,amsthm,boxedminipage,color,url,fullpage}
\definecolor{weborange}{rgb}{.8,.3,.3}
\definecolor{webblue}{rgb}{0,0,.8}
\definecolor{internallinkcolor}{rgb}{0,.5,0}
\definecolor{externallinkcolor}{rgb}{0,0,.5}

\providecommand{\remove}[1]{}

\usepackage{amsthm} 

\usepackage{comment}

\usepackage{enumerate,paralist}
\usepackage[labelfont=bf]{caption}

\usepackage{aliascnt}

\usepackage[
hyperfootnotes=false,
colorlinks=true,
urlcolor=externallinkcolor,
linkcolor=internallinkcolor,
filecolor=externallinkcolor,
citecolor=internallinkcolor,
breaklinks=true,
pdfstartview=FitH,
pdfpagelayout=OneColumn]{hyperref}

\usepackage[backend=bibtex,style=ieee-alphabetic,natbib=true,backref=true]{biblatex} 
\usepackage{cleveref}
\usepackage{xspace}
\usepackage{xstring}
\usepackage{tikz}\usetikzlibrary{arrows}
\usetikzlibrary{decorations.markings}
\usepackage{subcaption}
\usepackage{tabularx}
\usepackage{bbold}
\usepackage{dsfont}
\usepackage{titlesec}
\usepackage{hyperref}

\providecommand{\remove}[1]{}
\newcommand{\Draft}[1]{\ifdefined\IsDraft\texttt{ #1} \fi}

\ifdefined\IsDraft
    \newcommand{\authnote}[2]{{\bf [{\color{red} #1's Note:} {\color{blue} #2}]}}
\else
    \newcommand{\authnote}[2]{}
\fi

\titleclass{\subsubsubsection}{straight}[\subsection]

\newcounter{subsubsubsection}[subsubsection]
\renewcommand\thesubsubsubsection{\thesubsubsection.\arabic{subsubsubsection}}

\titleformat{\subsubsubsection}
{\normalfont\normalsize\bfseries}{\thesubsubsubsection}{1em}{}
\titlespacing*{\subsubsubsection}
{0pt}{3.25ex plus 1ex minus .2ex}{1.5ex plus .2ex}

\makeatletter
\renewcommand\paragraph{\@startsection{paragraph}{5}{\z@}%
	{3.25ex \@plus1ex \@minus.2ex}%
	{-1em}%
	{\normalfont\normalsize\bfseries}}
\renewcommand\subparagraph{\@startsection{subparagraph}{6}{\parindent}%
	{3.25ex \@plus1ex \@minus .2ex}%
	{-1em}%
	{\normalfont\normalsize\bfseries}}
\def\toclevel@subsubsubsection{4}
\def\toclevel@paragraph{5}
\def\toclevel@paragraph{6}
\def\l@subsubsubsection{\@dottedtocline{4}{7em}{4em}}
\def\l@paragraph{\@dottedtocline{5}{10em}{5em}}
\def\l@subparagraph{\@dottedtocline{6}{14em}{6em}}
\makeatother

\newcommand{\sdotfill}{\textcolor[rgb]{0.8,0.8,0.8}{\dotfill}} 

\newenvironment{protocol}{\begin{proto}}{\end{proto}}

\newenvironment{algorithm}{\begin{algo}}{\vspace{-\topsep}\end{algo}}

\let\originalleft\left
\let\originalright\right
\renewcommand{\left}{\mathopen{}\mathclose\bgroup\originalleft}
\renewcommand{\right}{\aftergroup\egroup\originalright}

\newcommand{\aka} {also known as,\xspace}
\newcommand{\resp}{resp.,\xspace}
\newcommand{\ie}  {i.e.,\xspace}
\newcommand{\eg}  {e.g.,\xspace}

\newcommand{\wrt} {with respect to\xspace}
\newcommand{\wlg} {without loss of generality\xspace}
\newcommand{\cf}{cf.,\xspace}

\newcommand{\abs}[1]{\left\lvert #1 \right\rvert}

\newcommand{\set}[1]{\ens{#1}}

\newcommand{\paren}[1]{\left(#1\right)}

\newcommand{\eqdef}{:=}

\newcommand{\N}{{\mathbb{N}}}

\newcommand{\F}{{\cal F}}

\newcommand{\naturals}{{\mathbb{N}}}
\newcommand{\zo}{\set{0,1}}

\newcommand{\zn}{{\zo^n}}

\newcommand{\condition}{\;\ifnum\currentgrouptype=16 \middle\fi|\;}

\newcommand{\eps}{\varepsilon}

\newcommand{\la}{\gets}

\makeatletter
\@ifpackageloaded{complexity}{}{%
\usepackage[classfont=roman,langfont=roman,funcfont=roman]{complexity}}
\makeatother

\newcommand{\diver}{D}
\newcommand{\re}[2]{\diver\paren{#1 || #2}}

\newcommand{\Enc}{\operatorname{Enc}}
\newcommand{\Dec}{\operatorname{Dec}}
\newcommand{\Gen}{\operatorname{Gen}}

\newcommand{\negl}{\operatorname{neg}}

\newcommand{\Supp}{\operatorname{Supp}}

\newcommand{\zfrac}[2]{#1/#2}

\newcommand{\Ensuremath}[1]{\ensuremath{#1}\xspace}

\newcommand{\tth}[1]{\Ensuremath{#1^{\rm th}}}

\newcommand{\ith}{\tth{i}}
\newcommand{\jth}{\tth{j}}

\renewcommand{\cref}{\Cref}
\usepackage{aliascnt}

\newtheorem{theorem}{Theorem}[section]

\newaliascnt{lemma}{theorem}
\newtheorem{lemma}[lemma]{Lemma}
\aliascntresetthe{lemma}
\crefname{lemma}{Lemma}{Lemmas}

\newaliascnt{claim}{theorem}
\newtheorem{claim}[claim]{Claim}
\aliascntresetthe{claim}
\crefname{claim}{Claim}{Claims}

\newaliascnt{corollary}{theorem}
\newtheorem{corollary}[corollary]{Corollary}
\aliascntresetthe{corollary}
\crefname{corollary}{Corollary}{Corollaries}

\newaliascnt{construction}{theorem}

\aliascntresetthe{construction}
\crefname{construction}{Construction}{Constructions}

\newaliascnt{fact}{theorem}
\newtheorem{fact}[fact]{Fact}
\aliascntresetthe{fact}
\crefname{fact}{Fact}{Facts}

\newaliascnt{proposition}{theorem}
\newtheorem{proposition}[proposition]{Proposition}
\aliascntresetthe{proposition}
\crefname{proposition}{Proposition}{Propositions}

\newaliascnt{conjecture}{theorem}

\aliascntresetthe{conjecture}
\crefname{conjecture}{Conjecture}{Conjectures}

\newaliascnt{definition}{theorem}
\newtheorem{definition}[definition]{Definition}
\aliascntresetthe{definition}
\crefname{definition}{Definition}{Definitions}

\newaliascnt{notation}{theorem}
\newtheorem{notation}[notation]{Notation}
\aliascntresetthe{notation}
\crefname{notation}{Notation}{Notation}

\newaliascnt{assertion}{theorem}
\newtheorem{assertion}[assertion]{Assertion}
\aliascntresetthe{assertion}
\crefname{assertion}{Assertion}{Assertion}

\newaliascnt{assumption}{theorem}
\newtheorem{assumption}[assumption]{Assumption}
\aliascntresetthe{assumption}
\crefname{assumption}{Assumption}{Assumption}

\newaliascnt{remark}{theorem}
\newtheorem{remark}[remark]{Remark}
\aliascntresetthe{remark}
\crefname{remark}{Remark}{Remarks}

\newaliascnt{example}{theorem}
\ifdefined\excludeexample
\else
   
\fi

\aliascntresetthe{example}
\crefname{exmaple}{Example}{Examples}

\crefname{equation}{Equation}{Equations}

\newaliascnt{proto}{theorem}

\newtheorem{proto}[proto]{Protocol}

\aliascntresetthe{proto}
\crefname{proto}{protocol}{protocols}

\newaliascnt{algo}{theorem}
\newtheorem{algo}[algo]{Algorithm}
\aliascntresetthe{algo}
\crefname{algo}{algorithm}{algorithms}

\newaliascnt{expr}{theorem}
\newtheorem{expr}[expr]{Experiment}
\aliascntresetthe{expr}
\crefname{experiment}{experiment}{experiments}

\newcommand{\stepref}[1]{Step~\ref{#1}}

\def\FullBox{$\Box$}
\def\qed{\ifmmode\qquad\FullBox\else{\unskip\nobreak\hfil
\penalty50\hskip1em\null\nobreak\hfil\FullBox
\parfillskip=0pt\finalhyphendemerits=0\endgraf}\fi}

\def\qedsketch{\ifmmode\Box\else{\unskip\nobreak\hfil
\penalty50\hskip1em\null\nobreak\hfil$\Box$
\parfillskip=0pt\finalhyphendemerits=0\endgraf}\fi}

{\proof[#1]\list{}{\leftmargin=1cm\rightmargin=1cm}}
{\endproof%
  \vspace{10pt}
}

\newcommand{\eex}[2]{\Ex_{#1}\left[#2\right]}
\newcommand{\ex}[1]{\Ex\left[#1\right]}
\newcommand{\Ex}{{\mathrm E}}
\renewcommand{\Pr}{{\mathrm {Pr}}}
\newcommand{\pr}[1]{\Pr\left[#1\right]}
\newcommand{\ppr}[2]{\Pr_{#1}\left[#2\right]}

\newcommand{\Sa}{\mathsf{S}}
\newcommand{\Ra}{\mathsf{R}}

\newcommand{\Sc}{\Sa}
\newcommand{\Rc}{\Ra}

\newcommand{\pT}{T'}

\newcommand{\hS}{\widehat{S}}
\renewcommand{\H}{H}

\newcommand{\ens}[1]{\{#1\}}
\newcommand{\size}[1]{\left|#1\right|}
\newcommand{\ssize}[1]{|#1|}

\newcommand{\tP}{\widetilde{P}}

\newcommand{\Uni}{{\mathord{\mathcal{U}}}}

\newcommand{\prob}[1]{\mathsf{\textsc{#1}}}

\newcommand{\SD}{\prob{SD}}

\providecommand{\cL}{{\cal{L}}}

\newcommand{\1}{\mathds{1}}

\newcommand{\I}{\mathcal{I}}

\newcommand{\pptm}{{\sc pptm}\xspace}
\newcommand{\ppt}{{\sc ppt}\xspace}

\newcommand{\cB}{\mathcal{B}}
\newcommand{\cG}{\mathcal{G}}
\newcommand{\cU}{\mathcal{U}}

\newcommand{\cs}{{\cal{S}}}

\newcommand{\cJ}{{\cal{J}}}

\newcommand{\cx}{{\cal{X}}}

\newcommand{\cX}{{\cal{X}}}
\newcommand{\cZ}{{\cal{Z}}}

\newcommand{\cE}{{\mathcal{E}}}

\newcommand{\ct}{{\cal{T}}}

\newcommand{\Tableofcontents}{
\thispagestyle{empty}
\pagenumbering{gobble}
\clearpage
\tableofcontents
\thispagestyle{empty}
\clearpage
\pagenumbering{arabic}
}

\newcommand{\bX}{X}
\newcommand{\bx}{x}
\newcommand{\bY}{Y}
\newcommand{\by}{y}

\newcommand{\indic}[1]{\mathds{1}_{\set{#1}}}

\DeclareMathOperator{\V}{V}
\newcommand{\Vc}{\V}
\newcommand{\Pc}{\P}

\DeclareMathOperator{\Bern}{Bern}
\DeclareMathOperator{\Bin}{Bin}

\DeclareMathOperator{\round}{round}

\makeatletter
\let\xx@thm\@thm
\AtBeginDocument{\let\@thm\xx@thm}
\makeatother

\newcommand{\hR}{\widehat{R}}

\usepackage{thmtools}
\usepackage{thm-restate}

\newcommand{\Inote}[1]{\authnote{Iftach}{#1}}
\newcommand{\Enote}[1]{\authnote{Eliad}{#1}}

\newcommand{\dval}{d}
\newcommand{\tdelta}{\tilde{\delta}}

\newcommand{\tpi}{\widetilde{\pi}}

\newcommand{\Unf}{P}
\newcommand{\Idl}{\tP}

\newcommand{\Rll}{Q}

\newcommand{\sP}{{\P^\ast}}
\newcommand{\nsP}{{{\P^n}^\ast}}
\newcommand{\nV}{\Vc^{n}}
\newcommand{\nP}{\Pc^{n}}

\newcommand{\tV}{\widetilde{\V}}
\newcommand{\hV}{\widehat{\V}}
\newcommand{\ntV}{\tV^{n}}
\newcommand{\nhV}{\hV^{n}}
\newcommand{\snsP}{{\nsP}}

\newcommand{\Ideal}{\ensuremath{\operatorname{Winning}}\xspace}
\newcommand{\Real}{\ensuremath{\operatorname{Attacking}}\xspace}

\newcommand{\dP}{P}
\newcommand{\dQ}{Q}
\newcommand{\eB}{B}
\newcommand{\tB}{\widetilde{B}}

\newcommand{\tY}{\widetilde{Y}}
\newcommand{\tR}{\widetilde{R}}

\title{A Tight Parallel Repetition Theorem for\\ Partially Simulatable  Interactive Arguments  \\
	via {\Large Smooth KL-Divergence}
	 \Draft{\\{\small \sc Working Draft: Please Do Not Distribute}}}
 \author{Itay Berman\thanks{MIT. E-mail:\texttt{itayberm@mit.edu}. Research supported
 		in part by NSF Grants CNS-1413920 and CNS-1350619, and by the Defense
 		Advanced Research Projects Agency (DARPA) and the U.S. Army Research Office
 		under contracts W911NF-15-C-0226 and W911NF-15-C-0236.} \and
 	Iftach Haitner\thanks{School of Computer Science, Tel Aviv
 		University. E-mail:\{\texttt{iftachh@cs.tau.ac.il},
 		\texttt{eliadtsf@tau.ac.il}\}. Research supported by ERC starting grant 638121 and Israel Science Foundation grant   666/19.}~\thanks{Member of the Check Point Institute for Information Security.}
 	\and Eliad Tsfadia$^{\dagger}$}

\begin{document}
\sloppy
\maketitle

\begin{abstract}
Hardness  amplification is a central problem in the study of  interactive protocols. While ``natural'' parallel repetition transformation is known to reduce the soundness error of some special cases of interactive arguments:  three-message protocols (\citeauthor*{BellareIN97} [FOCS '97]) and public-coin protocols (\citeauthor*{HastadPWP10} [TCC '10], \citeauthor{ChungL10} [TCC '10] and  \citeauthor{ChungP15} [TCC '15]), it fails to do so in  the general case (the above \citeauthor{BellareIN97}; also \citeauthor{PietrzakW12} [TCC '07]). 
	
The only known round-preserving approach that applies to all  interactive arguments is \citeauthor{HaitnerPR13}'s \textit{random-terminating} transformation [SICOMP '13], who showed that the parallel repetition of the transformed protocol reduces the soundness error at a \emph{weak}  exponential rate:  if  the original $m$-round protocol has  soundness  error $1-\eps$, then   the $n$-parallel repetition of its random-terminating variant  has soundness error  $(1-\eps)^{\eps n / m^4}$ (omitting constant factors). \citeauthor{HastadPWP10}  have generalized this result to  \textit{partially simulatable  interactive arguments},  showing that the $n$-fold repetition of an $m$-round  $\delta$-simulatable  argument of soundness error  $1-\eps$  has soundness error  $(1-\eps)^{\eps \delta^2 n / m^2}$.  When applied to  random-terminating arguments, the \citeauthor{HastadPWP10}  bound matches  that of \citeauthor{HaitnerPR13}.   

In this work we prove that parallel repetition of random-terminating arguments reduces the soundness error at a much stronger exponential rate: the soundness error of the $n$ parallel repetition is $(1-\eps)^{n / m}$, only an $m$ factor from the optimal rate of $(1-\eps)^n$ achievable in public-coin and three-message arguments. The result generalizes  to $\delta$-simulatable  arguments, for which we prove a bound of  $(1-\eps)^{\delta n / m}$. This is achieved by presenting a tight bound on  a relaxed variant of the KL-divergence between the distribution induced by our reduction and its ideal variant,   a result whose scope extends beyond parallel repetition proofs. We prove the tightness of the above bound for random-terminating arguments, by presenting a matching protocol.

\end{abstract}

\Tableofcontents

\section{Introduction}\label{sec:Introduction}
Hardness amplification is a central question in the study of computation: can  a somewhat secure   primitive  be made fully secure,  and, if so, can this  be accomplished without loss (\ie  while preserving certain desirable properties the original primitive may have).   In this paper we focus on better understanding the above question \wrt interactive arguments (\aka computationally sound proofs).  In an interactive argument, a {prover} tries to convince a verifier in the validity of a statement. The basic properties of such proofs are \emph{completeness} and \emph{soundness}. Completeness means that the prover, typically using  some extra information, convinces the verifier to accept valid statements with high probability. Soundness means that  a cheating \emph{polynomial-time} prover cannot convince the verifier to accept invalid statements, except with small probability. Interactive arguments should  be compared with the related notion of \textit{interactive proofs}, whose soundness  should hold against \emph{unbounded} provers. Interactive argument  are important for  being ``sufficiently secure'' proof system that sometimes achieve properties (\eg compactness) that are beyond the reach of  interactive proofs. Furthermore,   the security of many cryptographic protocols (\eg binding of a computationally binding commitment) can be cast as the soundness of a related interactive argument, but  (being computational)  cannot be cast as the soundness of a related interactive proof.

The question of hardness amplification \wrt interactive arguments  is whether   an argument  with \emph{non-negligible} soundness error, \ie a cheating prover can convince the verifier to accept false statements with some non-negligible probability,  can be transformed   into a new argument, with similar properties, of  negligible soundness error (\ie the verifier almost never accepts false statements). The most common paradigm to obtain such an amplification is via \emph{repetition}: repeat the protocol multiple times with independent randomness, and the verifier accepts only if the verifiers of the original protocol accept in \emph{all} executions. Such repetitions can be done in two different ways, sequentially (known as \textit{sequential repetition}), where the $(i+1)$ execution of the protocol starts only after the \ith execution has finished, or in parallel (known as \textit{parallel repetition}), where  the executions are all simultaneous. Sequential repetition is known to reduce the soundness error in most computational models (\cf \citet{DamPfi98}), but has the undesired effect of increasing the round complexity of the protocol. Parallel repetition, on the other hand, does preserve the round complexity, and reduces the soundness error for (single-prover)
interactive proofs (\citet{GoldreichModernCrypto}) and  two-prover  interactive proofs (\citet{Raz98,Holenstein09,Rao11}). Parallel repetition was also shown to reduce the soundness error in three-message arguments (\cite*{BellareIN97}) and public-coin arguments (\citet*{HastadPWP10,ChungLu,ChungP15}). Unfortunately, as shown by \citet{BellareIN97}, and by \citet{PietrzakW12}, parallel repetition \emph{might not} reduce the soundness error of any interactive argument: assuming common cryptographic assumptions, \cite{PietrzakW12} presented an 8-message interactive proof with constant soundness error, whose parallel repetition, for \emph{any} polynomial number of repetitions, still has a constant soundness error.

Faced with the above barrier, \citet{HaitnerPR13} presented a simple method for transforming any interactive argument $\pi$ into a slightly modified protocol $\tpi$, such that the parallel repetition of $\tpi$ does reduce the soundness
error. Given any $m$-round interactive protocol $\pi= (\P,\V)$, let $\tV$ be the following \textit{random-terminating variant} of $\V$: in each round, $\tV$ flips a coin that takes one with probability $\zfrac1{m}$ and zero otherwise. If the
coin outcome is one, $\tV$ accepts and aborts the execution. Otherwise, $\tV$ acts as $\V$ would, and continues to the next round. At the end of the prescribed execution, if reached, $\tV$ accepts if and only if $\V$ would.  Observe that if the
original protocol $\pi$ has soundness error $1-\eps$, then the new protocol $\tpi = (\P,\tV)$ has soundness error $1-\zfrac{\eps}4$ (\ie only slightly closer to one). \citet{HaitnerPR13} proved that the parallel repetition of
$\tpi$ does reduce the soundness error (for any protocol $\pi$). \citet*{HastadPWP10} have generalized the above to \textit{partially-simulatable interactive arguments}, a family of interactive arguments that contains  the random-terminating variant protocols as a special case.    An interactive argument $\pi= (\Pc,\Vc)$   is $\delta$-simulatable if given  any partial view $v$ of  an efficient prover $\sP$ interacting with $\Vc$, the verifier's future messages in  $(\sP,\Vc)$ can be simulated with probability $\delta$. This means that  one can efficiently sample a random continuation of the execution  conditioned on an event of density $\delta$ over $\Vc$'s coins consistent with $v$. It is easy to see that the random-terminating variant of any  protocol is $1/m$ simulatable.  Unfortunately, the soundness bound  proved by   \citet{HaitnerPR13,HastadPWP10}  lags way behind what one might have hoped for, making  parallel repetition impractical in many typical settings.  Assuming a $\delta$-simulatable argument $\pi$ has soundness error is $1-\eps$, then $\pi^n$, the $n$-parallel repetition of $\pi$, was shown to have  soundness error $(1-\eps)^{\eps \delta^2 n/m^2}$ (equals $(1-\eps)^{\eps n/m^4}$ if $\pi$ is a random-terminating variant),  to be compared with the  $(1-\eps)^n$  bound achieved by parallel repetition of interactive proofs, and by three-message and public-coin interactive arguments.\footnote{\label{fn:1}As in all known amplifications of computational hardness, and proven to be an inherent limitation (at least to   some extent) in \citet{Dodis12}, the improvement in the soundness error does  not go below negligible. We ignore this subtly in the introduction. We also ignore constant factors in the exponent.}   Apart from the intellectual challenge, improving the above bound is important  since repeating the random-termination variant  in parallel  is the \emph{only} known unconditional round-preserving amplification method for arbitrary interactive arguments.

\subsection{Proving  Parallel Repetition}\label{sec:Intro:ProvingPR}
Let  $\pi = (\Pc,\Vc)$ be an interactive argument with assumed soundness error $1-\eps$, \ie a polynomial time prover cannot make the verifier accept  a false statement with probability larger than  $1-\eps$. Proving amplification theorems for such proof systems is done via reduction: assuming the  existence of   a cheating prover $\nsP$ making all the $n$ verifiers in  $n$-fold protocol $\pi^n = (\nP,\nV)$   accept a false statement   ``too well'' (\eg more than $(1-\eps)^n$), this prover  is  used  to construct a cheating prover $\sP$   making $\Vc$ accept this false statement with probability larger than $1-\eps$, yielding a contradiction.  Typically, the cheating prover  $\sP$ emulates an execution of $(\nsP,\nV)$ while \textit{embedding} the (real) verifier $\Vc$ as one of the $n$ verifiers (\ie by embedding its messages).   Analyzing the success probability of this $\sP$  is directly reduced to bounding  the ``distance'' (typically statistical  distance or KL-divergence)  between the following  \Ideal and  \Real distributions: the \Ideal distribution is   the  $n$ verifiers' messages  distribution in a winning (all verifiers accept) execution of $(\nsP,\nV)$. The \Real distribution is  the  $n$ verifiers' messages distribution in the emulated execution done by $\sP$ (when interacting with $\Vc$).  

If the verifier is public-coin, or if the prover is  unrestricted (as in single-prover interactive proofs), an optimal strategy for $\sP$ is sampling the emulated verifiers messages uniformly at random conditioned on all verifiers accept, and the  messages so far.   \citet{HastadPWP10}  have bounded the statistical distance between the induced \Ideal and \Real distributions in such a case,  while  \citet{ChungP15} gave a tight bound for the KL-divergence between these distributions,  yielding an optimal result for public-coin arguments.  

For non public-coin protocols, however, a computationally bounded prover cannot always perform the above sampling task (indeed, this inability underneath the counter examples for parallel  repetition of such arguments). However, if the argument is  random terminating,  the cheating prover can sample the following  ``skewed'' variant of the desired distribution:  it samples as described above, but  conditioned that the real verifier  \emph{aborts at the end of the current round}, making the simulation of its  future messages trivial. More generally, for  partially-simulatable arguments, the cheating prover   samples the future messages of the real verifier using the built-in mechanism for sampling a skewed sample of its coins.  Analyzing the prover success probability for such an attack, and thus upper-bounding the soundness error  of the parallel repetition of such arguments, reduces to understanding the (many-round) skewed distributions induced by the above attack. This will be discussed in the next section. 

\subsection{Skewed Distributions}\label{sec:Intro:SkewedDIst}
The \Real distribution induced by the security proof of parallel repetition of partially-simulatable arguments discussed in \cref{sec:Intro:ProvingPR}, gives rise to the following notion of (many-round)   skewed distributions.  Let $\Unf= \Unf_X$ be a distribution over an $m \times n$ size  matrices, letting $\Unf_{X_i}$ and $\Unf_{X^j}$ denoting the induced distribution over the  \ith row and \jth column of $X$, respectively. For an   event  $W$,   let $\Idl = \Unf|W$. The following  distribution  $\Rll_{X,J}$ is a  skewed variant of $\Idl$  induced by an event family  $\cE= \set{E_{i,j}}_{i\in[m],j\in[n]}$ over $\Unf$:  let  $ \Rll_{J} = U_{[n]}$, and let
\begin{align}\label{eq:intro:Q}
	\Rll_{\bX|J} = \prod_{i=1}^{m} \Unf_{X_{i,J} | X_{<i,J}} \Idl_{X_{i,-J} | \bX_{<i}, X_{i,J}, E_{i,J}} 
\end{align}
for $X_{<i} = (X_{1},\ldots,X_{i-1})$, $X_{<i,j}= (X_{<i})^j = (X_{1,j},\ldots,X_{i-1,j})$ and  $X_{i,-j}= X_{i,[n] \setminus \set{j}}$.  That is,  $\Rll$  induced by first  sampling $J\in [n]$ uniformly at random, and then sampling the following skewed variant of $\Idl$: At round $i$
\begin{enumerate}
	\item Sample $X_{i,J}$ according to  $\Unf_{\bX_{i,J}| \bX_{<i,J}} $ (rather than  $\Unf_{\bX_{i,J}| \bX_{<i},W} $ as in  $\Idl$), 
	\item Sample $X_{i,-J}$ according $\Idl_{\bX_{i,-J}| \bX_{<i},X_{i,J},E_{i,J}}$ (rather than  $\Idl_{\bX_{i,J}| \bX_{<i},X_{i,J}}$).
\end{enumerate}

At a first glance,  the  distribution $\Rll$ looks somewhat arbitrary. Nevertheless, as we explain below, it naturally arises in the analysis of parallel repetition theorem of partially-simulatable interactive arguments, and thus of random-terminating variants. Somewhat similar skewed distributions  also come up when proving parallel repetition  of two-prover proofs, though there we only  care for single round distributions, \ie $m=1$.  

The distributions $\Idl$  and $\Rll$  relate to the \Ideal and \Real  distributions described in  \cref{sec:Intro:ProvingPR} in the following way: let $\pi= (\Pc,\Vc)$ be an $m$-round $\delta$-simulatable  argument,  and let $\nsP$ be an efficient  (for simplicity) deterministic  cheating prover  for $\pi^n$. Let   $\Unf$ to be the distribution of the $n$ verifiers messages  in a random  execution of $\pi^n$, and let $W$ be the event that $\nsP$ wins in $(\nsP,\Vc^n)$. By definition, $\Idl= \Unf|W$ is just the \Ideal distribution. Assume for sake of simplicity that $\Vc$ is a random-termination variant (halts at the end of each round with probability $1/m$), let $E_{i,j}$ be the set of coins in which the \jth verifier halts at the end of the \ith round of $(\nP,\nV)$,  and let $\Rll = \Rll(\Unf,W,\set{E_{i,j}})$ be according to \cref{eq:intro:Q}.  Then, ignoring some efficiency  concerns,  $\Rll$ is just the \Real distribution. Consequently, a bound on the soundness error of $\pi^n$ can be proved via the following result:

\begin{lemma}[informal]\label{lemma:intro:KLtoPR}
	Let $\pi$ be a partially simulatable argument of soundness error $(1-\eps)$. Assume that for every efficient   cheating prover  for $\pi^n$ and every event $T$, it holds that  
	$$ \ppr{\Rll_X}{T} \le   \ppr{\Idl_X}{T} + \gamma $$
	where  $W$,  $\Idl$  and $\Rll$ are as defined  above \wrt this adversary, and that $\Rll$ is efficiently samplable.  Then  $\pi^n$ has soundness error $(1-\eps)^{\log (1/\Unf[W])/\gamma}$.
\end{lemma}

It follows that proving a parallel repetition theorem for partially simulatable arguments, reduces to  proving that low probability events in $\Idl_X$ have  low probability in $\Rll_X$ (for the sake of the introduction, we ignore the less fundamental   samplability condition assumed for $\Rll$). One can try to prove the latter, as implicitly done in \cite{HaitnerPR13,HastadPWP10},  by bounding the \emph{statistical distance} between $\Idl$ and $\Rll$  (recall that  $\SD(\dP,\dQ) = \max_E (\ppr{\dP}{E} - \ppr{\dQ}{E})$).  This approach, however, seems doomed to give non-tight bounds for several  reasons: first,  statistical distance is not geared to bound non-product  distributions  (\ie iterative processes) as the one defined by     $\Rll$,  and one is forced  to use a wasteful  hybrid argument in order to bound the statistical distance of such distributions. A second reason is that    statistical distance  bounds  the difference in probability between the two distributions  for \emph{any} event, where we only care that this difference is small for low (alternatively,  high) probability events. In many settings, achieving this  (unneeded) stronger guarantee inherently yields a weaker bound. 

What seems to be a more promising approach is bounding the \emph{KL-divergence} between  $\Idl$ and $\Rll$ (recall that    $D(\dP||\dQ)=\Ex_{x\sim \dP}\log\frac{\dP(x)}{\dQ(x)}$). Having a chain rule, KL-divergence is typically an excellent choice for non-product  distributions. In particular, bounding it only requires understanding the non-product  nature (\ie the dependency between the different entries) of the  left-hand-side distribution. This  makes KL-divergence  a very useful measure in settings where the iterative nature of  the right-hand-side distribution is  much more complicated. Furthermore, a small KL-divergence guarantees that  low probability events  in $\Idl$ happen with almost the same  probability in $\Rll$, but it only guarantees  a weaker guarantee for other  events (so it has the potential to yield a tighter result). \citet{ChungP15}  took advantage of this observation for  proving their tight bound on parallel repetition of public-coin argument by bounding the KL-divergence between their variants of $\Idl$ and $\Rll$. Unfortunately, for partially simulatable (and for random terminating)  arguments, the KL-divergence between these distributions might be infinite.

Faced with the above  difficulty, we propose a relaxed variant of KL-divergence that we name \textit{smooth KL-divergence}. On the one hand, this measure  has the  properties of  KL-divergence that make it suitable for our settings. However,  on the other hand, it is less fragile (\ie oblivious to events of small probability), allowing us to tightly bound its value for the distributions under consideration.

\subsection{Smooth KL-divergence}\label{sec:Intro:SmoothKL}
The KL-divergence between distributions $\dP$ and $\dQ$ is a very sensitive distance measure: an event $x$ with $\dP(x)\gg \dQ(x)$ might make $D(\dP||\dQ)$ huge even if $\dP(x)$ is tiny (\eg $\dP(x)>0=\dQ(x)$ implies $D(\dP||\dQ)=\infty$). While events of tiny probability are important in some settings, they have no impact in ours. 
So we seek a less sensitive measure that enjoys the major properties of KL-divergence, most notably having chain-rule and  mapping  low probability events to  low probability events. A natural attempt would be to define it as $\inf_{\dP',\dQ'}\set{D(\dP'||\dQ')}$, where the infimum is over all pairs of distributions such that both $\SD(\dP,\dP')$ and $\SD(\dQ,\dQ')$ are small. This relaxation, however, requires an upper bound on the probability of events \wrt $\dQ$, which in our case is the complicated skewed  distribution $\Rll$.  Unfortunately, bounding the probability of events \wrt the   distribution $\Rll$ is exactly the issue in hand. 

Instead, we take advantage of the asymmetric nature of the KL-divergence to propose a relaxation that only requires upper-bounding events \wrt $\dP$, which in our case is the much simpler $\Idl$ distribution. Assume  $\dP$ and $\dQ$ are over a domain $\Uni$.  Then the \emph{$\alpha$-smooth KL-divergence of $\dP$ and $\dQ$}  is defined by
\begin{align*}
	D^{\alpha}(\dP||\dQ) = \inf_{(F_{\dP},F_{\dQ})\in \F} \set{D(F_{\dP}(\dP)||F_{\dQ}(\dQ))}
\end{align*}
for $\F$ being the set of randomized function pairs, such that for every  $(F_{\dP},F_{\dQ}) \in \F$:
\begin{enumerate}
	\item $\Pr_{x\sim \dP}[ F_\dP(x) \neq x]\leq \alpha$.
	
	\item $\forall x\in\Uni$ and  $C \in \set{\dP,\dQ}$:   $F_C(x) \in  \set{x} \cup \overline{\Uni}$. 
\end{enumerate}
Note that for any pair $(F_{\dP},F_{\dQ})\in \F$ and any event $\eB$ over $\Uni$, it holds
that $\Pr_{\dQ}[\eB]\geq \Pr_{F_{\dQ}(\dQ)}[\eB]$, and
$ \Pr_{F_{\dP}(\dP)}[\eB] \ge \Pr_P[\eB] - \alpha$. Thus, if $\Pr_P[\eB]$ is low, a bound
on $D(F_{\dP}(\dP)||F_{\dQ}(\dQ))$ implies that $\Pr_Q[\eB]$ is also low. Namely, low
probability events in $\dP$ happen with low probability also in $\dQ$. 

\paragraph{Bounding smooth KL-divergence.}
Like the  (standard)  notion of KL-divergence, the power of smooth KL-divergence is best manifested when applied to non-product distributions.  Let $\dP$ and $\dQ$ be two distributions for which we would like to prove that small events in $\dP_{X=(X_1,\ldots,X_m)}$ are small in $\dQ_{X=(X_1,\ldots,X_m)}$ (as a running example, let $\dP$ and $\dQ$ be the distributions $\Idl_X$ and $\Rll_{X,J}$  from the previous section, respectively).  By chain rule of KL-divergence, it suffices to show that for some events  $\eB_1,\ldots,\eB_m$  over $\dQ$ (\eg $\eB_i$ is the event that $J|X_{<i}$ has high min entropy) it holds that 
\begin{align}\label{eq:intro:smooth-div-main-prop}
	&\sum_{i=1}^m D(\dP_{X_i} || \dQ_{X_i \mid \eB_{\leq i}} \mid \dP_{X_{<i}}) &\hfil&\hfil     \left(\text{\ie} \sum_{i=1}^m\eex{x\gets \dP_{X_{<i}}}{D\paren{\dP_{X_i \mid X_{<i} = x} ||  \dQ_{X_i \mid X_{<i} = x, \eB_{\leq i}}} } \right) 
\end{align}
is small, and  $\dQ[\eB_{\le m}]$ is  large. 
Bounding \cref{eq:intro:smooth-div-main-prop} only requires understanding $\dP$ and  simplified  variants of $\dQ$ (in which all but the \ith entry is sampled according to $\dP$). Unfortunately, bounding  $\dQ[\eB_{\le m}]$ might be hard since it requires a good understanding of the distribution $\dQ$ itself. We would have liked to  relate the desired bound to $\dP[\eB_{\le m}]$, but the events $\set{\eB_i}$ might  not even be defined over $\dP$ (in the above example, $\dP$ has no $J$ part). However, smooth KL-divergence gives us the means to do almost that.  
\begin{lemma}[Bounding smooth KL-divergence, informal]\label{lemma:intro:smooth-div-main-prop}
	Let $\dP$, $\dQ$ and $\set{\eB_i}$ be as above. Associate the   events   $\set{\tB_i}$ with $\dP$, each  $\tB_i$  (independently) occur with probability  $\dQ[\eB_i \mid \eB_{<i}, X_{<i}]$.	
	Then
	\begin{align*}
		D^{1 - \dP[  \tB_{\le m}]}(\dP_X ||\dQ_X) \leq \sum_{i=1}^m D\paren{\dP_{X_i} || \dQ_{X_i \mid \eB_{\leq i}} \mid \dP_{X_{<i} \mid \tB_{\le  i}}}.
	\end{align*}
\end{lemma}
Namely, $\set{\tB_i}$ mimics the events $\set{B_i}$, defined over $\dQ$, in (an extension of) $\dP$. It follows that bounding   the smooth KL-divergence of $\dP_X$ and $\dQ_X$ 	(and thus guarantee that small events in $\dP_X$ are small in $\dQ_X$), is reduced to  understanding   $\dP$ and \emph{simplified}  variants of $\dQ$.

\subsection{Main Results}\label{sec:intro:our_result}
We prove the following results (in addition to \cref{lemma:intro:KLtoPR,lemma:intro:smooth-div-main-prop}). The first result, which is the main technical contribution of this paper,  is the following bound on the smooth KL-divergence between a distribution and its many-round skewed variant. 
\begin{theorem}[Smooth KL-divergence for skewed distributions, informal]\label{thm:intro:BoundingSmoothKL}
	Let $\Unf= \Unf_X$ be a distribution over an $m \times n$ matrices with independent columns, and let $W$ and  $\cE=\set{E_{i,j}}$ be  events over $\Unf$.  Let $\Idl = \Unf|W$ and let $\Rll = \Rll(\Unf,W,\cE)$ be the skewed variant of $\Idl$ defined in \cref{eq:intro:Q}. Assume   $\forall (i,j) \in [m]\times [n]$: (1)  $E_{i,j}$ is determined by $X^j$ and (2) There exists $\delta_{i,j} \in (0,1]$ such that $\Unf[E_{i,j} | X_{\leq i, j}] = \delta_{i,j}$  for any fixing of $X_{\leq i, j}$. Then  (ignoring constant factors, and under some restrictions on $n$ and $\Unf[W]$) 
	\begin{align*}
		D^{\eps m + 1/\delta n}(\Idl_X || \Rll_X) \leq  \eps m + m/\delta n
	\end{align*}
	for $\delta =  \min_{i,j}\set{\delta_{i,j}}$ and $\eps = \log(\frac1{\Unf[W]})/\delta n$.
	In a special case where $E_{i,j}$ is determined by  $X_{\leq i+1,j}$, it holds that
	\begin{align*}
		D^{\eps + 1/\delta n}(\Idl_X || \Rll_X) \leq \eps + m/\delta n .
	\end{align*}
\end{theorem}

Combining \cref{lemma:intro:KLtoPR} and \cref{thm:intro:BoundingSmoothKL} yields the following bound on parallel repetition of partially simulatable arguments.  We give separate bounds for partially simulatable argument  and for \textit{partially prefix-simulatable arguments}: a $\delta$-simulatable argument is  $\delta$-prefix-simulatable if for any $i$-round view, the event $E$ guaranteed by the  simulatable property for this view is determined by the  coins used in the first $i+1$ rounds. It is clear that the random-termination variant of an $m$-round argument  is $1/m$-prefix-simulatable.

\begin{theorem}[Parallel repetition for partially simulatable arguments, informal]\label{thm:intro:PR}
	Let $\pi$ be an $m$-round $\delta$-simulatable interactive argument with soundness error
	$1 - \eps$, and let $n\in \N$. Then $\pi^n$ has soundness error $(1-\eps)^{\delta n/m}$. Furthermore,  if $\pi$ is $\delta$-prefix-simulatable,  then $\pi^n$ has soundness error $(1-\eps)^{\delta n}$.\footnote{Throughout, we assume that the protocol transcript contains the verifier's Accept/Reject decision (which is \wlg for random-terminating variants). We deffer the more general case for the next version.}
\end{theorem}
A subtlety that arises when  proving \cref{thm:intro:PR} is that  a direct composition of    \cref{lemma:intro:KLtoPR} and \cref{thm:intro:BoundingSmoothKL} only yields the desired result when the number of repetitions  $n$ is ``sufficiently'' large compared to the number of rounds $m$ (roughly, this is because we need the additive term $m/\delta n$ in  \cref{thm:intro:BoundingSmoothKL} to be smaller than  $\eps$).  We bridge this gap by presenting a sort of upward-self reduction from  a few repetitions to  many repetitions. The idea underlying this reduction is rather general and applies to other proofs of this type, and in particular to those of \cite{HastadPWP10,HaitnerPR13,ChungL10}.\footnote{Upward-self reductions trivially exist for interactive proof: assume the existence of a cheating prover $\nsP$ breaking  the $\alpha$ soundness error of $\pi^n$, then $(\nsP)^\ell$,  \ie the prover using $\nsP$in parallel  for $\ell$ times, violates the assumed $\alpha^\ell$ soundness error of $\pi^{n\ell}$. However, when considering interactive arguments, for which we cannot  guarantee a soundness error below negligible (see \cref{fn:1}), this approach breaks down  when  $\alpha^\ell$ is negligible.}

We complete the picture by showing that an  $\delta$ factor in the exponent in \cref{thm:intro:PR} is unavoidable.
\begin{theorem}[lower bound, informal]\label{thm:intro:LowerBound}
	Under suitable cryptographic assumptions, for any $n,m \in \N$ and $\eps \in [0,1]$, there exists an $m$-round $\delta$-prefix-simulatable  interactive argument $\pi$ with soundness error $1-\eps$, such that  $\pi^n$ has soundness error at least  $(1-\eps)^{\delta n}$. Furthermore, protocol $\pi$ is a random-terminating variant of an interactive argument.
\end{theorem}

It follows that  our bound for  partially prefix-simulatable  arguments and  random-termination variants, given in \cref{thm:intro:PR},  is tight.

\subsubsection{Proving \cref{thm:intro:BoundingSmoothKL}}\label{sec:intro:Tech}
We highlight some details about the proof of \cref{thm:intro:BoundingSmoothKL}.  Using  \cref{lemma:intro:smooth-div-main-prop}, we prove the theorem by showing  that the following holds  for a  carefully chosen events  $\set{\eB_i}$ over $\Rll_{X,J}$:
\begin{itemize}
	\item $ \sum_{i=1}^m D\paren{\Idl_{X_i} || \Rll_{X_i \mid \eB_{\leq i}} \mid \Idl_{X_{<i} \mid \tB _{\le  i}}}$ is small, and  
	
	\item $\Idl[ \tB_{\le m}]$ is large,
\end{itemize}
where   $\set{\tB_i}$ are events over (extension of) $\Idl$, with  $\tB_i$ taking the value $1$ with probability  $\Rll[\eB_i \mid \eB_{<i}, X_{<i}]$. We chose the events $\set{\eB_i}$ so that we have the following guarantees  on   $\Rll_{X_{i},J \mid \eB_{\leq i}, X_{<i}}$: 
\begin{enumerate}
	\item    $J|X_{< i}$ has high entropy (like it has  without any conditioning), and 
	
	\item $\Unf[W \mid X_{<i},X_{i,J} ,E_{i,J}] \geq \Unf[W| X_{<i}]/2$.
\end{enumerate}
Very roughly, these guarantees make the task of bounding the required KL-divergence much simpler since they guarantee that the skewing induced by $\Rll$
does not divert it too much (compared to $\Idl$). The remaining challenge is therefore lower-bounding $\Idl[ \tB_{\le m}]$.  We bound the latter distribution by associating  a martingale sequence  with the  distribution $\Ideal$. In order to bound this sequence,  we prove a new concentration bound for ``slowly evolving''  martingale  sequences, \cref{lemma:prelim:Martingales-new-bound}, that we believe  to  be of independent interest.

\subsection{Related Work}\label{sec:intro:related_work}
\subsubsection{Interactive Arguments}

\paragraph{Positive results.}
\citet{BellareIN97} proved that the parallel repetition of three-message interactive arguments reduces the soundness error at an exponential, but not optimal, rate. \citet{CanettiHS05} later showed that parallel repetition does achieve an optimal exponential decay in the soundness error for such arguments. \citet{PassV12} have proved the same for constant-round public-coin arguments. For public-coin arguments of any (polynomial) round complexity, \citet{HastadPWP10} were the first to show that parallel repetition reduces the soundness error exponentially, but not at an optimal rate. The first optimal analysis of parallel repetition in public-coin arguments was that of \citet{ChungL10}, who showed that the soundness error of the $k$ repetitions improves to $(1-\eps)^k$. \citet{ChungP15} proved the same bound using KL-divergence.  For non-public coin argument (of  any round complexity), \citet{HaitnerPR13} introduced the random-terminating variant of a protocol, and proved that the parallel repetition of these variants improves the soundness error at a weak exponential rate. \citet{HastadPWP10} proved the same, with essentially the same parameters, for partially-simulatable arguments, that contain  random-terminating protocols as a special case.  All the above results extend to ``threshold verifiers'' where the parallel repetition is considered accepting if the number of accepting verifiers is above a certain threshold. Our result rather easily extends to such verifiers, but we defer the tedious details to the next version.  \citet{ChungP11} proved that full independence of the parallel executions is not necessary to improve the soundness of public-coin arguments, and that the verifier can save randomness by carefully correlating the different executions. It is unknown whether similar savings in randomness can be achieved for
random-terminating arguments. Finally, the only  known  round-preserving alternative to the random-terminating transformation is the elegant approach of  \citet{ChungL10},   who showed that a fully-homomorphic encryption (FHE) can be used to compile any interactive argument to a one (with the same soundness error) for which parallel repetition improves the soundness error at ideal rate, \ie $(1-\eps)^n$.  However, in addition to being conditional (and currently it  is only known how to construct FHE assuming hardness of  {learning with errors}  \cite{brakerski2014efficient}), the compiled protocol might lack some of the guarantees of the original protocol (\eg fairness). Furthermore, the reduction is \emph{non} black box (the  parties homomorphically evaluate  \emph{each} of the protocol's gates), making the  resulting protocol  highly impractical, and preventing the use of this approach when only black-box access is available (\eg the weak protocol is given as a DLL or implemented in   hardware).

\paragraph{Negative results.}
\citet{BellareIN97} presented for any $n\in \N$, a four-message interactive argument of soundness error
$1/2$, whose $n$-parallel repetition soundness remains $1/2$. \citet{PietrzakW12} ruled out the possibility that enough repetitions will eventually improve the soundness of an interactive argument. They presented  a \emph{single} $8$-message argument for which the above phenomenon holds for all polynomial $n$ simultaneously. Both results hold under common cryptographic assumptions.

\subsubsection{Two-Prover Interactive Proofs}
The techniques used in analyzing parallel-repetition of interactive arguments are closely related to those for analyzing parallel repetition of two-prover one-round games. Briefly, in such a game, two unbounded \emph{isolated} provers try to convince a  verifier in the validity of a statement. Given a game of soundness error $(1-\eps)$, one might expect the soundness error of its $n$ parallel repetition to be $(1-\eps)^n$, but as in the case of  interactive arguments, this turned out to be false \cite{Feige91,FeigeV02,FortnowRS90}. Nonetheless, \citet{Raz98} showed that parallel repetition does achieve an exponential decay for any two-prover one-round game, and in particular reduces the soundness error to $(1-\eps)^{\eps^{O(1)} n/s}$, where $s$ is the provers' answer length. These  parameters were later improved by  \citet{Holenstein09}, and improved further for certain types of games  by \citet{Rao11,DinurS14,Moshkovitz14}.   The core challenge   in the analysis of  parallel repetition of  interactive arguments and of multi-prover one-round games is very similar: how to simulate a random accepting execution of the proof/game   given the verifier messages. In interactive arguments, this is difficult since the prover lacks computational power. In multi-prover one-round games, the issue is that the different provers cannot communicate. 

\subsection*{Open Questions}
While our bound for the parallel repetition of  partially prefix-simulatable  arguments is tight, this question for (non prefix) partially simulatable  arguments is still open (there is a $1/m$ gap in the exponent). A more important  challenge is to develop  a better (unconditional) round-preserving amplification technique for  arbitrary interactive arguments (which cannot be  via random termination), or alternatively to prove that  such an amplification does not exist.

\subsection*{Paper Organization}
Basic notations, definitions and tools used throughout the paper are stated and proved  in \cref{sec:prelim}. The definition of smooth KL-divergence  and  some properties of this measure are given in \cref{sec:smoothKL}. The definition of many-round skewed distributions and our main bound for such distributions  are given in  \cref{sec:SkewedDistributions}. The aforementioned bound is proven in \cref{sec:BoundingSmoothKL}, and is used in \cref{sec:PR} for proving   our bound on the parallel repetition of partially simulatable arguments. The matching lower bound on such parallel repetition, along with an intuitive explanation of why random-termination helps to beat \cite{BellareIN97}'s counterexample, is given in  \cref{sec:LowerBound}. Missing proofs can be found in \cref{sec:missinProofs}.

\section{Preliminaries}\label{sec:prelim}

\subsection{Notation}\label{sec:prelim:notation}
We use calligraphic letters to denote sets, uppercase for random variables, and  lowercase for values and functions. All logarithms considered here are natural logarithms (\ie in base $e$). For $n\in \N$, let $[n] \eqdef \set{1,\ldots,n}$ and $(n) \eqdef \set{0,\ldots,n}$.  Given a vector $v\in \Sigma^m$, we let $v_i$ denote its \ith entry, and for ordered $\cs = (s_1,\ldots,s_k) \subseteq [n]$ let $c_\cs = (v_{s_1},\ldots,v_{s_k})$. In particular, $v_{< i} = v_{1,\ldots,i-1}$ and   $v_{\le i} = v_{1,\ldots,i}$. For $v\in \zn$, let $1_v = \set{i \in [n] \colon v_i = 1}$. For $m\times n$ matrix $x$, let $x_i$ and $x^j$ denote their \ith row and \jth column respectively, and defined $x_{< i}$, $x_{\le i}$,  $x^{< j}$ and  $x^{\le j}$ respectively.  Given a Boolean statement $S$ (\eg $X \geq 5$),  let $\1_{S}$ be the indicator function that outputs $1$ if $S$ is a true statement and $0$ otherwise.
For $a\in \R$ and $b\geq 0$,  let $a\pm b$ stand for the interval $[a-b,a+b]$.


Let $\poly$ denote the set of all polynomials,  \ppt denote  for probabilistic  polynomial time, and  \pptm denote a \ppt algorithm (Turing machine).  A function $\nu \colon \N \to [0,1]$ is \textit{negligible}, denoted $\nu(n) = \negl(n)$, if $\nu(n)<1/p(n)$ for every $p\in\poly$ and large enough $n$. Function  $\nu$ is \textit{noticeable}, denoted $\nu(n) \geq 1/\poly(n)$, if  exists $p\in\poly$ such that $\nu(n) \geq 1/p(n)$ for all $n$.

\subsection{Distributions and Random Variables}\label{sec:prelim:dist}

A discrete random variable $X$ over $\cX$ is sometimes defined by its probability mass function (pmf) $P_X$ ($P$ is an arbitrary symbol). A conditional probability distribution is a function $P_{Y|X}(\cdot|\cdot)$ such
that for any $x\in\cX$, $P_{Y|X}(\cdot|x)$ is a pmf over $\mathcal{Y}$.  The joint pmf $P_{XY}$ can be written the product $P_XP_{Y|X}$, where $(P_XP_{Y|X})(x,y)=P_X(x)P_{Y|X}(y|x)=P_{XY}(xy)$. The marginal pmf $P_Y$ can be written as the composition $P_{Y|X}\circ P_{X}$, where $(P_{Y|X}\circ P_{X})(y)=\sum_{x\in\cX}P_{Y|X}(y|x)P_X(x)=P_Y(y)$. We sometimes write $P_{\cdot,Y}$ to denote a pmf $P_{X,Y}$ for which we do not care about the random variable $X$.  We denote by $P_X[W]$ the probability that an event $W$ over $P_X$ occurs, and given a set $\cs \subseteq \cx$ we define $P_X(\cs) = P_X[X \in \cs]$.  Distribution $P'_{XY}$ is an extension of $P_X$ if $P'_X \equiv P_X$.  Random variables and events defined over $P_X$ are  defined over the extension  $P'_{XY}$ by ignoring the value of $Y$. We sometimes abuse notation and say that $P_{XY}$ is an extension of $P_X$. 

The support of a distribution $P$ over a finite set $\cx$, denoted $\Supp(P)$, is defined as $\set{x\in \cx: P(x)>0}$. The \emph{statistical distance} of two distributions $P$ and $Q$ over a finite set $\cx$, denoted as $\SD(P,Q)$, is defined as  $\max_{\cs\subseteq \cx} \size{P(\cs)-Q(\cs)} = \frac{1}{2} \sum_{x\in   \cs}\size{P(x)-Q(x)}$.  Given a set $\cs$, let $U_{\cs}$ denote the uniform distribution over the elements of $\cs$. We sometimes write $x \sim \cs$ or $x \la \cs$, meaning that $x$ is uniformly drawn  from $\cs$. For $p \in [0,1]$,  let  $\Bern(p)$ be the Bernoulli distribution over $\zo$, taking the value $1$ with probability $p$. 


\subsection{KL-Divergence}\label{sec:prelim:facts_kl_diverg}

\begin{definition}\label{def:diver}
  The \emph{\textsf{KL-divergence}} (\aka Kullback-Leibler divergence and relative
  entropy) between two distributions $P,Q$ on a discrete alphabet $\cX$
  is
  \begin{align*}
    D(P||Q) = \sum_{x\in\cX}P(x)\log\frac{P(x)}{Q(x)} = \Ex_{x \sim P}\log\frac{P(x)}{Q(x)},
  \end{align*}
  where $0\cdot\log\frac00 = 0$ and if $\text{ }\exists x\in\cX$ such that
  $P(x)>0=Q(x)$ then $D(P||Q)=\infty$.
\end{definition}

\begin{definition}
  \label{def:cond-diver}
  Let $P_{XY}$ and $Q_{XY}$ be two probability distributions over
  $\cX\times\mathcal{Y}$. The \emph{\textsf{conditional divergence}}
  between $P_{Y|X}$ and $Q_{Y|X}$ is
  \begin{align*}
    D(P_{Y|X}||Q_{Y|X}|P_X) = \Ex_{x\sim P_X}[D(P_{Y|X=x}||Q_{Y|X=x})] = \sum_{x\in\cX}P_X(x)D(P_{Y|X=x}||Q_{Y|X=x}).
  \end{align*}
\end{definition}

\begin{fact}[Properties of divergence]\label{fact:prelim:diver-properties}
  $P_{XY}$ and $Q_{XY}$ be two probability distributions over
  $\cX\times\mathcal{Y}$. It holds that:
  \begin{enumerate}
  \item\label{fact:diver-properties:item:non-negative} (Information inequality) $D(P_X||Q_X)\geq 0$,
    with equality holds iff $P_X=Q_X$.
  \item\label{fact:diver-properties:item:monotone} (Monotonicity) $D(P_{XY}||Q_{XY})\geq D(P_Y||Q_Y)$.
  \item\label{fact:diver-properties:item:chain-rule} (Chain rule) $D(P_{X_1\cdots X_n}||Q_{X_1\cdots X_n}) = \sum_{i=1}^n D(P_{X_i|X_{<i}}||Q_{X_i|X_{<i}}|P_{X_{<i}})$.
    
   \noindent If $Q_{X_1\cdots X_n} = \prod_{i=1}^n Q_{X_i}$ then
    \begin{align*}
      D(P_{X_1\cdots X_n}||Q_{X_1\cdots X_n}) = D(P_{X_1\cdots
      X_n}||P_{X_1}P_{X_2}\cdots P_{X_n}) + \sum_{i=1}^n D(P_{X_i}||Q_{X_i}).
    \end{align*}

  \item\label{fact:diver-properties:item:condition-diver} (Conditioning increases divergence) If
  $Q_Y=Q_{Y|X}\circ P_X$ (and $P_{Y} = P_{Y|X}\circ P_{X}$), then $D(P_Y||Q_Y) \leq D(P_{Y|X}||Q_{Y|X}|P_X)$.

  \item\label{fact:diver-properties:item:data-processing} (Data-processing) If $Q_Y = P_{Y|X}\circ Q_{X}$ (and
    $P_{Y} = P_{Y|X}\circ P_{X}$), it holds that $D(P_Y||Q_Y) \leq D(P_X||Q_X)$.
  \remove{\item\label{fact:diver-properties:item:convexity} (Convexity) If $P = \lambda P_1 + (1-\lambda) P_2$ and $Q = \lambda Q_1 + (1-\lambda) Q_2$ for $\lambda \in [0,1]$, then
  \begin{align*}
  	\re{P}{Q} \leq \lambda \re{P_1}{Q_1} + (1-\lambda) \re{P_2}{Q_2}
  \end{align*}}
  \end{enumerate}
\end{fact}

\begin{fact}\label{fact:kl-divergence:diver-cond-on-event}
	Let $X$ be random variable drawn from $P$ and let $W$ be an
	event defined over $P$. Then
	\begin{align*}
		\re{P_{X|W}}{P_X} \leq \log{\frac1{P[W]}}.
	\end{align*}
\end{fact}

\begin{fact}\label{fact:kl-divergence:condition-on-cond-diver}
  Let $X,Y$ be random variables drawn from either $P$ or $Q$ and let $W$ be an
  event defined over $P$. It holds that
  \begin{align*}
    \Ex_{x\sim P_{X|W}}D(P_{Y|X=x}||Q_{Y|X=x})\leq \frac{1}{P[W]}\cdot D(P_{Y|X}||Q_{Y|X}||P_X).
  \end{align*}
\end{fact}
\begin{proof}
  \begin{align*}
    \Ex_{x\sim P_{X|W}}D(P_{Y|X=x}||Q_{Y|X=x})
    &= \sum_xP_{X|W}(x)D(P_{Y|X=x}||Q_{Y|X=x})\\
    &= \sum_x\frac{P[X=x,W]}{P[W]}D(P_{Y|X=x}||Q_{Y|X=x})\\
    &\leq\sum_x\frac{P_X(x)}{P[W]}D(P_{Y|X=x}||Q_{Y|X=x})\\
    &=\frac{1}{P[W]}\cdot D(P_{Y|X}||Q_{Y|X}||P_X),
  \end{align*}
  where the inequality follows since $P[X=x,W]\leq P_X(x)$ and
  $D(\cdot||\cdot)\geq0$.
\end{proof}

\begin{fact}\label{fact:kl-divergence:conditionP}
  Let $X$ be a random variable over $\cX$ drawn form either $P_X$ or
  $Q_X$ and let $\mathcal{S}\subseteq\cX$. It holds that
  \begin{align*}
    D(P_{X|X\in\mathcal{S}}||Q_X)\leq
    \frac{1}{P_X(\mathcal{S})}\cdot\paren{D(P_X||Q_X) + \frac1e + 1}.
  \end{align*}
\end{fact}
\begin{proof}
  If $D(P_X||Q_X)=\infty$, then the statement holds trivially. Assume that
  $D(P_X||Q_X)<\infty$ and compute
  \begin{align*}
    D(P_{X|X\in\mathcal{S}}||Q_X) &= \sum_{x\in\mathcal{S}}P_{X|X\in\mathcal{S}}(x)\log\frac{P_{X|X\in\mathcal{S}}(x)}{Q_X(x)}\\
                                  &=\sum_{x\in\mathcal{S}}\frac{P_{X}(x)}{P_X(\mathcal{S})}\log\frac{P_{X}(x)/P_X(\mathcal{S})}{Q_X(x)}\\
                                  &=\sum_{x\in\mathcal{S}}\frac{P_{X}(x)}{P_X(\mathcal{S})}\log\frac{1}{P_X(\mathcal{S})}
                                    + \sum_{x\in\mathcal{S}}\frac{P_{X}(x)}{P_X(\mathcal{S})}\log\frac{P_{X}(x)}{Q_X(x)}.
  \end{align*}

  To bound the left sum, compute
  \begin{align*}
    \sum_{x\in\mathcal{S}}\frac{P_{X}(x)}{P_X(\mathcal{S})}\log\frac{1}{P_X(\mathcal{S})}
    &\leq \sum_{x\in\mathcal{S}}\frac{P_{X}(x)}{P_X(\mathcal{S})}\cdot\frac{1}{P_X(\mathcal{S})}\\
    &\leq \frac{1}{P_X(\mathcal{S})},
  \end{align*}
  where the first inequality follows since $\log(x)\leq x$ for all $x$.

  To bound the right sum, compute
  \begin{align*}
    \sum_{x\in\mathcal{S}}\frac{P_{X}(x)}{P_X(\mathcal{S})}\log\frac{P_{X}(x)}{Q_X(x)}
    &=\frac{1}{P_X(\mathcal{S})}\paren{\sum_{x\in\mathcal{S}}P_{X}(x)\log\frac{P_{X}(x)}{Q_X(x)}
      +
      \sum_{x\notin\mathcal{S}}P_{X}(x)\log\frac{P_{X}(x)}{Q_X(x)}
      - \sum_{x\notin\mathcal{S}}P_{X}(x)\log\frac{P_{X}(x)}{Q_X(x)}}\\
    &=\frac{1}{P_X(\mathcal{S})}\paren{D(P_X||Q_X) - \sum_{x\notin\mathcal{S}}P_{X}(x)\log\frac{P_{X}(x)}{Q_X(x)}}.
  \end{align*}
  The following calculation completes the proof:
  \begin{align*}
    \sum_{x\notin\mathcal{S}}P_{X}(x)\log\frac{P_{X}(x)}{Q_X(x)}
    &= \sum_{x\notin\mathcal{S}}Q_X(x)\frac{P_{X}(x)}{Q_X(x)}\log\frac{P_{X}(x)}{Q_X(x)}\\
    &\geq\sum_{x\notin\mathcal{S}}Q_X(x)(-e^{-1})\\
    &\geq -e^{-1},
  \end{align*}
  where the first inequlity holds since $x\log(x)\geq -e^{-1}$ for all $x>0$.
\end{proof}

\remove{
\Enote{delete}
\begin{fact}
	Let $\Unf_{XY} = \prod_{i=1}^n \Unf_{X_i Y_i}$, let $W$ be an event over $\Unf$ and let $\Idl = \Unf_{XY|W}$. Let $\eps \in [0,\frac12]$ and let
	$S_i \subseteq \set{x_i \in \Supp(\Unf_{X_i}) \colon \size{\log \frac{\Idl_{X_i}(x_i)}{\Unf_{X_i}(x_i)}} > \epsilon}$.
	Then it holds that
	\begin{align*}
		\sum_{i=1}^n \Unf(S_i) \leq \frac{10}{\eps^2} D(\Idl || \Unf)
	\end{align*}
\end{fact}
\Enote{TODO. Note that if we take $\epsilon = \sqrt{\frac{20}{n} D(\Idl || \Rll)}$ then it holds that $\sum_{i=1}^n \Unf(S_i) \leq n/2$. We want to make sure that $\frac{D(\Idl || \Rll)}{\delta n} \leq \sqrt{\frac{20}{n} D(\Idl || \Rll)}$, i.e.,  $\delta \geq \sqrt{\frac{D(\Idl || \Rll)}{20 n}}$. It does not work!!!}
}


\begin{definition}
	For  $p,q \in [0,1]$ let  $\re{p}{q} \eqdef \re{\Bern(p)}{\Bern(q)}$.
\end{definition}

\begin{fact}[{\cite[Implicit in Corollary 3.2 to 3.4]{Mulzer}}]\label{fact:prelim:bernoulli-div-est}
	For any $p \in [0,1]$ it holds that
	\begin{enumerate}
		\item $\re{(1-\delta)p}{p} \geq \delta^2 p/2 $ for any $\delta \in [0,1]$.\label{fact:prelim:bernoulli-div-est:minus}
		\item $\re{(1+\delta)p}{p} \geq \min\set{\delta,\delta^2} p/4 $ for any $\delta \in [0,\frac1{p}-1]$.\label{fact:prelim:bernoulli-div-est:plus}
	\end{enumerate}
\end{fact}

The proof of the following proposition, which relies on \citet{DonskerV83}'s inequality, is given in \cref{sec:appendix:sub-exp-to-divergence}.

\def\propSubExpToDiv{
	Let $X$ be a random variable drawn form either $P$ or $Q$. Assume that
	$\Pr_P[\abs{X}\leq 1] = 1$ (i.e., if $X$ is drawn from $P$ then $\abs{X}\leq 1$ almost surely) and
	that there exist $\eps,\sigma^2,K_1,K_2>0$ such that $\Pr_Q[\abs{X}\leq
        1] \geq 1-\eps$ and
	\begin{align*}
	\Pr_Q[\abs{X} \geq t] \leq K_2\cdot\exp\paren{-\frac{t^2}{K_1\sigma^2}} \quad\text{for all
		$0\leq t \leq 1$}.
	\end{align*}
	Then, there exists $K_3=K_3(K_1,K_2,\eps)>0$ such that
	\begin{align*}
	\Ex_P[X^2] \leq K_3\cdot\sigma^2\cdot\paren{D(P||Q) + 1}.
	\end{align*}
}

\begin{proposition}\label{prop:prelim:sub-exp-to-divergence}
  \propSubExpToDiv
\end{proposition}

\subsection{Interactive Arguments}\label{sec:prelim:interactive-arg}
\begin{definition}[Interactive arguments]\label{def:IA}
  A \ppt protocol $(\Pc,\Vc)$ is an {\sf interactive argument} for a language
  $\cL\in \NP$ with {\sf completeness $\alpha$} and {\sf soundness error
    $\beta$}, if the following holds:

	\begin{itemize}
		\item $\pr{(\Pc(w),\Vc)(x) = 1} \geq  \alpha(\size{x})$ for any $(x,w)\in R_\cL$.

		\item $\pr{(\Pc^\ast,\Vc)(x) = 1} \le  \max\set{\beta(\ssize{x}),\negl(\ssize{x})}$ for any \ppt $\Pc^\ast$ and large enough  $x \notin\cL$.
	\end{itemize}
We refer to party $\Pc$ as the {\sf prover}, and to $\Vc$ as the {\sf verifier}.
\end{definition}

Soundness against\emph{ non-uniform} provers is analogously  defined, and  all the results in this paper readily extend to this model.

Since in our analysis we  only care about soundness amplification,  in the following we fix  $\cL$ to be  the empty language, and assume the input to the protocol is just  a string of ones,  which we refer to as the \emph{security parameter}, a parameter we omit when cleared from the context.

\subsubsection{Random-Terminating Variant}\label{sec:prelim:RT}

\begin{definition}[Random-terminating variant,  \cite{HaitnerPR13}]\label{def:prelim:RT}
	Let $\V$ be a $m$-round randomized interactive algorithm. The {\sf random-terminating variant of $\V$}, denoted $\tV$, is defined as follows: algorithm $\V$ acts exactly as $\V$ does, but adds the following step at the end of each communication round: it tosses an $(1-1/m,1/m)$ biased coin (\ie $1$ is tossed with probability $1/m$),  if the outcome is one then it outputs $1$ (\ie accept) and halts. Otherwise, it continues as $\V$ would.
	
	For a protocol  $\pi=(\Pc,\Vc)$, the protocol   $\tpi = (\Pc,\tV)$ is referred to as the random-terminating variant of $\pi$. 
\end{definition}

\subsubsection{Partially Simulatable Interactive Arguments}\label{sec:prelim:PS}
\begin{definition}[Partially simulatable protocols,  \cite{HastadPWP10}]\label{def:prelim:PS}
	A randomized interactive algorithm $\Vc$ is  {\sf $\delta$-simulatable}, if there exists an oracle-aided $\Sc$ (simulator) such that the  following holds: for every strategy $\sP$ and a partial view $v$ of $\sP$ in an interaction of $(\sP,\Vc)(1^\kappa)$,  the output of $\Sc^{\sP}(1^\kappa,v)$ is  $\sP$'s view in a random continuation of $(\sP,\Vc)(1^\kappa)$ conditioned on $v$ and $\Delta$, for $\Delta$ being  a $\delta$-dense  subset of the coins of $\Vc$ that are  consistent with $v$.  The running time of $\Sc^{\sP}(1^\kappa,v)$ is polynomial in $\kappa$ and  the running time of $\sP(1^\kappa)$.

	Algorithm $\Vc$ is  {\sf $\delta$-prefix-simulatable} if membership in the guaranteed event $\Delta$ is determined by the coins $\Vc$ uses in the first $\round(v)+1$ rounds.\footnote{$\Delta = \Delta_1 \times \Delta_2$, for $\Delta_1$ being a ($\delta$-dense) subset of the possible values for first $\round(v)+1$ round coins, and $\Delta_2$ is the set of  all possible  values for the coins used in rounds $\round(v)+2,\ldots,m$, for $m$ being the round complexity of $\Vc$.}

	An interactive argument   $(\Pc,\Vc)$ is \sf $\delta$-simulatable/ $\delta$-prefix-simulatable, if $\Vc$ is. 
\end{definition}
It is clear that random termination variant of an $m$-round interactive argument is $1/m$-prefix-simulatable. 

\begin{remark}
One can relax the above definition and allow a different (non-black)  simulator per  $\sP$, and then only require it to exists for poly-time $\sP$. While our proof readily extends to this  relaxation, we prefer to use the above definition  for presentation clarity.
\end{remark}

\subsubsection{Parallel Repetition}\label{sec:prelim:PR}

\begin{definition}[Parallel repetition]\label{def:prelim:PR}
	Let $(\P,\V)$ be an interactive protocol, and let $n \in \N$. We define the {\sf $n$-parallel-repetition} of $(\P,\V)$ to be the protocol $(\P^n,\V^n)$ in which $\P^n$ and $\V^n$ execute $n$ copies of $(\P,\V)$ in parallel, and at the end of the execution, $\V^n$ accepts if all copies accept.
\end{definition}

\paragraph{Black-box  soundness reduction.}
As in most such proofs, our proof for the parallel repetition of partially-simulatable arguments  has the following black-box form.
\begin{definition}[Black-box  reduction for parallel repetition]\label{def:BBProof}
Let $\pi= (\Pc,\Vc)$ be an interactive argument.  An oracle-aided algorithm $\Rc$  is a {\sf black-box reduction for the  $g$-soundness  of the parallel repetition of   $\pi$}, if the following holds for any poly-bounded $n$: let $\kappa\in \N$ and $\nsP$  be deterministic cheating prover   breaking  the  soundness  of $\pi^{n=n(\kappa)}(1^\kappa)$ with probability $\eps' \ge g(n,\eps= \eps(\kappa))$. Then  
	\begin{description}
		\item[Sucesss probability.]  $\Rc = \Rc^{\nsP}(1^\kappa,1^{n})$ breaks the  soundness of $\pi$ with probability at least $1-\eps/3$.

		\item[Running time.] Except with probability $\eps/3$, the running time of $\Rc$  is polynomial in  $\kappa$, the running time of $\nsP(1^\kappa)$ and $1/\eps'$.
	\end{description}
\end{definition}
We use the following fact.
\begin{proposition}\label{prop:BBProof}
Assume there exists a black-box reduction for the  $g$-soundness  of the parallel repetition of  any $\delta$-simulatable [\resp $\delta$-prefix-simulatable] interactive argument, then  for any poly-bounded $n$, the soundness error of the $n$-fold repetition of any such argument  is bounded by $g(n,\eps)$.
\end{proposition}
\begin{proof}
The only non-trivial part is how to handle randomized cheating provers (the above  definition of black-box reduction only considers deterministic provers).  Let $\pi = (\Pc,\Vc)$ be a $\delta$-simulatable interactive argument (the proof for $\delta$-prefix-simulatable arguments follows the same lines).  Let  $\nsP$ be an efficient randomized cheating prover violating the $g(n,\eps)$ soundness error of $\pi$, and let  $r(\kappa)$ be a bound in the number of coins it uses.  Let  $\hV$ be the variant of $\Vc$ that appends    $r(\kappa)$ uniform coins to its first message. It is clear that if $\hV$  is also $\delta$-simulatable. Consider  the deterministic cheating prover $\nsP'$ that attack $\nhV$  by acting as $\nsP$ whose random coins set to the randomness appended to the first message of the first verifier. It is clear that  $\nsP'$ success probability (when attacking $\nhV$) equals  that of $\nsP$ (when attacking $\nV$).  Hence, the existence of a black-box reduction for the $g$-soundness of  $(\Pc,\hV)^n$, yields an efficient attacker $\sP'$ breaking the   $(1-\eps)$ soundness of $(\Pc,\hV)$. This  attacker can be easily modified to create an  efficient attacker breaking the   $(1-\eps)$ soundness of $\pi$.
\end{proof}

\subsection{Martingales}\label{sec:prelim:facts_concent_bounds:martingales}

\begin{definition}\label{def:prelim:martingales}
	A sequence of random variables $Y_0,Y_1,\ldots,Y_n$ is called a \textbf{martingale sequence with respect to} a sequence $X_0,X_1,\ldots,X_n$, if for all $i \in [n]$: (1) $Y_i$ is a deterministic function of $X_0,\ldots, X_i$, and (2) $\ex{Y_{i} \mid X_0, \ldots, X_{i-1}} = Y_{i-1}$.
\end{definition}

\remove{
\Enote{Should we write the following? : The most famous concentration bound on Martingales is Azuma's inequality which achieves .... However, in some settings we only guaranteed a weaker restrictions about the sequences. The following fact, which extends Chebyshev's inequality, presents a weak ...}
}

The following lemma (proven in \cref{sec:appendix:Martingales-new-bound}) is a new concentration bound on ``slowly evolving'' martingales.

\def\MartingalesLemma{
	Let $Y_0 = 1, Y_1,\ldots, Y_n$ be a martingale w.r.t $X_0,X_1,\ldots,X_n$ and assume that $Y_i \geq 0$ for  all $i\in [n]$. Then for every $\lambda \in (0,\frac14]$ it holds that
	\begin{align*}
	\pr{\exists i \in [n]\text{ s.t. } \size{Y_i - 1} \geq \lambda} \leq \frac{23\cdot \ex{\sum_{i=1}^n \min\set{\size{R_i}, R_i^2}}}{\lambda^2}
	\end{align*}
	for $R_i = \frac{Y_i}{Y_{i-1}} - 1$, letting $R_i = 0$ in case $Y_{i-1} = Y_i = 0$.
}
\begin{lemma}[A bound on slowly evolving martingales]\label{lemma:prelim:Martingales-new-bound}
	\MartingalesLemma
\end{lemma}

That is, if $Y_i$ is  unlikely to be far from $Y_{i-1}$ in a multiplicative manner, then the sequence is unlikely to get far from $1$. We use the following corollary of \cref{lemma:prelim:Martingales-new-bound} (proven in \cref{sec:prelim:martingale-specific-bound}).

\def\MartingalesProp{
	Let $Y_0 = 1, Y_1,\ldots, Y_n$ be a martingale w.r.t $X_0,X_1,\ldots,X_n$ where $Y_i \geq 0$ for all $i \in [n]$. Let $Z_1,\ldots, Z_n$ and $T_1,\ldots,T_n$ be sequences of random variables satisfying for all $i \in [n]$: (1) $Y_i = Y_{i-1}\cdot \paren{1 + Z_i}/\paren{1 + T_i}$, and (2) $T_i$ is a deterministic function of $X_0,X_1,\ldots,X_{i-1}$. Then
	\begin{align*}
	\pr{\exists i \in [n]\text{ s.t. } \size{Y_i - 1} \geq \lambda} \leq \frac{150\cdot \ex{\sum_{i=1}^n \paren{\min\set{\size{Z_i}, Z_i^2} + \min\set{\size{T_i}, T_i^2}}}}{\lambda^2}
	\end{align*}
}

\begin{proposition}\label{prop:prelim:martingale-specific-bound}
	\MartingalesProp
\end{proposition}

\subsection{Additional Fact and Concentration Bounds}\label{sec:prelim:facts_concent_bounds}
We use the following fact.

\begin{fact}[\cite{HaitnerPR13}, Proposition 2.5]\label{fact:Prelim:smooth-sampling}
	Let $P_{X_1,\ldots,X_m}$ be a distribution and let $W$ be an event over $P$. Then for every $i \in [m]$ it holds that $\eex{x_{<i} \sim P_{X_{<i} \mid W}}{1/P[W \mid X_{<i}=x_{<i}]} = 1/P[W]$.
\end{fact}

\subsubsection{Sum of Independent Random Variables}\label{sec:prelim:facts_concent_bounds:indep}

\begin{fact}[Hoeffding's inequality]\label{fact:prelim:indep_concent:hoeffding}
	Let $X=X_1+\cdots+X_n$ be the sum of independent random variables such that
	$X_i\in[a_i,b_i]$. Then for all $t\geq 0$:
	
	\begin{enumerate}
		\item $\Pr[X-\Ex[X] \geq t] \leq \exp\paren{-\frac{2t^2}{\sum_{i=1}^{n}(b_i-a_i)^2}}$.
		
		\item $\Pr[\abs{X-\Ex[X]} \geq t] \leq 2\exp\paren{-\frac{2t^2}{\sum_{i=1}^{n}(b_i-a_i)^2}}$.
	\end{enumerate}
\end{fact}

\remove{
	\begin{fact}[{\cite[Theorem 5.3]{AminCO}}]\label{fact:prelim:indep_concent:binom}
		Let  $X \sim \Bin(n,p)$, then for all $t \geq 0$:
		\begin{enumerate}
			\item $\pr{X \geq \ex{X} + t} \leq \exp\paren{-\frac{t^2}{2\left(np + \frac{t}3\right)}}$.
			\item $\pr{X \leq \ex{X} - t} \leq \exp\paren{-\frac{t^2}{2np}}$.
		\end{enumerate}
	\end{fact}
}

\begin{fact}[{\cite[Lemma 2.1]{ChungLu}}]\label{fact:prelim:indep_concent:variance}
	Let $X_1,\ldots,X_n$ be independent random variables such that $X_i\sim
	\Bern(p_i)$. Let $X=\sum_{i=1}^nb_iX_i$ with $b_i>0$, and  let  $v=\sum_{i=1}^nb_i^2p_i$. Then for
	all $t\geq 0$:
	\begin{align*}
	\Pr[\abs{X-\Ex[X]} \geq t] \leq 2\exp\paren{-\frac{t^2}{2(v+bt/3)}}
	\end{align*}
	for $b=\max\set{b_1,b_2,\ldots,b_n}$.
\end{fact}

\def\LiYiProposition{
	Let $L_1,\ldots,L_n$ be independent random variables over $\mathbb{R}$ with $\size{L_i} \leq \ell$ for all $i \in [n]$ and let $Z_i = \paren{L_i/p_i}\cdot \Bern(p_i)$ with $p_i > 0$ for all $i \in [n]$. Let $L = \sum_{i=1}^n L_i$, let $Z = \sum_{i=1}^n Z_i$, let $\mu = \ex{L}$ and let $p = \min_{i \in [n]}\set{p_i}$. Finally, let $\Gamma = Z/\mu - 1$. Then for any $\gamma \in [0,1]$ it holds that
	\begin{align*}
	\pr{\size{\Gamma} \geq \gamma} \leq 4\exp\paren{-\frac{p \mu^2 \gamma^2}{5 \ell^2 n}}
	\end{align*}
}

We use the following fact.
\begin{fact}\label{fact:prelim:L_i_and_Y_i}
	\LiYiProposition
\end{fact}
\begin{proof}
	Note that
	\begin{align}\label{eq:prelim:calc_under_product}
	\Pr[\size{\Gamma} \geq \gamma]
	&= \Pr[\size{Z-\mu} \geq \mu \gamma]\nonumber\\
	&\leq \Pr[\size{Z-L} \geq \mu \gamma/2] + \Pr[\size{L - \mu} \geq \mu \gamma/2]
	\end{align}
	
	We bound each term in \cref{eq:prelim:calc_under_product} separately. For the right-hand side term, we use Hoeffding's inequality (\cref{fact:prelim:indep_concent:hoeffding}) to get
	
	\begin{align}\label{eq:prelim:A-minus-aval}
	\Pr[\size{L - \mu} \geq \mu \gamma/2]
	\leq 2 \exp\paren{-\frac{2(\mu \gamma/2)^2}{\ell^2\cdot n}}
	\leq 2 \exp\paren{-\frac{\mu^2 \gamma^2}{\ell^2 n}},
	\end{align}
	
	We now focus on bounding the left-hand side term. The following holds for any fixing of $L_1,\ldots,L_n$.
	Since $p_i > 0$ for all $i \in [n]$, it holds that $\ex{Z_i} = L_i \implies \ex{Z} = L$. Moreover, the $Z_i$'s are independent random variables such that $Z_i = b_i \cdot \Bern(p_i)$ for $b_i = L_i/p_i$, where $b = \max\set{b_1,\ldots,b_n} \leq \ell/p$ and $v = \sum_{i=1}^n b_i^2 p_i \leq \ell^2 n/p$. 
	\cref{fact:prelim:indep_concent:variance} yields that
	\begin{align}\label{eq:prelim:Z-minus-A}
	\Pr[\size{Z-L} \geq \mu \gamma/2]
	&\leq 2\exp\paren{-\frac{(\mu \gamma/2)^2}{2(v + b \mu \gamma/6)}}
	\leq 2\exp\paren{-\frac{\mu^2 \gamma^2}{4(\ell^2 n/p + \ell \mu \gamma/6p)}}\nonumber\\
	&\leq 2\exp\paren{-\frac{p \mu^2 \gamma^2}{5 \ell^2 n}},
	\end{align}
	where the last inequality holds since $\mu \leq \ell n$ and $\gamma \leq 1$. The proof follows by \cref{eq:prelim:calc_under_product,eq:prelim:A-minus-aval,eq:prelim:Z-minus-A}.
\end{proof}

\section{Smooth KL-Divergence}\label{sec:smoothKL}
In this section we formally define the notion of smooth KL-divergence, state some basic properties of this measure in \cref{sec:smoothKL:Basic}, and develop a tool to help bounding it in \cref{sec:smoothKL:Bound}.

\begin{definition}[\emph{$\alpha$-smooth} divergence]\label{def:smoothdiver}
	Let $P$ and $Q$ be two distributions over a universe $\cU$ and let $\alpha \in [0,1]$.
	The {\sf $\alpha$-smooth divergence of $P$ and $Q$}, denoted  $\diver^{\alpha}(P || Q)$, is defined as   $\inf_{(F_P,F_Q)\in \F} \set{\re{F_P(P)}{F_Q(Q)}}$, for $\F$ being the set of randomized functions pairs such that for every  $(F_P,F_Q)\in \F$:
	\begin{enumerate}
		\item $\Pr_{x\sim P}[ F_P(x) \neq x]\leq \alpha$, where the probability is also over the coins of $F_P$.
		
		\item $\forall x\in\Uni$:\  $\Supp(F_P(x)) \cap \Uni \subseteq{\set{x}}$ and $\Supp(F_Q(x)) \cap \Uni \subseteq{\set{x}}$.
	\end{enumerate}
\end{definition}

\begin{remark}[comparison to H-Technique]
	At least syntactically, the above notion of 	smooth KL-divergence is similar to the distance measure used by the \textit{(coefficients) H-Technique} tool, introduced by \citet{Patarin90}, for upper-bounding \emph{statistical distance}.  Consider	the following alternative definition of statistical distance: $\SD(A,B) = \Ex_{x \sim A} \max\set{0,1 - \frac{B(x)}{A(x)}}$. The H-Technique 	approach considers a smooth variant of the above formulation: small events	\wrt $A$ are ignored. However, while smooth KL-divergence is useful in	settings when the actual KL-divergence might be \emph{unbounded}, as in our settings, the above smooth variant of statistical distance is always very	close to the actual statistical distance, and as such, it is more of a tool for bounding statistical distance than a measure of interest for its own sake.
\end{remark}

\subsection{Basic Properties}\label{sec:smoothKL:Basic}
The following proposition (proven in \cref{sec:appendix:smooth-div}) states  that  small smooth KL-divergence  guarantees that small events \wrt the left-hand-side distribution  are also small  \wrt the right-hand-side  distribution.

\def\PropSmallToSmallEvents{
	Let $P$ and $Q$ be two distributions over $\cU$ with $\diver^{\alpha}(P || Q) < \beta$. Then for every event  $E$ over $\Uni$, it holds that  $Q[E] < 2\cdot \max\set{\alpha + P[E], 4\beta}$.
}

\begin{proposition}\label{prop:prelim:smooth-div}
	\PropSmallToSmallEvents
\end{proposition}

\def\DataProcessSmoothKL{
	Let $P$ and $Q$ be two distributions over a universe $\cU$, let $\alpha \in
	[0,1]$ and let $H$ be a randomized function over $\cU$. Then $ \diver^{\alpha}(H(P) || H(Q)) \le \diver^{\alpha}(P
	|| Q)$.
}

Like any useful distribution measure, smooth KL-divergence posses a data-processing property. The following proposition is proven in \cref{sec:appendix:prop:prelim:smooth-DP}.
\begin{proposition}[Data processing of smooth KL-divergence]\label{prop:prelim:smooth-DP}
	\DataProcessSmoothKL
\end{proposition}

\subsection{Bounding Smooth KL-Divergence}\label{sec:smoothKL:Bound}
\newcommand{\hP}{P}
\newcommand{\hQ}{Q}
\newcommand{\cut}{{\mathsf{cut}}}
\newcommand{\cutP}{P^\cut}
\newcommand{\cutQ}{Q^\cut}

The following lemma allow us to bound the smooth KL-divergence between $\dP$ and $\dQ$, while only analyzing simpler variants of  $\dQ$.
		
\begin{lemma}[Bounding smooth KL-Divergence, restatement of
\cref{lemma:intro:smooth-div-main-prop}]\label{lemma:smooth-div-main-prop}

	Let $P$ and $Q$ be distributions with $P_{X}$ and $Q_{X}$ being   over universe $\cU^{m}$,  and let  $A_1,\ldots,A_m$ and  $B_1,\ldots,B_m$ be  two sets of events over $P$ and $Q$ respectively.  Let $P_{\cdot,XY}$ be an extension of $P = P_{\cdot,X}$ defined by $P_{Y|\cdot,X} = \prod_i P_{Y_i| X}$ for  $P_{Y_i|X}  =  \Bern\paren{P[A_i \mid  X,A_{<i}]\cdot Q[B_i \mid X_{<i}, B_{<i}]}$, letting $P_{Y_i \mid  X} = 0$ if $P[A_{<i} \mid X] = 0$ or $Q[B_{<i} \mid X_{<i}] = 0$, and let $C_i = \set{Y_i=1}$. 
	Then\footnote{Note that \cref{lemma:intro:smooth-div-main-prop} is a special case of \cref{lemma:smooth-div-main-prop} that holds when choosing $A_1,\ldots,A_m$ with $P[A_{\leq m}]=1$.}
	\begin{align*}
		D^{1 - P[C_{\le m}]}(P_X || Q_X) \leq \sum_{i=1}^m D(P_{X_i \mid A_{\leq i}} || Q_{X_i \mid B_{\leq i}} \mid P_{X_{<i} \mid C_{\leq i}}).
	\end{align*}
\end{lemma}
\begin{proof}
	Let $Q_{\cdot, XY}$ be an extension of $Q = Q_{\cdot,X}$ defined by $Q_{Y|\cdot,X} = \prod_i Q_{Y_i|X}$ for  $Q_{Y_i|X}  =  \Bern\paren{P[A_i \mid  X_{<i},A_{<i}] \cdot Q[B_i \mid X,B_{<i}]}$, letting $Q_{Y_i \mid X} = 0$ if $P[A_{<i} \mid X_{<i}] = 0$ or $Q[B_{<i} \mid X] = 0$. Our goal is to show that
	\begin{align}\label{eq:goal-under-exten}
		D^{1 - P[C_{\le m}]}(P_{Y_1, X_1, \ldots, Y_m, X_m} || Q_{Y_1, X_1, \ldots, Y_m, X_m}) \leq \sum_{i=1}^m D(P_{X_i \mid A_{\leq i}} || Q_{X_i \mid B_{\leq i}} \mid P_{X_{<i} \mid C_{\leq i}})
	\end{align}
	The proof then follows by data processing of smooth KL-divergence (\cref{prop:prelim:smooth-DP}). 
\remove{
	By chain rule of KL-divergence, 
	\begin{align}\label{eq:goal-under-exten:1}
	\lefteqn{D(P_{Y_1, X_1, \ldots, Y_m, X_m} || Q_{Y_1, X_1, \ldots, Y_m, X_m})}\\
	&= \sum_{i=1}^m \eex{P_{X_{<i}Y_{<i}}}{D(P_{X_iY_i|X_{<i},Y_{<i}} || Q_{X_iY_i|X_{<i},Y_{<i}})}\nonumber \\
	&= \sum_{i=1}^m \eex{P_{X_{<i}Y_{<i}}}{D(P_{Y_i|X_{<i},Y_{<i}} || Q_{Y_i|X_{<i},Y_{<i}})} 
	+ \sum_{i=1}^m \eex{P_{X_{<i}Y_{\le i}}}{D(P_{X_i|X_{<i},Y_{\le i}} || Q_{X_i|X_{<i},Y_{\le i}})}
	\nonumber.
	\end{align}
	
	We show how to bound the terms in above sum assuming the $Y_i$'s are all ones, and  then obtain the desired bound on the smooth KL-divergence, \cref{eq:goal-under-exten}, by showing that we can ignore the case that the $Y_i$'s are not such. 
}	
	 By definition, for  any $i \in [m]$:
	\begin{align}\label{cond:P_<i}
		P_{X_{<i} \mid Y_{\leq i} = 1^i} \equiv P_{X_{<i} \mid C_{\leq i}}
	\end{align}
	and for any fixing of $x_{<i} \in \Supp(\hP_{X_{<i} \mid Y_{\leq i} = 1^{i}})$:
	\begin{align}
	P_{X_i \mid Y_{\leq i} = 1^i, X_{<i}=x_{<i}} &\equiv P_{X_i \mid X_{<i}, A_{\leq i}}\label{cond:P_i}\\
	Q_{X_i \mid Y_{\leq i} = 1^i, X_{<i}=x_{<i}} &\equiv Q_{X_i \mid X_{<i}, B_{\leq i}}\label{cond:Q_i}
	\end{align}
	and for any fixing of $x_{<i} \in \Supp(\hP_{X_{<i} \mid Y_{<i} = 1^{i-1}})$:
	\begin{align}\label{cond:Y_i}
		\lefteqn{P_{Y_i \mid Y_{<i} = 1^{i-1}, X_{<i}=x_{<i}}(1)}\\
		 &\equiv \eex{x\gets P_{X \mid Y_{<i} = 1^{i-1}, X_{<i}=x_{<i}}}{P[A_i \mid X=x, A_{<i}] \cdot  Q[B_i \mid X_{< i}=x_{< i}, B_{<i}]}\nonumber\\
		& \equiv P[A_i \mid X_{<i}=x_{<i}, A_{<i}] \cdot Q[B_i \mid X_{< i}=x_{< i}, B_{<i}]\nonumber\\
		 &\equiv \eex{x\gets Q_{X \mid Y_{<i} = 1^{i-1}, X_{<i}=x_{<i}}}{P[A_i \mid X_{<i}=x_{<i}, A_{<i}] \cdot  Q[B_i \mid X=x, B_{<i}]}\nonumber\\
		&\equiv Q_{Y_i \mid Y_{<i}=1^{i-1}, X_{<i}=x_{<i}}(1). \nonumber
	\end{align}
By \cref{cond:P_<i,cond:P_i,cond:Q_i}:
\begin{align}\label{eq:smooth-div-main-prop:3}
\eex{P_{X_{<i}\mid Y_{\le i} =1^i }}{D\paren{P_{X_i|X_{<i},Y_{\le i}=1^i} || Q_{X_i|X_{<i},Y_{\le i}=1^i}}} =\eex{P_{X_{<i}\mid C_{\le i} }}{D\paren{P_{X_i \mid X_{<i},A_{\leq i}} || Q_{X_i \mid X_{<i}, B_{\leq i}} }}
\end{align}
and by  \cref{cond:Y_i}, for any fixing of $x \in \Supp(P_{X_{<i} \mid Y_{<i}=1^{i-1}})$:
\begin{align}\label{eq:smooth-div-main-prop:2}
D\paren{P_{Y_i|X_{<i} = x,Y_{<i}=1^{i-1}} || Q_{Y_i|X_{<i}=x,Y_{<i} =1^{i-1}}} =0
\end{align}	 

	We use  \cref{eq:smooth-div-main-prop:2,eq:smooth-div-main-prop:3} for proving  \cref{eq:goal-under-exten}, by applying on both distributions a function   that ``cuts'' all values after the first appearance of $Y_i = 0$. Let  $f_\cut(y_1, x_1,\ldots y_m, x_m) = (y_1, x_1,\ldots y_m, x_m)$ if $y = (y_1,\ldots,y_m) = 1^m$, and $f_\cut(y_1, x_1,\ldots y_m, x_m) = (y_1, x_1,\ldots y_{i-1},x_{i-1}, y_i, \perp^{2n -2i + 1})$ otherwise, where $i$ is the minimal index with $y_i = 0$, and $\perp$ is an arbitrary symbol $\notin \cU$. By definition, 
	\begin{align*}
	\ppr{s \sim \hP_{Y_1, X_1, \ldots, Y_m, X_m}}{f_\cut(s) \neq s} = \hP[Y \neq 1^m] = 1 - P[C_{\leq m}],
	\end{align*}
	and by  \cref{eq:smooth-div-main-prop:2,eq:smooth-div-main-prop:3} along with data-processing of standard KL-divergence (\cref{fact:prelim:diver-properties}(\ref{fact:diver-properties:item:chain-rule})),
	\begin{align*}
	D\paren{f_\cut\paren{P_{Y_1, X_1, \ldots, Y_m, X_m}} || f_\cut\paren{Q_{Y_1, X_1, \ldots, Y_m, X_m}}} \leq  \sum_{i=1}^m D(P_{X_i \mid A_{\leq i}} || Q_{X_i \mid B_{\leq i}} \mid P_{X_{<i} \mid C_{\leq i}}).
	\end{align*}
	That is, $f_\cut$ is the function realizing the stated bound on the smooth KL-divergence of $P_X$ and  $Q_X$.

	\remove{
	By chain-rule of KL-divergence, \cref{fact:prelim:diver-properties}(\ref{fact:diver-properties:item:chain-rule}), it suffices to show that for all $i\in [m]$:
	\begin{align*}
	D(\cutP_{Y_i, X_i} || \cutQ_{Y_i, X_i} \mid \cutP_{Y_{<i}, X_{<i}}) \leq D(P_{X_i \mid A_{\leq i}} || Q_{X_i \mid B_{\leq i}} \mid P_{X_{<i} \mid C_{\leq i}})
	\end{align*}
	By \cref{cond:Y_i} and the definition of $f_\cut$, the above holds by showing that for all $i \in [m]$:
	\begin{align}
	D(\cutP_{X_i \mid Y_{\leq i}=1^i} || \cutQ_{X_i \mid Y_{\leq i}=1^i} \mid \cutP_{X_{<i} \mid Y_{\leq i}=1^i}) \leq D(P_{X_i \mid A_{\leq i}} || Q_{X_i \mid B_{\leq i}} \mid P_{X_{<i} \mid C_{\leq i}}).
	\end{align}
	This concludes the proof since
	\begin{align*}
	D(\cutP_{X_i \mid Y_{\leq i}=1^i} || \cutQ_{X_i \mid Y_{\leq i}=1^i} \mid \cutP_{X_{<i} \mid Y_{\leq i}=1^i})
	&= D(\hP_{X_i \mid Y_{\leq i}=1^i} || \hQ_{X_i \mid Y_{\leq i}=1^i} \mid \hP_{X_{<i} \mid Y_{\leq i}=1^i})\\
	&= D(P_{X_i \mid A_{\leq i}} || Q_{X_i \mid B_{\leq i}} \mid P_{X_{<i} \mid C_{\leq i}}).
	\end{align*}
	The second equality holds by \cref{cond:P_i,cond:Q_i,cond:P_<i}.
}
\end{proof}

\section{Skewed Distributions}\label{sec:SkewedDistributions}
In this section we formally define the notion of many-round skewed distributions and state our main result for such distributions.

\begin{definition}[The skewed distribution $Q$]\label{def:SkewedDistrbiution}
	Let  $\Unf$ be a distribution  with $\Unf_X$ being a distribution over $m \times n$ matrices,  and let $W$ and $\cE= \set{E_{i,j}}_{i \in [m], j \in [n]}$ be events over $\Unf$. We define the skewed distribution $\Rll_{X,J} = \Rll(\Unf,W,\cE)$ of $\Idl_{X} = \Unf|W$,  by $\Rll_{J} = U_{[n]}$ and 
	\begin{align*}
		\Rll_{\bX|J} =  \prod_{i=1}^{m} \Unf_{X_{i,J} | X_{<i,J}} \Idl_{X_{i,-J} | \bX_{<i}, X_{i,J}, E_{i,J}} 
	\end{align*}
\end{definition}

\begin{definition}[dense and prefix events]\label{def:denseEvents}
	Let   $\Unf_X$ be  a distribution over $m \times n$ matrices, and let $\cE = \set{E_{i,j}}_{i\in[m],j\in[n]}$  be  an event family over $\Unf_X$ such that $E_{i,j}$, for each $i,j$, is determined by  $X^j$. The family $\cE$ has {\sf density $\delta$}  if $\forall (i,j) \in [m]\times [n]$ and  for any fixing of $X_{\leq i, j}$, it holds that $\Unf[E_{i,j} | X_{\leq i, j}] = \delta_{i,j}\ge \delta$.  The family $\cE$ is a {\sf prefix family} if $\forall (i,j) \in [m]\times [n]$ the event $E_{i,j}$ is determined by $X_{\leq i+1,j}$.
\end{definition}

\paragraph{Bounding smooth KL-divergence of smooth distributions.}
The following theorem states our main result for skewed distributions. In \cref{sec:BoundingSmoothKL:Warmup} we give a proof sketch of \cref{thm:BoundingSmoothKL}, and in \cref{sec:ConditionalDits} we give the full details.

\begin{theorem}\label{thm:BoundingSmoothKL}
	Let  $\Unf$ be a distribution  with $\Unf_X$ being a distribution over $m \times n$ matrices with independent columns, let $W$ be an event over $\Unf$ and  let $\cE = \set{E_{i,j}}$ be a $\delta$-dense event family over $\Unf_X$. Let $\Idl = \Unf|W$ and let $\Rll_{X,J} = \Rll(\Unf,W,\cE)$ be the skewed variant of $\Idl$ defined in \cref{def:SkewedDistrbiution}.  Let   $Y_i = (Y_{i,1},\ldots,Y_{i,n})$ for $Y_{i,j}$ being the indicator for $E_{i,j}$, and let $\dval = \sum_{i=1}^{m} D(\Idl_{\bX_i \bY_i} || \Unf_{\bX_i \bY_i} | \Idl_{\bX_{<i}})$.   Assuming $n \geq c \cdot m/\delta$ and $d \leq \delta n / c$, for a universal constant $c > 0$, then
	\begin{align*}
	D^{\frac{c}{\delta n}(d+1)}(\Idl || \Rll) \leq \frac{c}{\delta n} (\dval + m).
	\end{align*}
\end{theorem}

We now prove that \cref{thm:intro:BoundingSmoothKL} is an immediate corollary of \cref{thm:BoundingSmoothKL}.

\begin{corollary}[Restatement of \cref{thm:intro:BoundingSmoothKL}]\label{cor:BoundingSmoothKL}
	Let $\Unf,\Idl,\Rll,W,\cE,\delta$ and $c$ be as  in \cref{thm:BoundingSmoothKL}, and let $\eps = \log(\frac1{\Unf[W]})/\delta n$. Then the following hold assuming $n \geq c \cdot m/\delta$:
	\begin{itemize}
		\item  if   $\Unf[W] \geq \exp\paren{-\delta n/c m}$, then $D^{c\cdot (\eps m + 1/\delta n)}(\Idl || \Rll) \leq  c\cdot(\eps m + m/\delta n)$, and 
		
		\item  	if   $\Unf[W] \geq \exp\paren{-\delta n/2c}$ and $\cE$ is a prefix family, then
		$D^{2c\cdot (\eps + 1/\delta n)}(\Idl || \Rll) \leq 2c\cdot (\eps + m/\delta n)$.
	\end{itemize}
\end{corollary}

\begin{proof}
	Let $\set{Y_{i,j}}$  be as in \cref{thm:BoundingSmoothKL}.
	Note that for each $i \in [m]$:
	\begin{align*}
		D(\Idl_{\bX_i \bY_i} || \Unf_{\bX_i \bY_i} \mid \Idl_{\bX_{<i}}) 
		\leq D(\Idl_{\bX_{\geq i}} || \Unf_{\bX_{\geq i}} \mid \Idl_{\bX_{<j}}) 
		\leq D(\Idl_{\bX} || \Unf_{\bX})
		\leq \log \frac1{\Unf[W]}.
	\end{align*}
	The first inequality holds by data-processing of KL-divergence (\cref{fact:prelim:diver-properties}(\ref{fact:diver-properties:item:data-processing})). The second inequality holds by chain-rule of KL-divergence (\cref{fact:prelim:diver-properties}(\ref{fact:diver-properties:item:chain-rule})).  The last inequality holds by \cref{fact:kl-divergence:diver-cond-on-event}.
	Assuming $\Unf[W] \geq \exp\paren{-\delta n/c m}$, it holds that
	\begin{align*}
		\dval \leq m \cdot \log \frac1{\Unf[W]} \leq \delta n/c,
	\end{align*}
	 concluding the proof of the first part.

	Assuming $\Unf[W] \geq \exp\paren{-\delta n/c}$ and $\cE$ is a prefix family (\ie   $E_{i,j}$ is a function  of  $\bX_{\leq i+1}$), then
	\begin{align*}
		\dval 
		&\leq \sum_{i=1}^{m-1} D(\Idl_{\bX_i \bX_{i+1}} || \Unf_{\bX_i \bX_{i+1}} \mid \Idl_{\bX_{<i}}) + D(\Idl_{\bX_m} || \Unf_{\bX_m} \mid \Idl_{\bX_{<m}})\\
		&= \sum_{i \in [m-1] \cap \N_{even}} D(\Idl_{\bX_i \bX_{i+1}} || \Unf_{\bX_i \bX_{i+1}} \mid \Idl_{\bX_{<i}}) + \sum_{i \in [m-1] \cap \N_{odd}} D(\Idl_{\bX_i \bX_{i+1}} || \Unf_{\bX_i \bX_{i+1}} \mid \Idl_{\bX_{<i}})\nonumber\\
		&+ D(\Idl_{\bX_m} || \Unf_{\bX_m} \mid \Idl_{\bX_{<m}})
		\leq 2 \cdot D(\Idl_X || \Unf_X)\nonumber\\
		&\leq 2 \cdot \log \frac1{\Unf[W]} \leq \delta n/c,\nonumber
	\end{align*}
	concluding the proof of the second part. The first inequality holds by data-processing of KL-divergence,  and the second one holds by chain-rule and data-processing of KL-divergence.  
\end{proof}

In order to show that the attacking distribution $\Rll$ can be carried out efficiently, it suffice to show that with high probability over $(x,j)\sim Q_{X,J}$, we have for all $i \in [m]$ that $\Unf[W \mid (X_{<i},X_{i,j})=(x_{<i}, x_{i,j}),E_{i,j}]$ is not much smaller than $\Unf[W]$. The following lemma (proven in \cref{sec:ConditionalDits}) states that the above holds under $\Idl_X$. Namely, when sampling $x \sim \Idl_X$ (instead of $x \sim \Rll_X$) and then $j \sim \Rll_{J|X=x}$, then $\Unf[W \mid (X_{<i},X_{i,j})=(x_{<i}, x_{i,j}),E_{i,j}]$ is indeed not too low. 

\begin{lemma}\label{lemma:ideal-running-time}
	Let $\Unf,\Idl,\Rll,W,\cE,\delta,\dval$ be as  in \cref{thm:BoundingSmoothKL}, let $t > 0$ and let
	\begin{align*}
		p_t \eqdef \ppr{x \sim \Idl_X \; ; \; j \sim \Rll_{J|X=x}}{\exists i \in [m]: \Unf[W \mid (X_{<i},X_{i,j})=(x_{<i}, x_{i,j}),E_{i,j}] < \Unf[W]/t}
	\end{align*}
	Assuming $n \geq c \cdot m/\delta$ and $d \leq \delta n / c$, for a universal constant $c > 0$, then
	\begin{align*}
		p_t \leq 2m/t + c(\dval+1)/(\delta n).
	\end{align*}
\end{lemma}

As  an immediate  corollary, we get the following result.

\begin{corollary}\label{cor:ideal-running-time}
	Let $\Unf,\Idl,\Rll,W,\cE,\delta$ be as in \cref{thm:BoundingSmoothKL}, let $\eps = \log(\frac1{\Unf[W]})/\delta n$, let $t > 0$ and let $c$ and $p_t$ as in \cref{lemma:ideal-running-time}.
	Assuming $n \geq c \cdot m/\delta$, it holds that
	\begin{itemize}
		\item  if   $\Unf[W] \geq \exp\paren{-\delta n/c m}$, then 
		$p_t \leq 2m/t + c\cdot(\eps m + 1/\delta n)$.
		
		\item  	if   $\Unf[W] \geq \exp\paren{-\delta n/2c}$ and $\cE$ is a prefix family, then
		$p_t \leq 2m/t + 2c\cdot (\eps + 1/\delta n)$.
	\end{itemize}
\end{corollary}

\remove{

\paragraph{Minimum hitting  probability.}
\Enote{read}The following fact allow us to argue that if for any, not too small, $W$ it holds that small events in $\Idl$ are also small in $\Rll$, then  the sampling implicitly done in the definition of $\Rll$ can be carried out  efficiently (\ie proportional to $1/\Unf[W]$).

\begin{definition}[Minimum hitting  probability]\label{def:MinWinProb}
	Let $\Unf$, $W$, $\cE$ and let $\Rll= \Rll(\Unf,W,\cE)$ be as in \cref{def:SkewedDistrbiution}. For an event $T$ over $\Unf$ and $(x,j) \in \Supp(\Rll_{X,J})$, let $$h_T(x,j) = \min_{i\in [m]} \set{\Unf[T\mid (X_{<i},X_{i,j}) = (x_{<i},x_{i,j}),E_{i,j}]}.$$
\end{definition}

\Enote{Idea that not completely work for small number of repetitions: Define an extension $P'$of $P$, where each $X_{i,j} = \perp$ with small probability (say, $(1-\frac{\eps^2 P[W]^2}{m^2})/(2^j n)$). Then under $\Idl'$ it holds that the probability of getting $\perp$ is very low, but under $\Rll'$, assuming that $\Idl[W \mid X_{<i},X_{i,J},E_{i,J}] < 2^{-m}\cdot \poly(\eps,P[W],1/m,1/n)$ with probability at least (say) $\eps/3$ then we can show that the probability of $\perp$ in $\Rll'$ is at least $\eps/9$, and this will lead to a contradiction. The problem is that $2^{-m}$ won't suffice for small number of repetitions.}

\begin{proposition}\label{lemma:MinWinProb}
	Let $\Unf$, $W$, $\cE$,  $\Rll= \Rll(\Unf,W,\cE)$ be as in as	 \cref{def:SkewedDistrbiution},  let $h= h_W$ be according to  \cref{def:MinWinProb}, let $d \in \N$ and let $\gamma = \ppr{(x,j)\gets \Rll}{h(x,j) < \Unf[W]/d}$. Then there exist an extension $\Unf'$ of $\Unf$, and events $W' \supset W $ and $V$ over  $\Unf'$ such that the following holds: let $\Idl' = \Unf'| W'$ and $\Rll'= \Rll(\Unf',W',\cE)$, then 
	\begin{enumerate}
		\item  $\Idl'[V] \le 5m/\sqrt{d}$,  and
		\item $\Rll'[V] \ge (\gamma - m/\sqrt{d})/16$.
	\end{enumerate}
\end{proposition}
\Inote{explain relevance to running time}

\begin{proof}
	Assuming  $\Unf= \Unf_{ZX}$, let $\Unf'= \Unf'_{ZXR}$ be (a product)  extension  of $\Unf$---the value of $R$ is independent of $(Z,X)$---for $\Unf'_R$ to be be determined below.  Let  $W' = W \cup L$, for $L$  being an arbitrary  event  over $\Unf'$ that  is determined by $R$ and happens with probability $\Unf[W]/\sqrt{d}$ \Inote{$\Unf[W]/d$?}. Let $V$ be the event $\set{\exists i \colon  \Unf[W \mid X_i] \le 4\cdot \Unf[W]/\sqrt{d}}$. By Bayes rule and a union bound, 
	 \begin{align}
	 \Idl[V] \le 4m/\sqrt{d}
	 \end{align}
	 Thus,   
	 
	 \begin{align*}
	 	\Idl'[V] 
	 	&= \Unf'[V \mid W \cup L] = \frac{\Unf'[V \cap (W \cup L)]}{\Unf'[W \cup L]} = \frac{\Unf'[(V \cap W) \cup (V \cap L)]}{\Unf'[W \cup L]}\\
	 	&\leq \frac{\Unf'[V \cap W]}{P'[W]} + \frac{\Unf'[L]}{\Unf'[W \cup L]} = \Idl[V] + \frac{\Unf'[L \cap W']}{\Unf'[W']}\\
	 	&\leq \Idl[V] + \frac{\Unf[W]/\sqrt{d}}{\Unf[W]} \leq 5m/\sqrt{d},
	 \end{align*}
	concluding the proof of the first part.
	
	For the second part,  note that for every fixing of $(j,X_{< i},X_{i,j})$:
	$$\paren{\Unf[W \mid X_{< i},X_{i,j},E_{i,j}] >  \Unf[W]/\sqrt{d}} \implies \paren{\Unf'[L \mid X_{< i},X_{i,j},E_{i,j},W'] <  1/\sqrt{d}}.$$ 
	Thus, a simple coupling argument between  $\Rll'$ and $\Rll$ yields that
	\begin{align}\label{eq:MinWinProb:2}
	\ppr{(x,j)\gets \Rll'}{h(x,j) \le  \Unf[W]/\sqrt{d}} &\ge \ppr{(x,j)\gets \Rll}{h(x,j) \le  \Unf[W]/\sqrt{d}}  - m/\sqrt{d} \\
	&\ge \gamma  - m/\sqrt{d}.\nonumber
	\end{align}
	One the other hand, the definition of $\Rll'$ yields that for any fixing of $(J,X_{< i},X_{i,J})$:
	$$\paren{\Rll'[W \mid J,X_{<i},X_{i,J},E_{i,J}] \le 4\cdot \Unf[W]/\sqrt{d}} \implies \paren{\Rll'[L \mid J, X_{<i},X_{i,J},E_{i,J},W']\ge 1/8}.$$ 
	Hence, by a Markov bound, for any fixing of   $(J,X_i,X_{i,J})$ with $\Rll'[W \mid J,X_{<i},X_{i,J},E_{i,J}] \le 4\Unf[W]/\sqrt{d}$:
	\begin{align}\label{eq:MinWinProb:3}
	\ppr{\Rll'_X \mid J,X_{<i},X_{i,J}}{\Unf[W \mid X_i] \le 4\cdot \Unf[W]/\sqrt{d}} \ge 1/16 
	\end{align}
	Combining \cref{eq:MinWinProb:2,eq:MinWinProb:3}, we deduce  that $\Rll'[V] \ge (\gamma - m/\sqrt{d})/16$, concluding the proof of the second part.
\end{proof}

The following proposition extends \cref{lemma:MinWinProb} for events $W$ that are product over $\Unf_X$.
\begin{definition}[Product events]\label{def:productEvent}
	Let  $\Unf= \Unf_X$ be a distribution over $m \times n$ matrices. A sequence of events $W_1,\ldots,W_\ell$ is {\sf a product partition of an event $W$ over $\Unf_X$  \wrt partition $\cJ_1,\ldots,\cJ_\ell$ of $[n]$}, if  $W = \bigwedge W_j$ and each $W_j$ is determined by $X^{\cJ_j}$.
	\end{definition}

\begin{proposition}\label{lemma:MinWinProbProduct}
Let $\Unf$, $W$, $\cE$,  $\Rll= \Rll(\Unf,\cE,W)$ be as in as	 \cref{def:SkewedDistrbiution}, let $d \in \N$ and  let $h$ be as in \cref{def:MinWinProb}. Let $W_1,\ldots,W_\ell$ be product partition of  $W$  \wrt partition $\cJ_1,\ldots,\cJ_\ell$ of $[n]$, and let $\gamma = \ppr{(x,j)\gets \Rll}{\min_{r\in [\ell]}\set{ h_{W_{r}}(x,j)} < \Unf[W_\I]/d}$. Then there exist an extension $\Unf'$ of $\Unf$, and events $W' \supset W$ and $V$ over  $\Unf$ such that the following hold: let  $\Idl' = \Unf'|W$  and $\Rll'= \Rll(\Unf',W',\cE)$, then 
\begin{enumerate}
	\item  $\Idl'[V] \le 5\ell m/\sqrt{d}$,  and
	\item $\Rll'[V] \ge (\gamma - \ell m/\sqrt{d})/16$.
\end{enumerate}
\begin{proof}
A simple extension of the proof of \cref{lemma:MinWinProb}. \Inote{verify}
\end{proof}
\end{proposition}

}

\section{The Parallel Repetition Theorem}\label{sec:PR}
In this section, we use \cref{thm:BoundingSmoothKL} for prove \cref{thm:intro:PR}, restated below.

\begin{theorem}[Parallel repetition for partially simulatable arguments, restatement of \cref{thm:intro:PR}]\label{thm:PR}
Let $\pi$ be an $m$-round $\delta$-simulatable [\resp prefix $\delta$-simulatable] interactive argument of soundness error $1-\eps$. Then  $\pi^n$ has soundness error  $(1-\eps)^{cn\delta/m}$ [\resp $(1-\eps)^{cn\delta}$], for a universal constant $c>0$. 
\end{theorem}
Since the random terminating variant of an $m$-round interactive argument is  $1/m$-prefix-simulatable, the (tight) result for such protocols immediately follows. The proof of \cref{thm:PR} follows  from our bound on the smooth KL-divergence of skewed distributions, \cref{thm:BoundingSmoothKL}, and  \cref{lemma:KLtoPR}, stated  and proven below.


\newcommand{\Bad}{\mathsf{Bad}}
\begin{definition}[bounding function for many-round skewing]\label{def:bounding-func}
	A function $f$ is a {\sf bounding function for many-round skewing} if there exists a polynomial $p(\cdot,\cdot)$ such that the following holds for every $\delta \in (0,1]$ and every $m,n \in \N$ with $n > p(m,1/\delta)$: let $\Unf$ be a distribution with  $\Unf_X$ being a column independent distribution  over  $m \times n$ matrices. Let  $W$ be an event and let $\cE$ be  a $\delta$-dense  [\resp prefix $\delta$-dense]   event family over $\Unf$ (see \cref{def:denseEvents}). Let $\Idl = \Unf |W$ and let $\Rll = \Rll(\Unf,W,\cE)$ be according to \cref{def:SkewedDistrbiution}. Then the following holds for $\gamma = \log\paren{1/\Unf[W]} / f(n,m,\delta)$:
	\begin{enumerate}
		\item $ \Rll_X[T]  \le 2 \cdot \Idl_X[T]    + \gamma$	for every event $T$,\footnote{The constant $2$ can be replaced with any other constant without changing (up to a constant factor) the decreasing rate which is promised by \cref{lemma:KLtoPR}.} and \label{def:bounding-func:events}

		\item   $\ppr{x \sim \Idl_X  \; ; \; j \sim \Rll_{J|X=x}}{(x,j)\in \Bad_t} \quad \le \quad p(m,1/\delta)/t +  \gamma$ for every $t>0$, letting  
		$$ \Bad_t \eqdef \set{(x,j)\colon \exists i \in [m]: \Unf[W \mid (X_{<i},X_{i,j})=(x_{<i}, x_{i,j}),E_{i,j}] <  \Unf[W]/t}.$$ \label{def:bounding-func:run-time}
		
	\end{enumerate}
\end{definition}

\begin{lemma}[Restatement of \cref{lemma:intro:KLtoPR}]\label{lemma:KLtoPR}	
Let $\pi$ be an $m$-round $\delta$-simulatable    [\resp prefix $\delta$-simulatable] interactive argument of soundness error $1-\eps$,  let $f$ be a bounding function for many-round skewing (according to \cref{def:bounding-func}). Then $\pi^n$ has soundness error  $(1-\eps)^{f(n,m,\delta)/160}$.

\end{lemma}

That is, \cref{lemma:KLtoPR} tells us that the task of maximizing the decreasing rate of $\pi^n$ is directly reduces to the task of maximizing a bounding function for many-round skewing. A larger bounding function yields a smaller $\gamma$ in \cref{def:bounding-func}. This $\gamma$ both defines an additive bound on the difference between a small event in $\Idl$ to a small event in $\Rll$, and bounds a specific event in $\Idl$ that captures the cases in which an attack can be performed efficiently.

We first prove \cref{thm:PR} using  \cref{lemma:KLtoPR}.

\paragraph{Proof of \cref{thm:PR}.}
\begin{proof}
We prove for $\delta$-simulatable arguments,  the proof for  $\delta$-prefix-simulatable arguments follows accordingly. Let $m,n$, $\Unf$, $\delta$, $\cE$, $W$, $\Idl$ and  $\Rll$ be as in \cref{lemma:KLtoPR},  where $\cE$ is $\delta$-dense, and let $c = \max\set{c',c''}$ where $c'$ is the constant from \cref{cor:BoundingSmoothKL} and $c''$ is the constant from \cref{cor:ideal-running-time}.
By 	\cref{cor:BoundingSmoothKL}, if  $n \geq c \cdot m/\delta$ and $\Unf[W] \geq \exp\paren{-\delta n/cm}$, then 
\begin{align}\label{eq:thm:PR:0}
D^{3c m\mu}(\Idl || \Rll) \le 3c m \mu
\end{align}
for  $\mu = \log(1/\Unf[W])/\delta n$, where we assumed \wlg   that $\Unf[W] \le  1/2$. Hence, assuming that $n \geq c \cdot m/\delta$ and $\Unf[W] \geq \exp\paren{-\delta n/cm}$, \cref{prop:prelim:smooth-div,eq:thm:PR:0} yields that  for every event $T$:
\begin{align}\label{eq:thm:PR:1}
\Rll[T]  \le  2\cdot \Idl[T] + \gamma,
\end{align}
where  $\gamma =  \log(1/\Unf[W])/f(n,m,\delta)$ for $f(n,m,\delta) = \delta n / (24cm)$. For event  $W$ of smaller probability, it holds that $\gamma \ge 24$, and therefore \cref{eq:thm:PR:1} trivially holds for such events. In addition, by \cref{cor:ideal-running-time}, if  $n \geq c \cdot m/\delta$ and $\Unf[W] \geq \exp\paren{-\delta n/cm}$, then
\begin{align}\label{eq:thm:PR:run-time}
	\ppr{x \sim \Idl_X \; ; \; j \sim \Rll_{J|X=x}}{\exists i \in [m]: \Unf[W \mid (X_{<i},X_{i,j})=(x_{<i}, x_{i,j}),E_{i,j}] < \Unf[W]/t}
	\leq 2m/t + \gamma,
\end{align}
where for event $W$ of smaller probability, \cref{eq:thm:PR:run-time} trivially holds.
By \cref{eq:thm:PR:1,eq:thm:PR:run-time}, $f$ is a bounding function for many-round skewing with the polynomial $p(m,1/\delta) =  c\cdot m/\delta$. Therefore, \cref{lemma:KLtoPR} yields that the soundness  error of $\pi^n$ is bounded by $(1-\eps)^{f(n,m,\delta)/80}= (1-\eps)^{\delta n/(c' m)}$, for $c' = 1920c$. 
\end{proof}

\subsection{Proving \cref{lemma:KLtoPR}}\label{sec:KLtoPR}

Let $f$ be a bounding function for many-round skewing with the polynomial $p(\cdot,\cdot) \in \poly$.
We first prove the case when the number of repetition $n$ is at least $p(m,1/\delta)$, and then  show how to extend the proof for the general case.
\paragraph{Many repetitions case.}

\begin{proof}[Proof of \cref{lemma:KLtoPR}, many repeitions]
Fix an $m$-round   $\delta$-simulatable interactive argument $\pi=(\P,\V)$ of soundness error $1-\eps$ (the proof of the  $\delta$-prefix-simulatable case follows the same lines),   and let $n = n(\kappa)> p(m(\kappa),1/\delta(\kappa))$.  Note that \wlg  $\eps(\kappa)\ge 1/\poly(\kappa)$.

 Our proof is a black-box reduction according to  \cref{def:BBProof}: we present an oracle-aided algorithm that given access to a  deterministic cheating  prover  for $\pi^n$ violating the claimed soundness of  $\pi^n$, uses it to break the assumed soundness of $\pi$ while not running for  too long. The lemma  then follows by \cref{prop:BBProof}.

Let $\Sc$ be the oracle-aided simulator guaranteed by the  $\delta$-simulatablily of $\Vc$.  For a cheating prover $\nsP$  for $\pi^n$, let $\sP$ be the cheating prover that  for interacting with $\Vc$, emulates a random execution of $(\nsP,\nV)$, letting   $\Vc$ plays one of the $n$ verifiers, at a random location. (Clearly, $\sP$  only requires oracle access to $\nsP$.) Assume \wlg that in each round $\Vc$ flips $t = t(\kappa)$ coins. The oracle-aided algorithm $\sP$ is defined as follows. 

\begin{algorithm}[$\sP$]\label{alg:CheatingProver}
	\item Input: $1^\kappa$,  $m= m(\kappa)$ and $n= n(\kappa)$.  
	
	\item Oracles: cheating prover  $\nsP$ for $\pi^n$.
	\item Operation:
	\begin{enumerate}

		\item Let $j \la [n]$.
	
		\item For $i=1$ to $m$ do:
		\begin{enumerate}
			\item Let $a_i$ be the \ith message sent by $\Vc$.
			\item Do the following (``rejection continuation''):
			\begin{enumerate}
				\item Let $x_{i,-j} \gets (\zo^t)^{n-1}$
				\item Let  $v= \Sc^{\snsP}(1^\kappa,(j,x_{\le i,-j},a_{\le i}))$.
				\item If all $n$ verifiers accept in $v$, break the inner loop.				
			\end{enumerate}
			\item Send to $\V$ the \ith message $\snsP$ sends in $v$. 
			
		\end{enumerate}
	\end{enumerate}
\end{algorithm}
Fix a cheating prover $\nsP$. We also fix $\kappa\in \N$, and omit it from the notation. Let $\Unf= \Unf_{X}$ denotes the coins $\nV$ use in a uniform  execution of $(\nsP,\nV)$. (Hence  $\Unf_X$ is uniformly distributed over $m\times n$ matrices.) Let $W$ be the event over $\Unf$ that $\nsP$ wins in $(\nsP,\nV)$ (\ie all verifiers accept), and let $\Idl_{X} = \Unf_X|W$.  For an $i$ rounds view $v= (j,\cdot)$ of $\snsP$ in $(\snsP,\V)$, let $\Delta_v$ be the $\delta$-dense subset of $\Vc$'s coins describing the output distribution of $\Sc^{\snsP}(v)$. Let  $\ct_{i,j}$ be all possible $i$ round views of  $\snsP$ in $(\snsP,\V)$ that are starting with $j$. Finally, let  $\cE =\set{E_{i,j}}_{i\in [m],j\in [n]}$ be the event family over $\Unf$ defined by $E_{i,j} = \bigcup_{v\in \ct_{i,j}} \Delta_v$, and let $\Rll_{X,J}$ be the e (skewed) distribution described in \cref{def:SkewedDistrbiution} \wrt $\Unf,W,\cE$.  By inspection,  $\Rll$ describes  the  distribution of $(j,x_{\le m})$ in a random execution  of $(\sP,\nV)$, where  $x_{\le m,j}$ denotes  the coins of $\Vc$, and $x_{\le m,-j}$ denote the final value of this term in the execution.  Assume  
\begin{align}\label{eq:KLtoPR:00}
	\pr{(\nsP,\nV) = 1} =  \Unf[W] > (1-\eps)^{f(n,m,\delta)/80},
\end{align}
and let $\gamma = \log\paren{1/\Unf[W]}/f(n,m,\delta)$. By \cref{eq:KLtoPR:00} it holds that
\begin{align}\label{eq:KLtoPR:0}
	\gamma < -\log(1-\eps)/80\leq \eps/80
\end{align}
Since $\Idl[W]=1$, we deduce by Property~\ref{def:bounding-func}(\ref{def:bounding-func:events}) of $f$ on the event $\neg W$ that 

\begin{align}\label{eq:KLtoPR:winning}
\pr{(\sP,\Vc) = 1} \ge  \Rll_X[W] > 1-\gamma > 1-\eps/80
\end{align}
So it is left to argue about the running time of $\sP$. By Property~\ref{def:bounding-func}(\ref{def:bounding-func:run-time}) of $f$ on $t = 80 \cdot   p(m,1/\delta)/\eps$ it holds that
\begin{align*}
	\ppr{x \sim \Idl_X  \; ; \; j \sim \Rll_{J|X=x}}{(x,j)\in \Bad_t} \quad \le \quad p(m,1/\delta)/t +  \gamma < \eps/40
\end{align*}

Therefore, we now can apply Property~\ref{def:bounding-func}(\ref{def:bounding-func:events}) of $f$ on the following event ``Given $x$, 
choose $j \sim \Rll_{J|X=x}$ and check whether $(x,j)\in \Bad_t$'' (note that this event defined over an extension of $\Idl$ that additionally samples $j$ according to $Q_{J|X}$). This yields that

\begin{align}\label{eq:KLtoPR:running-time}
	\ppr{x \sim \Rll_X  \; ; \; j \sim \Rll_{J|X=x}}{(x,j)\in \Bad_t} \quad \le \quad 2\eps/40 + \gamma < \eps/10
\end{align}
By \cref{eq:KLtoPR:winning,eq:KLtoPR:running-time} we obtain that 

\begin{align}\label{eq:KLtoPR:overall}
	\ppr{(x,j) \sim \Rll_{X,J}}{W \land \paren{(x,j)\notin \Bad_t}} > 1-\eps/5
\end{align}
Namely, with probability larger than $1-\eps/5$, the attacker $\sP$ wins and its expected running time in each round is bounded by $O(t/P[W]) \leq \poly(\kappa)$. This contradicts the soundness guaranty of $\pi$.

\end{proof}

\paragraph{Any number of repetitions.}
\newcommand{\lnsP}{(\nsP)^\ell}
\newcommand{\lnV}{(\nV)^{\ell}}
\newcommand{\hsP}{\widehat{\Pc}}
\newcommand{\Ind}{q}

The assertions of  the function $f$ in \cref{eq:KLtoPR:winning,eq:KLtoPR:running-time} only guarantee to hold if $n>p(m,1/\delta)$ (for some $p(\cdot,\cdot)\in \poly$). We now prove the lemma for smaller values of repetitions. As mentioned in the  introduction, for interactive arguments (and unlike interactive proofs), there is no generic reduction from large to small number of repetitions.  Assume 
\begin{align}
\alpha \eqdef\pr{(\nsP,\nV) = 1} > (1-\eps)^{{f(m,n,\delta)/80}}
\end{align}
and let $\ell \in \poly$ be such that $\ell n \ge p(m,1/\delta)$. It is immediate that 
\begin{align}
\alpha_\ell \eqdef \pr{(\lnsP,\lnV) = 1} =  \alpha^\ell > (1-\eps)^{\ell \cdot {f(m,n,\delta)/80}}
\end{align}
for $\lnsP$ and $\lnV$ being the $\ell$ repetition of $\nsP$ and $\nV$ respectively. Therefore, the same lines as the proof above yields that the cheating prover  $\sP^{\lnsP}$ breaks the soundness of $\pi$ with probability $1-\eps/80$. The problem is that the running time of  $\sP^{\lnsP}$ is proportional to $1/\alpha_\ell$ and not to $1/\alpha$, and in particular is not polynomial  even if  $\alpha> 1/\poly$. We overcome this difficulty by giving a different (efficient) implementation of $\sP^{\lnsP}$ that takes advantage of the parallel nature of $\lnsP$.

\begin{proof}[Proof of \cref{lemma:KLtoPR}, small number of repetitions]
Let $\pi$, $\nsP$, $\sP$ and  $\S$  be as in the  proof for the many repetitions case. Let $\ell \in \poly$ be such that $\ell n \ge p(m,1/\delta)$, and for  $\Ind\in [\ell]$ let $\cZ^\Ind=  \set{(\Ind-1)n+1,\ldots, \Ind n}$.  The oracle-aided algorithm $\hsP$ is defined as follows. 

\begin{algorithm}[$\hsP$]\label{alg:CheatingProverSmallRep}
	\item Input: $1^\kappa$,  $m= m(\kappa)$, $n= n(\kappa)$ and $\ell=\ell(\kappa)$.
	
	\item Oracles: cheating prover $\nsP$ for $\pi^n$.
	\item Operation:
	\begin{enumerate}

		\item Let $j \la [n\ell]$.
		
		\item For $i=1$ to $m$ do:
		\begin{enumerate}
			\item Let $a_i$ be the \ith message sent by $\Vc$.
			\item For $\Ind = 1$ to $\ell$ do the following (``rejection continuation''):

			If $j\in  \cZ^\Ind$:
			
			\begin{enumerate}
				\item Let $x_{i,\cZ^\Ind \setminus \set {j}} \gets (\zo^t)^{n-1}$.
				\item Let  $v= \Sc^{\snsP}(1^\kappa,(j \bmod n,x_{i,\cZ^\Ind \setminus \set {j}} ,a_{\le i}))$.
				\item If all $n$ verifiers accept in $v$, break the inner loop.				
			\end{enumerate}
			Else,
			\begin{enumerate}
				\item Let $x_{>i,\cZ^\Ind} \gets (\zo^t)^{n}$
				\item If all $n$ verifiers accept in $x_{\cZ^\Ind}$, break the inner loop.				
			\end{enumerate}

			\item  Send $\V$ the \ith message $\snsP$ sends in $v$. 
		\end{enumerate}
	\end{enumerate}
\end{algorithm}
Namely, $\hsP^\nsP$ emulates $\sP^{\lnsP}$, for $\lnsP$ being the $\ell$ parallel repetition of  $\nsP$, while exploiting the product nature of $\lnsP$  for separately sampling the coins of each the $\ell$ groups of verifiers.  

Fix a cheating prover $\nsP$ and $\kappa \in \N$, and define  $\Unf= \Unf_X$,  $W$, $\Rll_{X,J}$ \wrt a random execution of $(\lnsP,\lnV)$ as done in the proof for large number of repetition.   Assume  
\begin{align}\label{eq:KLtoPR:S:0}
\pr{(\nsP,\nV) = 1} > (1-\eps)^{f(m,n,\delta)/80}
\end{align}
then 
\begin{align}\label{eq:KLtoPR:S:01}
\Unf[W] = \pr{(\lnsP,\lnV) = 1}  > (1-\eps)^{\ell \cdot f(m,n,\delta)/80}
\end{align}

\cref{eq:KLtoPR:winning} yield that  
\begin{align}\label{eq:KLtoPR:Small:winning}
\pr{(\hsP,\Vc) = 1} > 1-\eps/80
\end{align}
So it is left to argue about the running time $\hsP$. For $q\in [\ell]$,  let $W_q$ be the event that all verifiers in $\cZ^\Ind$ accept in $\Unf_{X}$.  Note that $\Unf[W_q] = \pr{(\nsP,\nV) = 1} = \alpha$ and that $\Unf[W] = \alpha^{\ell}$. Moreover, For $j \in [n \ell]$, let $q_j$ be the (unique) value $q \in [\ell]$ such that $j \in \cZ^{q}$.
By \cref{eq:KLtoPR:running-time} it holds that
\begin{align}\label{eq:KLtoPR:Small:running-time}
	\ppr{(j,x) \sim \Rll_{J,X}}{(x,j)\in \Bad_t} \quad  < \eps/10
\end{align}
for $t = 80\cdot p(m,1/\delta)/\eps$, where recall that
\begin{align*}
	\Bad_t = \set{(x,j)\colon \exists i \in [m]: \Unf[W \mid (X_{<i},X_{i,j})=(x_{<i}, x_{i,j}),E_{i,j}] <  \Unf[W]/t}.
\end{align*}
Note that by construction, it holds that
\begin{align}\label{eq:KLtoPR:Small:product-explanation}
	\lefteqn{\Unf[W \mid (X_{<i},X_{i,j})=(x_{<i}, x_{i,j}),E_{i,j}]}\\
	&= \Unf[W_{q_j} \mid X_{<i,\cZ^{q_j}} = x_{<i,\cZ^{q_j}}, X_{i,j}=x_{i,j}, E_{i,j}] \cdot \prod_{q \in [\ell]\setminus\set{q_j}} \Unf[W_{q} \mid X_{<i,\cZ^{q}} = x_{<i,\cZ^{q}}].\nonumber
\end{align}
Moreover, by Markov inequality we have
\begin{align}\label{eq:KLtoPR:Small:product-markov}
	\ppr{(j,x) \sim \Rll_{J,X}}{\prod_{q \in [\ell]\setminus\set{q_j}} \Unf[W_{q} \mid X_{<i,\cZ^{q}} = x_{<i,\cZ^{q}}] > 10\cdot \alpha^{\ell-1}/\eps} < \eps/10.
\end{align}
Recall that $\Unf[W] = \alpha^{\ell}$. Therefore, by \cref{eq:KLtoPR:Small:running-time,eq:KLtoPR:Small:product-explanation,eq:KLtoPR:Small:product-markov} we deduce that
\begin{align}\label{eq:KLtoPR:Small:hard-copy}
	\ppr{(j,x) \sim \Rll_{J,X}}{\exists i \in [m]: \Unf[W_{q_j} \mid X_{<i,\cZ^{q_j}} = x_{<i,\cZ^{q_j}}, X_{i,j}=x_{i,j}, E_{i,j}] < \eps \alpha/(10 t)} \quad  < \eps/5
\end{align}
Moreover, by \cref{fact:Prelim:smooth-sampling} along with Markov inequality and a union bound, we have
\begin{align}\label{eq:KLtoPR:Small:easy-copy}
	\ppr{(j,x) \sim \Rll_{J,X}}{\exists (i,q) \in [m]\times\paren{[\ell]\setminus\set{q_j}}: \Unf[W_{q} \mid X_{<i,\cZ^{q}} = x_{<i,\cZ^{q}}] < \eps \alpha/(5 m)} \quad  < \eps/5
\end{align}
Hence, \cref{eq:KLtoPR:Small:hard-copy,eq:KLtoPR:Small:easy-copy} yields that with probability $> 1- \eps/2$ it holds that at the beginning of each inner round of $\hsP$, the expected running time of it is bounded by $\max\set{10 t/(\eps \alpha), 5m/(\eps \alpha)} \leq \poly(\kappa)$. This (along with \cref{eq:KLtoPR:Small:winning}) contradicts the soundness guarantee of $\pi$.

\end{proof}

\section{Bounding Smooth KL-Divergence of Skewed Distributions}\label{sec:BoundingSmoothKL}
In this section we prove \cref{thm:BoundingSmoothKL}. As a warmup, we  give in \cref{sec:BoundingSmoothKL:Warmup}   a   proof sketch  and explain the difficulties that arise. In \cref{sec:ConditionalDits} we define conditional variants  of $\Idl$ and $\Rll$, and use \cref{lemma:smooth-div-main-prop} to prove the theorem assuming that (1) the standard  KL-divergence of these variants is small and (2) these variants are not too far, in the  sense that allows us to use  \cref{lemma:smooth-div-main-prop}, from their origin. We prove (1) in \cref{sec:bounding_div_lemma}, and prove (2), which is the most challenging part, in \cref{sec:bad_events}.

In the following we fix distribution $\Unf$ with  $\Unf_X$ being a distribution over $\Uni^{m \times n}$ matrices with independent columns,  event  $W$ over $\Unf$ and  $\delta$-dense event family $\cE = \set{E_{i,j}}$  over $\Unf_X$. We let $\Idl = \Unf|W$ and let $\Rll_{X,J} = \Rll(\Unf,W,\cE)$ be the skewed variant of $\Idl$ defined in \cref{def:SkewedDistrbiution}.  Let   $Y_i = (Y_{i,1},\ldots,Y_{i,n})$ for $Y_{i,j}$ be the indicator for $E_{i,j}$, and let $\dval = \sum_{i=1}^{m} D(\Idl_{\bX_i \bY_i} || \Unf_{\bX_i \bY_i} | \Idl_{\bX_{<i}})$.

\subsection{Warmup}\label{sec:BoundingSmoothKL:Warmup}
In this section we give a rather detailed proof sketch (more accurately, an attempt
proof sketch) for \cref{thm:BoundingSmoothKL}. Specifically, we try to bound the divergence between
$\Idl$ and $\Rll$; That is, to show that
\begin{align}\label{eq:thm_proof:bounding-div}
	D(\Idl || \Rll) \leq O\paren{\frac{1}{\delta n}}\cdot (\dval + m)
\end{align}
 We try to do so by showing that for every $i \in [m]$ it holds that
\begin{align}\label{eq:thm_proof:bounding-div-round-j}
	D(\Idl_{\bX_i} || \Rll_{\bX_i} | \Idl_{\bX_{<i}}) \leq O\paren{\frac{1}{\delta n}}\cdot (\dval_i + 1)
\end{align}
for $\dval_i = D(\Idl_{\bX_i \bY_i} || \Unf_{\bX_i \bY_i} | \Idl_{\bX_{<i}})$, and applying chain-rule of KL-divergence for deducing \cref{eq:thm_proof:bounding-div}.  By data-processing of KL-divergence (\cref{fact:prelim:diver-properties}(\ref{fact:diver-properties:item:data-processing})), it holds that

\begin{align}\label{eq:thm_proof:div_ideal_real_round_i-first}
	D(\Idl_{\bX_i} || \Rll_{\bX_i} | \Idl_{\bX_{<i}}) \leq D(\Idl_{X_i Y_i} || \Rll'_{X_i Y_i} | \Idl_{X_{<i}}),
\end{align}
where 
\begin{align*}
	\Rll'_{X_i Y_i \mid X_{<i}} = \Idl_{X_i Y_i \mid X_{<i}, X_{i,J}, Y_{i,J}=1} \circ \Rll_{J, X_{i,J} \mid X_{<i}} \equiv \Unf_{X_{i,J} \mid X_{<i}} \Idl_{X_i Y_i \mid X_{<i}, X_{i,J}, Y_{i,J}=1} \circ \Rll_{J \mid X_{<i}}
\end{align*}
(note that $\Rll'_{X_i} \equiv \Rll_{X_i}$ and that $\Unf_{X_{i,J} \mid X_{<i}} \equiv \Unf_{X_{i,J} \mid X_{<i,J}}$ because the columns under $\Unf$ are independent). 
By definition of $\Rll'$, for any fixing of $\bx_{\leq i} \by_i \in \Supp(\Idl_{\bX_{\leq i} \bY_i})$ it holds that

\begin{align}\label{eq:thm_proof:real-prob}
	\Rll'_{\bX_i \bY_i \mid \bX_{<i} = \bx_{<i}}(\bx_i \by_i)
	&= \eex{j \sim \Rll_{J \mid \bX_{<i} = \bx_{<i}}}{\Unf_{X_{i,j} \mid \bX_{<i} = \bx_{<i}}(x_{i,j}) \cdot \Idl_{\bX_i \bY_i \mid \bX_{<i} = \bx_{<i}, X_{i,j}=x_{i,j}, Y_{i,j}=1}(\bx_i \by_i)}\\
	&= \sum_{j=1}^{n}{\Rll_{J \mid \bX_{<i} = \bx_{<i}}(j)\cdot  \Unf_{X_{i,j} \mid \bX_{<i} = \bx_{<i}}(x_{i,j}) \cdot \frac{\Idl_{\bX_i \bY_i X_{i,j} Y_{i,j} \mid \bX_{<i} = \bx_{<i}}(\bx_i \by_i x_{i,j} 1)}{\Idl_{X_{i,j}, Y_{i,j} \mid \bX_{<i} = \bx_{<i}}(x_{i,j},1)}}\nonumber\\
	&= \sum_{j \in 1_{\by_i}}{\Rll_{J \mid \bX_{<i} = \bx_{<i}}(j)\cdot  \Unf_{X_{i,j} \mid \bX_{<i} = \bx_{<i}}(x_{i,j}) \cdot \frac{\Idl_{\bX_i \bY_i \mid \bX_{<i} = \bx_{<i}}(\bx_i \by_i)}{\Idl_{X_{i,j}, Y_{i,j} \mid \bX_{<i} = \bx_{<i}}(x_{i,j},1)}}\nonumber\\
	&= \sum_{j \in 1_{\by_i}}{\Rll_{J \mid \bX_{<i} = \bx_{<i}}(j) \cdot \frac{\beta_{i,j}(x_{i,j})\cdot \Idl_{\bX_i \bY_i \mid \bX_{<i} = \bx_{<i}}(\bx_i \by_i)}{\tdelta_{i,j}}},\nonumber
\end{align}
for $\beta_{i,j}(x_{i,j}) = \beta_{i,j}(x_{i,j} ; \bx_{<i}) = \frac{\Unf_{X_{i,j} \mid \bX_{<i} = \bx_{<i}}(x_{i,j})}{\Idl_{X_{i,j} \mid \bX_{<i} = \bx_{<i}, Y_{i,j} = 1}(x_{i,j})}$ and $\tdelta_{i,j} = \tdelta_{i,j}(\bx_{<i}) = \Idl_{Y_{i,j} \mid \bX_{<i} = \bx_{<i}}(1)$ ($= \Idl[E_{i,j} \mid \bX_{<i} = \bx_{<i}]$), where recall that we denote $1_{y_i} = \set{j \in [n] \colon y_{i,j}=1}$. In addition, note that

\begin{align}\label{eq:thm_proof:real-I-prob-part1}
\Rll_{J \mid \bX_{<i} = \bx_{<i}}(j)
&= \frac{\Rll[\bX_{<i} = \bx_{<i} \mid J=j]\cdot \Rll[J=j]}{\Rll[\bX_{<i} = \bx_{<i}]}
= \frac{\Rll[\bX_{<i} = \bx_{<i} \mid J=j]\cdot \Rll[J=j]}{\sum_{t=1}^n \Rll[J=t] \Rll[\bX_{<i} = \bx_{<i} \mid J=t]}\\
&= 1/\paren{\sum_{t=1}^n \frac{\Rll[\bX_{<i} = \bx_{<i} \mid J=t]}{\Rll[\bX_{<i} = \bx_{<i} \mid J=j]}}.\nonumber
\end{align}

\noindent Since for all $t \in [n]$ it holds that

\begin{align}\label{eq:thm_proof:real-I-prob-part2}
	\lefteqn{\Rll[\bX_{<i} = \bx_{<i} \mid J=t] 
	= \prod_{s=1}^{i-1} \Unf[X_{s,t} = x_{s,t} \mid \bX_{<s}=\bx_{<s}]\cdot \Idl[\bX_s = \bx_s \mid \bX_{<s} = \bx_{<s}, X_{s,t} = x_{s,t}, E_{s,t}]}\nonumber\\
	&=   \prod_{s=1}^{i-1} \Unf[X_{s,t} = x_{s,t} \mid \bX_{<s}=\bx_{<s}]\cdot \frac{\Idl[X_{s,t}=x_{s,t}, E_{s,t} \mid \bX_{\leq s} = \bx_{\leq s}] \cdot \Idl[X_s = x_s \mid \bX_{<s} = \bx_{<s}]}{\Idl[E_{s,t} \mid \bX_{<s} = \bx_{<s}] \cdot \Idl[X_{s,t} = x_{s,t} \mid \bX_{<s} = \bx_{<s}, E_{s,t}]}\nonumber\\
	&= \prod_{s=1}^{i-1} 
	\frac{\Unf[X_{s,t} = x_{s,t} \mid \bX_{<s}=\bx_{<s}]}{\Idl[X_{s,t} = x_{s,t} \mid \bX_{<s} = \bx_{<s}, E_{s,t}]} \cdot \frac{\Idl[E_{s,t} \mid \bX_{\leq s} = \bx_{\leq s}]}{\Idl[E_{s,t} \mid \bX_{<s} = \bx_{<s}]} \cdot
	\Idl[X_s = x_s \mid \bX_{<s} = \bx_{<s}]
\end{align}

\noindent we deduce from \cref{eq:thm_proof:real-I-prob-part1,eq:thm_proof:real-I-prob-part2} that

\begin{align}\label{eq:thm_proof:real-I-prob}
\Rll_{J \mid \bX_{<i} = \bx_{<i}}(j) = \frac{\omega_{i,j}}{\sum_{t=1}^{n} \omega_{i,t}},
\end{align}
where
\begin{align*}
	\omega_{i,j} 
	&= \omega_{i,j}(\bx_{<i})\\
	&= \frac{n}{\sum_{t=1}^n \omega'_{i,t}} \cdot \prod_{s=1}^{i-1} 
	\frac{\Unf[X_{s,j} = x_{s,j} \mid \bX_{<s}=\bx_{<s}]}{\Idl[X_{s,j} = x_{s,j} \mid \bX_{<s} = \bx_{<s}, E_{s,j}]}\cdot 
	\frac{\Idl[E_{s,j} \mid \bX_{\leq s} = \bx_{\leq s}]}{\Idl[E_{s,j} \mid \bX_{<s} = \bx_{<s}]}\\
	&= \frac{n}{\sum_{t=1}^n \omega'_{i,t}} \cdot \prod_{s=1}^{i-1} 
	\frac{\Unf[X_{s,j} = x_{s,j} \mid \bX_{<s}=\bx_{<s}]}{\Idl[X_{s,j} = x_{s,j} \mid \bX_{<s} = \bx_{<s}]}\cdot 
	\frac{\Idl[E_{s,j} \mid \bX_{< s} = \bx_{< s}]}{\Idl[E_{s,j} \mid \bX_{<s} = \bx_{<s}, X_{s,j}=x_{s,j}]} \cdot
	\frac{\Idl[E_{s,j} \mid \bX_{\leq s} = \bx_{\leq s}]}{\Idl[E_{s,j} \mid \bX_{<s} = \bx_{<s}]}\\
	&= \frac{n \cdot \omega'_{i,j}}{\sum_{t=1}^n \omega'_{i,t}} \cdot \prod_{s=1}^{i-1} 
	\frac{\Idl[E_{s,j} \mid \bX_{< s} = \bx_{< s}]}{\Idl[E_{s,j} \mid \bX_{<s} = \bx_{<s}, X_{s,j}=x_{s,j}]}\cdot 
	\frac{\Idl[E_{s,j} \mid \bX_{\leq s} = \bx_{\leq s}]}{\Idl[E_{s,j} \mid \bX_{<s} = \bx_{<s}]}
\end{align*}
for $\omega'_{i,j} = \omega'_{i,j}(\bx_{<i}) = \prod_{s=1}^{i-1} \frac{\Unf[X_{s,j} = x_{s,j} \mid \bX_{<s}=\bx_{<s}]}{\Idl[X_{s,j} = x_{s,j} \mid \bX_{<s} = \bx_{<s}]}$.
Note that $\omega_{i,j}$ is basically a relative ``weight'' for the column $j$, where a large $\omega_{i,j}$ with respect to the other $\omega_{i,t}$'s means that $\Rll_{J \mid \bX_{<i} = \bx_{<i}}(j)$ is higher. In an extreme case it is possible that $\omega_{i,j} = \infty$, meaning that $\Rll_{J \mid \bX_{<i} = \bx_{<i}}(j) = 1$. However, we assume for now that all $\omega_{i,j} < \infty$. Later in this proof attempt we even assume that all the terms are close to $1$, meaning that $Q_{J | X_{<i}=x_{<i}}$ has high min entropy (assumptions that are eliminated in \cref{sec:ConditionalDits}). As a side note, observe that $\omega_{1,j} = 1$ for all $j \in [n]$ (meaning that $Q_J$ is the uniform distribution over $[n]$). At this point, we just mention that we added (the same) multiplicative factor of $\frac{n}{\sum_{t=1}^n \omega'_{i,t}}$ to all $\set{\omega_{i,j}}_{j=1}^n$. On the one hand this does not change the relative weight, but on the other hand it will help us to claim in the coming sections that these $\omega_{i,j}$'s are indeed close to $1$.
By \cref{eq:thm_proof:div_ideal_real_round_i-first,eq:thm_proof:real-prob,eq:thm_proof:real-I-prob}, it holds that

\begin{align}\label{eq:thm_proof:div_ideal_real_round_i}
	D(\Idl_{\bX_i} || \Rll_{\bX_i} | \Idl_{\bX_{<i}})
	&\leq D(\Idl_{\bX_i \bY_i} || \Rll'_{\bX_i \bY_i} | \Idl_{\bX_{<i}})\\
	&= \Ex_{\bx_{<i}\sim \bX_{<i}}\eex{\bx_i \by_i \sim \Idl_{\bX_i \bY_i \mid \bX_{<i} = \bx_{<i}}}{\log \frac{\Idl_{\bX_i \bY_i \mid \bX_{<i} = \bX_{<i}}(\bx_i \by_i)}{\Rll_{\bX_i \bY_i \mid \bX_{<i} = \bX_{<i}}(\bx_i \by_i)}}\nonumber\\
	&= \Ex_{\bx_{<i}\sim \bX_{<i}}\eex{\bx_i \by_i \sim \Idl_{\bX_i \bY_i \mid \bX_{<i} = \bx_{<i}}}{\log \frac{\sum_{j=1}^{n} \omega_{i,j}}{\sum_{j \in 1_{\by_i}}{\frac{\omega_{i,j}\cdot \beta_{i,j}(x_{i,j})}{\tdelta_{i,j}}}}}\nonumber\\
	&= \Ex_{\bx_{<i}\sim \bX_{<i}}\eex{\bx_i \by_i \sim \Idl_{\bX_i \bY_i \mid \bX_{<i} = \bx_{<i}}}{-\log \paren{1 + \gamma_i(\bx_i\by_i)}},\nonumber
\end{align}
for
\begin{align}\label{eq:thm_proof:gamma-def}
	\gamma_i(\bx_i\by_i) = \gamma_i(\bx_i\by_i;\bx_{<i}) =\paren{\sum_{j \in 1_{\by_i}}{\frac{\omega_{i,j}\cdot \beta_{i,j}(x_{i,j})}{\tdelta_{i,j}}}} / \paren{\sum_{j=1}^{n} \omega_{i,j}} - 1
\end{align}

Naturally, we would like to approximate the logarithm in the above equation with a low-degree
polynomial. However, we can only do  if $\gamma_i$ is far away from $-1$. In particular, if $\Idl[\gamma_i(\bX_i\bY_i;\bX_{<i}) = -1] > 0$ (which happens if the event $W$ allows for none of the events $\set{E_{i,j}}_{i=1}^n$ to occur), the above expectation is unbounded. At that point, we only show how to bound \cref{eq:thm_proof:div_ideal_real_round_i} under simplifying assumptions, while in  \cref{sec:ConditionalDits} we present how to eliminate the assumptions via smooth KL-divergence. We now assume that for any $\bx_{<i} \in \Supp(\Idl_{\bX_{<i}})$ and any $j \in [n]$, the following holds:

\begin{assumption}\label{title:assumptions}~
\begin{enumerate}
	\item $\size{\gamma_i(\bx_i \by_i)} \leq 1/2$ for any $\bx_i \by_i \in \Supp(\Idl_{\bX_i \bY_i \mid \bX_{<i} = \bx_{<i}})$.\label{cond:gamma}
	\item $\tdelta_{i,j} \geq  0.9 \delta_{i,j}$ (recall that $\delta_{i,j} = \Unf[E_{i,j}] = \Unf[E_{i,j} \mid X_{\leq i}]$ for any fixing of $X_{\leq i}$).\label{cond:pval}
	\item $\omega_{i,j} \in 1 \pm 0.1$.\label{cond:aval}
	\item $\Supp(\Unf_{X_{i,j} \mid \bX_{<i}=\bx_{<i}}) \subseteq \Supp(\Idl_{X_{i,j} \mid \bX_{<i} = \bx_{<i}, Y_{i,j}=1})$.\label{cond:support}
	\item $\beta_{i,j}(x_{i,j}) \leq 1.1$ for any $x_{i,j} \in \Supp(\Idl_{X_{i,j} \mid \bX_{<i} = \bx_{<i}})$.\label{cond:bval}
\end{enumerate}
\end{assumption}

\noindent Note that Assumption~\ref{cond:aval} implies that $\Rll_{J \mid X_{<i}}$ has high min-entropy, and Assumptions \ref{cond:pval} along with \ref{cond:bval} imply that for all $j$:
\begin{align*}
	\lefteqn{\Unf[W \mid (X_{<i},X_{i,j}) = (x_{<i},x_{i,j}), E_{i,j}]}\\
	&= \frac{\Idl_{X_{i,j} \mid X_{<i}=x_{<i}, E_{i,j}}(x_{i,j})}{\Unf_{X_{i,j} \mid X_{<i}=x_{<i}, E_{i,j}}(x_{i,j})}\cdot \frac{\Idl[E_{i,j} \mid X_{<i}=x_{<i}]}{\Unf[E_{i,j} \mid X_{<i}=x_{<i}]}\cdot \Unf[W \mid X_{<i}=x_{<i}]\\
	&= \beta_{i,j}(x_{i,j})\cdot \paren{\tdelta_{i,j} / \delta_{i,j}} \cdot \Unf[W \mid X_{<i}=x_{<i}]
	\geq \Unf[W \mid X_{<i}=x_{<i}]/2,
\end{align*}
which fits the explanation in \cref{sec:intro:Tech} (note that in the second equality we used the fact that $\Unf_{X_{i,j} \mid X_{<i}=x_{<i}, E_{i,j}}(x_{i,j}) = \Unf_{X_{i,j} \mid X_{<i}=x_{<i}}(x_{i,j})$ by assumption).
 
\noindent By \cref{eq:thm_proof:div_ideal_real_round_i}, note that in order to prove \cref{eq:thm_proof:bounding-div-round-j}, it is enough to show that for any $\bx_{<i} \in \Supp(\Idl_{\bx_{<i}})$ it holds that
\begin{align}\label{eq:thm_proof:div_1_plus_delta_perfect}
	&\\
	&\eex{\bx_i \by_i \sim \Idl_{\bX_i \bY_i \mid \bX_{<i} = \bx_{<i}}}{-\log \paren{1 + \gamma_i(\bx_i\by_i)}}
	\leq O\paren{\frac1{\delta n}}\cdot \paren{D(\Idl_{\bX_i \bY_i \mid \bX_{<i} = \bx_{<i}} || \Unf_{\bX_i \bY_i \mid \bX_{<i} = \bx_{<i}}) + 1}\nonumber
\end{align} 
In the following, fix $\bx_{<i} \in \Supp(\Idl_{\bx_{<i}})$. We now focus on proving \cref{eq:thm_proof:div_1_plus_delta_perfect}. Using the inequality $-\log (1+x) \leq -x + x^2$ for $\size{x} \leq \frac12$, we deduce from Assumption~\ref{cond:gamma} that
\begin{align}\label{eq:thm_proof:delta_plus_delta2_square}
	\eex{\bx_i \by_i \sim \Idl_{\bX_i \bY_i \mid \bX_{<i} = \bx_{<i}}}{-\log \paren{1 + \gamma_i(\bx_i\by_i)}}
	\leq \eex{\bx_i \by_i \sim \Idl_{\bX_i \bY_i \mid \bX_{<i} = \bx_{<i}}}{-\gamma_i(\bx_i\by_i) + \gamma_i(\bx_i\by_i)^2}
\end{align}

\noindent Note that
\begin{align}\label{eq:thm_proof:exp_gamma_internal}
	\lefteqn{\eex{\bx_i \by_i \sim \Idl_{\bX_i \bY_i \mid \bX_{<i} = \bx_{<i}}}{\sum_{j \in 1_{\by_i}} \frac{\omega_{i,j}\cdot \beta_{i,j}(x_{i,j})}{\tdelta_{i,j}}}
	= \sum_{j=1}^{n} \eex{x_{i,j} y_{i,j} \sim \Idl_{X_{i,j} Y_{i,j} \mid \bX_{<i} = \bx_{<i}}}{y_{i,j}\cdot \frac{\omega_{i,j}\cdot \beta_{i,j}(x_{i,j})}{\tdelta_{i,j}}}}\\
	&= \sum_{j=1}^{n} \omega_{i,j}\cdot \eex{x_{i,j} \sim \Idl_{X_{i,j} \mid \bX_{<i} = \bx_{<i}, Y_{i,j} = 1}}{\beta_{i,j}(x_{i,j})}
	= \sum_{j=1}^{n} \omega_{i,j}\cdot \eex{x_{i,j} \sim \Idl_{X_{i,j} \mid \bX_{<i} = \bx_{<i}, Y_{i,j} = 1}}{\frac{\Unf_{X_{i,j} \mid \bX_{<i} = \bx_{<i}}(x_{i,j})}{\Idl_{X_{i,j} \mid \bX_{<i} = \bx_{<i}, Y_{i,j} = 1}(x_{i,j})}}\nonumber\\
	&= \sum_{j=1}^{n} \omega_{i,j}\cdot \Unf_{X_{i,j} \mid \bX_{<i}=\bx_{<i}}(\Supp(\Idl_{X_{i,j} \mid \bX_{<i} = \bx_{<i}, Y_{i,j}=1}))
	= \sum_{j=1}^{n} \omega_{i,j}.\nonumber
\end{align}
The second equality holds since $y_{i,j} \in \zo$ and since Assumption~\ref{cond:pval} implies that $\Idl_{Y_{i,j} \mid \bX_{<i} = \bx_{<i}}(1) = \tdelta_{i,j}  > 0$ for all $j \in [n]$, and the last equality holds by Assumption~\ref{cond:support}. Therefore, we deduce from \cref{eq:thm_proof:exp_gamma_internal} that

\begin{align}\label{eq:thm_proof:exp_gamma}
	&\\
	&\eex{\bx_i \by_i \sim \Idl_{\bX_i \bY_i \mid \bX_{<i} = \bx_{<i}}}{\gamma_i(\bx_i\by_i)}
	= \paren{\eex{\bx_i \by_i \sim \Idl_{\bX_i \bY_i \mid \bX_{<i} = \bx_{<i}}}{\sum_{j \in 1_{\by_i}}{\frac{\omega_{i,j}\cdot \beta_{i,j}(x_{i,j})}{\tdelta_{i,j}}}}} / \paren{\sum_{j=1}^{n} \omega_{i,j}} - 1
	= 0.\nonumber
\end{align}

Hence, in order to prove \cref{eq:thm_proof:div_1_plus_delta_perfect}, we deduce from \cref{eq:thm_proof:delta_plus_delta2_square,eq:thm_proof:exp_gamma} that it is left to prove that

\begin{align}\label{eq:thm_proof:exp_gamma_square}
	\eex{\bx_i \by_i \sim \Idl_{\bX_i \bY_i \mid \bX_{<i} = \bx_{<i}}}{\gamma_i(\bx_i\by_i)^2}
	\leq O\paren{\frac1{\delta n}}\cdot \paren{D(\Idl_{\bX_i \bY_i \mid \bX_{<i} = \bx_{<i}} || \Unf_{\bX_i \bY_i \mid \bX_{<i} = \bx_{<i}}) + 1}
\end{align}

\noindent In the following, rather than directly bounding the expected value of $\gamma_i(\bx_i\by_i)^2$ under $\Idl_{\bX_i \bY_i \mid \bX_{<i} = \bx_{<i}}$, we show that under the product of the marginals of $\Idl_{\bX_i \bY_i \mid \bX_{<i} = \bx_{<i}}$ (namely, under the distribution $\prod_{j=1}^n \Idl_{X_{i,j} Y_{i,j} \mid \bX_{<i} = \bx_{<i}}$), the value of $\gamma_i(\bx_i\by_i)$ is well concentrated around its mean (\ie zero), and the proof will follow by \cref{prop:prelim:sub-exp-to-divergence}. More formally, let $\Gamma$ be the value of $\gamma_i(\bx_i\by_i)$ when $\bx_i\by_i$ is drawn from either $\Idl = \Idl_{\bX_i \bY_i \mid \bX_{<i} = \bx_{<i}}$ or $\Idl^{\Pi} = \prod_{j=1}^n \Idl_{X_{i,j} Y_{i,j} \mid \bX_{<i} = \bx_{<i}}$. 
We prove that there exist two constants $K_1, K_2 > 0$ such that for any $\gamma \in [0,1]:$

\begin{align}\label{eq:thm_proof:goal_under_product}
	\Idl^{\Pi}[\size{\Gamma} \geq \gamma] \leq K_2\cdot \exp\paren{-\frac{\gamma^2}{K_1\cdot \sigma^2}}
\end{align}
for $\sigma^2 = 1/\delta n$. Using \cref{eq:thm_proof:goal_under_product} and the fact that $\size{\Gamma} \leq 1$ (Assumption~\ref{cond:gamma}), \cref{prop:prelim:sub-exp-to-divergence} yields that
\begin{align}
	\eex{\bx_i \by_i \sim \Idl_{\bX_i \bY_i \mid \bX_{<i} = \bx_{<i}}}{\gamma_i(\bx_i\by_i)^2}
	&= \eex{\Idl}{\Gamma^2}
	\leq \frac{K_3}{\delta n}\cdot \paren{D(\Idl || \Idl^{\Pi}) + 1}\\
	&= \frac{K_3}{\delta n}\cdot \paren{D(\Idl_{\bX_i \bY_i \mid \bX_{<i} = \bx_{<i}} || \prod_{j=1}^n \Idl_{X_{i,j} Y_{i,j} \mid \bX_{<i} = \bx_{<i}}) + 1}\nonumber\\
	&\leq \frac{K_3}{\delta n}\cdot \paren{D(\Idl_{\bX_i \bY_i \mid \bX_{<i} = \bx_{<i}} || \Unf_{\bX_i \bY_i \mid \bX_{<i} = \bx_{<i}}) + 1}.\nonumber
\end{align}
The last inequality holds by chain rule of KL-divergence when the right-hand side distribution is product (\cref{fact:prelim:diver-properties}(\ref{fact:diver-properties:item:chain-rule}), where recall that $\Unf_{\bX_i \bY_i \mid \bX_{<i} = \bx_{<i}} = \prod_{j=1}^n \Unf_{\bX_{i,j} \bY_{i,j} \mid \bX_{<i} = \bx_{<i}}$). This concludes the proof of \cref{eq:thm_proof:exp_gamma_square}.
It is left to prove \cref{eq:thm_proof:goal_under_product}. In the following, given $\bx_i \by_i$ which are drawn from either $\Idl^{\Pi} = \prod_{j=1}^n \Idl_{X_{i,j} Y_{i,j} \mid \bX_{<i} = \bx_{<i}}$ or $\Idl^{\Pi'} = \prod_{j=1}^n \Idl_{Y_{i,j} \mid \bX_{<i} = \bx_{<i}} \cdot \Idl_{X_{i,j} \mid \bX_{<i} = \bx_{<i}, Y_{i,j}=1}$, we define the random variables  $L_j$,$Z_j$,$L$ and $Z$ (in addition to $\Gamma$), where $L_j$ is the value of $\omega_j\cdot \beta_j(x_{i,j})$, $L = \sum_{j=1}^{n} L_j$, $Z_j = \begin{cases} L_j / \tdelta_j & y_{i,j}=1 \\ 0 & y_{i,j}=0 \end{cases}$ and $Z = \sum_{j=1}^{n} Z_j$, letting $\omega_j = \omega_{i,j}$, $\beta_j(\cdot) = \beta_{i,j}(\cdot)$ and $\tdelta_j = \tdelta_{i,j}$. Note that by definition, $Z = (1 + \Gamma) \mu$ for $\mu = \sum_{j=1}^{n} \omega_j$. Namely, $\Gamma$ measures how far $Z$ is from its expected value $\mu$ (follows by \cref{eq:thm_proof:exp_gamma_internal} that calculates $\eex{\Idl}{Z}$, which also equals to $\eex{\Idl^{\Pi}}{Z}$ and $\eex{\Idl^{\Pi'}}{Z}$). Note that the distribution of $Z$ and $\Gamma$ when $\bx_i \by_i$ is drawn from $\Idl^{\Pi}$ is identical to the distribution of $Z$ and $\Gamma$ (respectively) when $\bx_i \by_i$ is drawn from $\Idl^{\Pi'}$. Therefore, in particular it holds that
\begin{align}\label{eq:thm_proof:Pi_equals_Pi-prime}
	\Idl^{\Pi}[\size{\Gamma} \geq \gamma] = \Idl^{\Pi'}[\size{\Gamma} \geq \gamma]
\end{align}
Under $\Idl^{\Pi'}$,  the $L_j$'s are independent random variables with $\eex{\Idl^{\Pi'}}{L_j} = \omega_j$ and $\eex{\Idl^{\Pi'}}{L} = \mu$ where $\mu = \sum_{j=1}^{n} \omega_j \geq n/2$ and $\size{L_j} \leq 2$ (by Assumptions~\ref{cond:aval} and \ref{cond:bval}). Moreover, for all $j \in [n]$, $Z_j = (L_j / \tdelta_j) \cdot \Bern(\tdelta_j)$ where $\tdelta_j \geq 0.9 \delta_{i,j} \geq 0.9 \delta$ (by Assumption~\ref{cond:pval}). Hence, \cref{fact:prelim:L_i_and_Y_i} yields that
\begin{align}\label{eq:thm_proof:calc_under_product}
	\Idl^{\Pi'}[\size{\Gamma} \geq \gamma] \leq 4 \exp\paren{-\frac{\delta n \gamma^2}{100}}
\end{align}

The proof of \cref{eq:thm_proof:goal_under_product} now follows by \cref{eq:thm_proof:calc_under_product,eq:thm_proof:Pi_equals_Pi-prime}, which ends the proof of \cref{thm:BoundingSmoothKL} under the assumptions in \ref{title:assumptions}. 

\subsubsection{Eliminating the Assumptions}\label{sec:eliminating-the-assum}

The assumptions we made in \ref{title:assumptions} may seem unjustified at first glance. For instance, even for $j=1$, there could be ``bad'' columns $j \in [n]$ with $\tdelta_{1,j} < 0.9 \delta_{1,j}$\remove{, and there could be  ``bad'' values of $x_1 \in \Supp(\Idl_{X_1})$ with $\size{\gamma_1(x_1)} > \frac12$}. We claim, however,  that the probability that a uniform $J$ (chosen by $\Rll$) will hit such a ``bad'' column $j$ is low\remove{, and conditioned on ``$J$ is good'', the probability that $x_1 \sim \Idl_{X_1}$ will hit a ``bad'' value is also low}. For showing that, let $\cB_1 = \set{j\in [n] \colon \tdelta_{1,j} < 0.9 \delta_{1,j}}$ be the set of ``bad'' columns $j \in [n]$ for $i=1$. A simple calculation yields that
\begin{align*}
d_1 
&= D(\Idl_{X_1 Y_1} || \Unf_{X_1 Y_1}) \geq D(\Idl_{Y_1} || \Unf_{Y_1}) \geq \sum_{j=1}^{n} D(\Idl_{Y_{1,j}} || \Unf_{Y_{1,j}})\\
&= \sum_{j=1}^{n} D(\tdelta_{1,j} || \delta_{1,j}) \geq \sum_{j \in \cB_1} D(\tdelta_{1,j} || \delta_{1,j}) \geq \sum_{j \in \cB_1} \delta_{1,j}/200 \geq \size{\cB_1} \cdot \delta/200.
\end{align*}
The second inequality holds by chain-rule of KL-divergence when the right-hand side distribution is product (\cref{fact:prelim:diver-properties}(\ref{fact:diver-properties:item:chain-rule}))) and the penultimate inequality  holds by \cref{fact:prelim:bernoulli-div-est}(\ref{fact:prelim:bernoulli-div-est:minus}). This implies that $\size{\cB_1} \leq 200 d_1/ \delta$, and hence, $Q_{J}[J \in \cB_1] < 200 d_1/ (\delta n)$. 
Extending the above argument for a row $i > 1$ is a much harder task. As we saw in \cref{eq:thm_proof:real-I-prob}, the conditional distribution $Q_{J \mid X_{<i}}$ is much more complicated, and it also seems not clear how to bound $\size{\cB_i}$ (now a function of $X_{<i}$) as we did for $i=1$, when $X_{<i}$ is drawn from $\Rll$. Yet, we show in the next sections that when $X_{<i}$ is drawn from $\Idl$ (and not from $\Rll$), then we are able to understand $Q_{J \mid X_{<i}}$ and $\cB_i(X_{<i})$ better and bound by $O(d/(\delta n))$ the probability of hitting a ``bad'' column for all $i \in [m]$. This is done by relating martingale sequences for each sequence $\set{\omega_{i,j}}_{i=1}^m$ under $\Idl$, and by showing (using \cref{lemma:prelim:Martingales-new-bound}) that with high probability, the sequences of most $j \in [n]$ remain around $1$.

\subsection{The Conditional Distributions}\label{sec:ConditionalDits}
Following the above discussion, the high level plan of  our proof is to define the ``good'' events $A_1,\ldots,A_n$ for $\Idl$ and $B_1,\ldots,B_n$ for $\Rll$ such that for all $i \in [m]$, the conditional distributions $\Idl_{X_i \mid A_{\leq i}}$ and $\Rll_{X_i \mid B_{\leq i}}$ satisfies the assumptions in \ref{title:assumptions}. Then, by only bounding the probability of ``bad'' events under $\Idl$, the proof of \cref{thm:BoundingSmoothKL} will follow by \cref{lemma:smooth-div-main-prop}. We start with notations.

\begin{notation}\label{def:ConditionalDits:definitions}~
\begin{itemize}
	
	\item $\omega'_{i,j} = \omega'_{i,j}(\bx_{<i}) = \prod_{s=1}^{i-1} \frac{\Unf_{X_{s,j} \mid \bX_{<s}=\bx_{<s}}(x_{s,j})}{\Idl_{X_{s,j} \mid \bX_{<s}=\bx_{<s}}(x_{s,j})}$.
	
	\item $\omega_{i,j} = \omega_{i,j}(\bx_{<i}) =  
	\frac{n \cdot \omega'_{i,j}}{\sum_{t=1}^n \omega'_{i,t}} \cdot \prod_{s=1}^{i-1} 
	\frac{\Idl[E_{s,j} \mid \bX_{< s} = \bx_{< s}]}{\Idl[E_{s,j} \mid \bX_{<s} = \bx_{<s}, X_{s,j}=x_{s,j}]}\cdot 
	\frac{\Idl[E_{s,j} \mid \bX_{\leq s} = \bx_{\leq s}]}{\Idl[E_{s,j} \mid \bX_{<s} = \bx_{<s}]}$
	
	\item $\beta_{i,j}(x_{i,j}) = \beta_{i,j}(x_{i,j} ; \bx_{<i}) = \Unf_{X_{i,j} \mid \bX_{<i} = \bx_{<i}}(x_{i,j}) / \Idl_{X_{i,j} \mid \bX_{<i} = \bx_{<i}, E_{i,j}}(x_{i,j})$
	
	\item $\tdelta_{i,j} = \tdelta_{i,j}(\bx_{<i}) = \Idl[E_{i,j} \mid \bX_{<i} = \bx_{<i}]$
	
	
	\item $\cX_{i,j} = \cX_{i,j}(\bx_{<i}) = \set{x_{i,j} \in \Supp(\Unf_{X_{i,j} \mid X_{<i}=x_{<i}}) \colon \beta_{i,j}(x_{i,j}) \leq 1.1}$.
	
	\item $\cJ_i = \cJ_i(\bx_{<i}) = \set{j \in [n] \colon \paren{\tdelta_{i,j} \geq 0.9 \delta_{i,j}} \land \paren{\omega_{i,j} \in 1 \pm 0.1} \land \paren{\Unf_{X_{i,j} \mid \bX_{<i} = \bx_{<i}}(\cX_{i,j}) \geq 0.9}}$.
	
	\item $\cG_i(\bx_{i}) = \cG_i(\bx_{i};\bx_{<i}) = \set{j \in [n] \colon \bigwedge_{s=1}^{i} \paren{j \in \cJ_s \land x_{s,j} \in \cX_{s,j}}}$, letting $\cG_0 = [n]$.
	
	\item $\cs_i = \cs_i(\bx_{<i}) = \cG_{i-1} \bigcap \cJ_i$.
	
	\item $\beta'_{i,j}(x_{i,j}) = \beta'_{i,j}(x_{i,j} ; \bx_{<i}) = \beta_{i,j} \cdot \indic{x_{i,j} \in \cX_{i,j}}$.
	
	\item $\gamma_i(\bx_i\by_i) = \gamma_i(\bx_i\by_i;\bx_{<i}) = \paren{\sum_{j \in \cs_{i} \cap 1_{y_i}}{\frac{\omega_{i,j}\cdot \beta'_{i,j}(x_{i,j})}{\tdelta_{i,j}}}}/\paren{\sum_{j \in \cs_{i}} \omega_{i,j} \cdot \Unf_{X_{i,j} \mid \bX_{<i} = \bx_{<i}}\paren{\cX_{i,j}}} - 1$.
\end{itemize}
\end{notation}

\begin{definition}[Events]\label{def:ConditionalDits:Events}
The event  $B_i$ is defined over $\Rll_{X,J}$ by $B_i$:   $J \in \cG_i(\bX_{\leq i})$.

The following events are defined over $\Unf = \Unf_{\cdot,X}$: 	
\begin{itemize}
	
	\item  $G_i$:   $\size{\cs_{i}(X_{<i})} \geq 0.9 n$. 
	
	\item  $T_i$:  $\size{\gamma_i(\bX_i\bX_i ;\bX_{<i})} \leq 1/2$.
	
	\item $\pT_i$: $\Idl[T_{i} \mid X_{<i}] \geq 1 - 1/n$.

	\item  $A_i = G_i \land T_i \land \pT_i$.
	
	\item $\tB_i$:   $\Bern(\Rll[B_i \mid X_{<i}, B_{<i}])=1$.
	
	(\ie a coin that takes one with probability $\Rll[B_i \mid X_{<i}, B_{<i}]$ is flipped and its outcome is one).
	
	\item $C_i = A_i \land \tB_i$.
	
	\end{itemize}

\end{definition}

A few words about these definitions are in order. For $i \in [m]$, the set $\cG_i(x_i)$ is basically the set of all columns $j \in [n]$ that are ``good'' for all rows $s \in [i]$ (in a sense that all values of $\tdelta_{s,j}$, $\beta_{s,j}$, $\omega_{s,j}$ are bounded as we would like), and the set $\cs_i$ is the set of all (potential) ``good'' columns \wrt the history $x_{<i}$ (\ie $\tdelta_{s,j}$, $\omega_{s,j}$ are bounded for all $s \in [i]$, but $\beta_{s,j}$ are only bounded for $s \in [i-1]$). $A_i$ is the event (over $\Idl$) that we have large number of potential good columns for the row $i$ (described by the event $G_i$), and that $\size{\gamma_{i}}$, the term that will appear in the analysis, is promised to be small (described by the event $T_i$). $B_i$ is the event (over $\Rll$) that $J$ is ``good'' for all rows in $[i]$.

The proof of \cref{thm:BoundingSmoothKL} follows by the following two lemmatas and \cref{lemma:smooth-div-main-prop}.

\begin{lemma}[Bounding  KL-divergence of conditional distributions]\label{lemma:bounding-div-lemma}
	Let $\Unf,\Idl,\Rll,W,\cE, Y, \delta, \dval$ as defined in \cref{thm:BoundingSmoothKL}, and let $\set{A_i}_{i=1}^m$, $\set{B_i}_{i=1}^m$ and $\set{T_i}_{i=1}^m$ be the events defined in \cref{def:ConditionalDits:Events}. Assuming that $\Idl[T_1 \land \ldots \land T_n] \geq 1/2$, then for every $i \in [m]$ it holds that
	\begin{align*}
		D(\Idl_{X_i | A_{\leq i}} || \Rll_{X_i | B_{\leq i}} \mid \Idl_{X_{<i} | C_{\leq i}}) \leq \frac{c}{\delta n} (\dval_i + 1) \cdot \frac1{\Idl[C_{\leq i}]}
	\end{align*}
	for some universal constant $c > 0$, and  $\dval_i = D(\Idl_{\bX_i \bY_i} || \Unf_{\bX_i \bY_i} | \Idl_{\bX_{<i}})$.
\end{lemma}

\def\badEventsLemma{
	Let $\Unf,\Idl,\Rll,W,\cE, Y, \delta, \dval$ as defined in \cref{thm:BoundingSmoothKL}, and let $\set{C_i}_{i=1}^m$ be the events defined in \cref{def:ConditionalDits:Events}. Then there exists a universal constant $c>0$ such that if $n \geq c \cdot m/\delta$ and $\dval \leq \delta n/c$, then 
	\begin{align*}
		\Idl[C_1 \land \ldots \land C_m] \geq 1 - c\cdot (\dval + 1)/\delta n.
	\end{align*}
}

\begin{lemma}[Bounding  probability of bad events under $\Idl$]\label{lemma:bad-events}
	\badEventsLemma
\end{lemma}

\paragraph{Proving \cref{thm:BoundingSmoothKL}.}
\begin{proof}[Proof of \cref{thm:BoundingSmoothKL}]
	We start by setting the constant of \cref{thm:BoundingSmoothKL} to $c = 4\cdot \max\set{c_1,c_2+1}$ where $c_1$ is the constant from \cref{lemma:bounding-div-lemma} and $c_2$ is the constant from \cref{lemma:bad-events}. 
	\remove{
	
	First. note that 
	\begin{align}\label{eq:d_i-is-small}
		d = \sum_{i=1}^m D(\Idl_{X_i Y_i} || \Unf_{X_i Y_i}) \leq m \cdot \log \frac1{\Unf[W]} \leq m \delta n/1000,
	\end{align}
	where in the first inequality holds by \cref{fact:kl-divergence:diver-cond-on-event} and the second one holds by the assumption about $\Unf[W]$.
}
	By \cref{lemma:bad-events}\remove{(\ref{lemma:bad-events:C})} it holds that 
	\begin{align}
		\Idl[C_1 \land \ldots \land C_m]
		&\geq 1 - (c_2 + 1) \cdot (\dval + 1)/\delta n\label{eq:C_1-to-C_m-actual-bound}\\
		&\geq 1/2\label{eq:C_1-to-C_m-half-bound},
	\end{align}
	 the last inequality holds by the assumption on $n$ and $\dval$. In particular, it holds that
	\begin{align}\label{eq:prob-T_i-above-half}
		\Idl[T_1 \land \ldots \land T_m] \geq 1/2
	\end{align}
	Therefore, by (\ref{eq:prob-T_i-above-half}) and \cref{lemma:bounding-div-lemma} it holds that
	\begin{align}\label{eq:div-bound-per-i}
		D(\Idl_{X_i | A_{\leq i}} || \Rll_{X_i | B_{\leq i}} \mid \Idl_{X_{<i} | C_{\leq i}})\nonumber
		&\leq \frac{c_1}{\delta n} (\dval_i + 1) \cdot \frac1{\Idl[C_{\leq i}]}\\
		&\leq \frac{c_1}{\delta n} (\dval_i + 1) \cdot \frac1{\Idl[C_1 \land \ldots \land C_m]}\nonumber\\
		&\leq \frac{c}{\delta n} (\dval_i + 1),
	\end{align}
	 the last inequality holds by \cref{eq:C_1-to-C_m-half-bound}. The proof now holds by \cref{eq:C_1-to-C_m-actual-bound,eq:div-bound-per-i,lemma:smooth-div-main-prop}.
\end{proof}

In addition, the proof of \cref{lemma:ideal-running-time} now follows by \cref{lemma:bad-events}.
\paragraph{Proving \cref{lemma:ideal-running-time}}

\begin{corollary}[Restatement of \cref{lemma:ideal-running-time}]\label{cor:rest-ideal-running-time}
	Let $\Unf,\Idl,\Rll,W,\cE,\delta,\dval$ be as in \cref{thm:BoundingSmoothKL}, let $c$ be the constant from \cref{lemma:bad-events}, let $t > 0$ and let
	\begin{align*}
		p_t \eqdef \ppr{x \sim \Idl_X \; ; \; j \sim \Rll_{J|X=x}}{\exists i \in [m]: \Unf[W \mid (X_{<i},X_{i,j})=(x_{<i}, x_{i,j}),E_{i,j}] < \frac{\Unf[W]}{t}}
	\end{align*}
	Assuming $n \geq c \cdot m/\delta$ and $d \leq \delta n / c$, then
	\begin{align*}
		p_t \leq \frac{2m}{t} + \frac{c(\dval+1)}{\delta n}.
	\end{align*}
\end{corollary}
\begin{proof}	
	Let $\cG_{m}, \tdelta_{i,j}, \beta_{i,j}$ be according to \cref{def:ConditionalDits:definitions}. Observe that for any fixing of $x \in \Supp(\Idl_X)$ and any $j \in \cG_{m}(x)$, the following holds for all $i \in [m]$:
	\begin{align}\label{eq:running-time-cor:Unf[W]}
	\Unf[W \mid (X_{<i},X_{i,j})=(x_{<i}, x_{i,j}),E_{i,j}]
	&= \frac{\Idl[E_{i,j} \mid X_{<i}=x_{<i}]}{\Unf[E_{i,j} \mid X_{<i}=x_{<i}]} \cdot \frac{\Idl_{X_{i,j} \mid  X_{<i}=x_{<i}, E_{i,j}}(x_{i,j})}{\Unf_{X_{i,j} \mid  X_{<i}=x_{<i}, E_{i,j}}(x_{i,j})} \cdot \Unf[W \mid X_{<i} = x_{<i}]\nonumber\\
	&= \frac{\tdelta_{i,j}(x_{<i})}{\delta_{i,j}} \cdot \beta_{i,j}(x_{i,j} ; x_{<i}) \cdot \Unf[W \mid X_{<i} = x_{<i}]\nonumber\\
	&\geq \Unf[W \mid X_{<i} = x_{<i}]/2,
	\end{align}
	where second equality holds since
	\begin{align*}
	\Unf_{X_{i,j} \mid  X_{<i}=x_{<i}, E_{i,j}}(x_{i,j}) 
	&= \frac{\Unf[E_{i,j} \mid X_{\leq i}=x_{\leq i}]\cdot \Unf[X_{i,j}=x_{i,j} \mid X_{<i}=x_{<i}]}{\Unf[E_{i,j}]}\\
	&= \Unf_{X_{i,j} \mid  X_{<i}=x_{<i}}(x_{i,j}),
	\end{align*}
	(recall that by assumption, $\Unf[E_{i,j} \mid X_{\leq i}=x_{\leq i}] = \Unf[E_{i,j}]$ for any fixing of $x_{\leq i}$), and
	the inequality holds since $j \in \cG_{m}(x)$. Let $\set{\tB_i}, \set{C_i}$ be the events from \cref{def:ConditionalDits:Events}. We deduce that
	\begin{align}\label{eq:running-time-cor:Eij-step}
	\ppr{\substack{x \sim \Idl_X \\ j \sim \Rll_{J|X=x}}}{\exists i \in [m]: \frac{\Unf[W \mid (X_{<i},X_{i,j})=(x_{<i}, x_{i,j}),E_{i,j}]}{\Unf[W \mid X_{<i}=x_{<i}]} < \frac{1}{2}}
	&\leq \ppr{\substack{x \sim \Idl_X \\ j \sim \Rll_{J|X=x}}}{j \in \cG_{m}(x)}\nonumber\\
	&\leq \Idl[\tB_1 \land \ldots \land \tB_m]\nonumber\\
	&\leq \frac{c(d+1)}{\delta n},
	\end{align}
	where the first inequality holds by \cref{eq:running-time-cor:Unf[W]} and the last one holds by \cref{lemma:bad-events}.
	In addition, by \cref{fact:Prelim:smooth-sampling} along with Markov's inequality and a union bound it holds that
	\begin{align}\label{eq:running-time-cor:easy-step}
	\ppr{x \sim \Idl_X}{\exists i \in [m]: \Unf[W \mid X_{<i} = x_{<i}] < \frac{2\Unf[W]}{t}} < \frac{2m}{t}.
	\end{align}
	The proof now follows by \cref{eq:running-time-cor:Eij-step,eq:running-time-cor:easy-step}
\end{proof}


\subsection{Bounding KL-Divergence of the Conditional Distributions}\label{sec:bounding_div_lemma}
In this section we prove \cref{lemma:bounding-div-lemma}.
\begin{proof}[Proof of \cref{lemma:bounding-div-lemma}.]
	We start by noting that for any $x_{<i} \in \Supp(\Idl_{X_{<i} \mid C_{\leq i}}) \subseteq \Supp(\Idl_{X_{<i} \mid S_{\leq i}, \pT_{\leq i}, T_{\leq i}})$, the following assertions hold.
	
	\begin{assertion}\label{properties-bounding-div-lemma}~
	\begin{enumerate}
		\item $\paren{\Idl[T_{\leq i} \mid \bX_{<i} = \bx_{<i}] > 1-\frac{1}{n}}$ (holds by the event  $\pT_i$ and $T_{\leq i-1}$).\label{prop:tT_i}
		
		\item $\paren{\size{\cs_i} \geq 0.9 n}$ (holds by the event $G_i$).\label{prop:G_i}
		
		\item $\Idl_{X_i Y_i \mid \bX_{<i} = \bx_{<i}, A_{\leq i}} \equiv \Idl_{X_i Y_i \mid \bX_{<i} = \bx_{<i}, T_{\leq i}}$ (holds since $G_i$, $\pT_i$ and $\tB_i$ are just random functions of $\bX_{<i}$).\label{prop:A_i-to-T_i}
		
		\item $\Rll_{X_i Y_i \mid \bX_{<i} = \bx_{<i}, B_{\leq i}} \equiv \Rll_{X_i Y_i \mid \bX_{<i} = \bx_{<i}, J \in \cG_i(X_i)}$.\label{prop:B_i}
		
		\item For all $x_i y_i \in \Supp(\Idl_{X_i Y_i \mid \bX_{<i} = \bx_{<i}, T_{\leq i}})$ it holds that $\size{\gamma_i(x_i y_i)} \leq 1/2$.\label{prop:gamma}
		
		\item For all $x_i y_i \in \Supp(\Idl_{X_i Y_i \mid \bX_{<i} = \bx_{<i}})$ it holds that $\gamma_i(x_i y_i) \leq 2/\delta$.\label{prop:delta_bound_on_gamma}
	\end{enumerate}
	\end{assertion}
	\noindent Note that Assertion~\ref{prop:delta_bound_on_gamma} holds since for any $j \in \cs_i$ and any $x_i$ it holds that: $\tdelta_{i,j} \geq 0.9 \delta$, $\omega_{i,j} \in 1 \pm 0.1$, $\beta'_{i,j}(x_i) \leq 1.1$ and $P_{X_{i,j} \mid X_{<i}=x_{<i}}(\cX_{i,j}) \geq 0.9$. Therefore,
	\begin{align*}
		\gamma_{i}
		&\leq  \paren{\sum_{j \in \cs_{i}}{\frac{\omega_{i,j}\cdot \beta'_{i,j}(x_{i,j})}{\tdelta_{i,j}}}}/\paren{\sum_{j \in \cs_{i}} \omega_{i,j} \cdot \Unf_{X_{i,j} \mid \bX_{<i}= \bx_{<i}}\paren{\cX_{i,j}}}\\
		&\leq \zfrac{\paren{\frac{1.1 \cdot 1.1}{0.9 \delta} \cdot \size{\cs_i}}}{\paren{0.9\cdot 0.9\cdot \size{\cs_i}}} \leq 2/\delta
	\end{align*}
	
	\noindent Our goal now is to show that for any fixing of $x_{<i} \in \Supp(\Idl_{X_{<i} \mid C_{\leq i}})$ is holds that
	
	\begin{align}\label{eq:bounding-div-lemma-main-goal}
		D(\Idl_{X_i \mid X_{<i}=x_{<i}, T_{\leq i}} || \Rll_{X_i \mid X_{<i}=x_{<i}, J \in \cG_i(X_i)}) \leq \frac{c}{\delta n}\cdot \paren{D(\Idl_{X_i \mid X_{<i}=x_{<i}} || \Unf_{X_i \mid X_{<i}=x_{<i}})+1},
	\end{align}
	for some constant $c>0$. The proof then follow by \cref{eq:bounding-div-lemma-main-goal} since
	\begin{align*}
		D(\Idl_{X_i | A_{\leq i}} || \Rll_{X_i | B_{\leq i}} \mid \Idl_{X_{<i} | C_{\leq i}})
		&= \eex{x_{<i} \sim \Idl_{X_{<i} | C_{\leq i}}}{D(\Idl_{X_i | X_{<i}=x_{<i}, A_{\leq i}} || \Rll_{X_i | X_{<i}=x_{<i}, B_{\leq i}})}\\
		&= \eex{x_{<i} \sim \Idl_{X_{<i} | C_{\leq i}}}{D(\Idl_{X_i | X_{<i}=x_{<i}, T_{\leq i}} || \Rll_{X_i | X_{<i}=x_{<i}, J \in \cG_i(X_i)})}\\
		&\leq \frac{c}{\delta n}\cdot \paren{\eex{x_{<i} \sim \Idl_{X_{<i} | C_{\leq i}}}{D(\Idl_{X_i \mid X_{<i}=x_{<i}} || \Unf_{X_i \mid X_{<i}=x_{<i}})}+1}\\
		&\leq \frac{c}{\delta n} (\dval_i + 1) \cdot \frac1{\Idl[C_{\leq i}]},
	\end{align*}
	
	where the second equality holds by Properties \ref{prop:A_i-to-T_i} and \ref{prop:B_i} in \ref{properties-bounding-div-lemma}, and the last inequality holds by \cref{fact:kl-divergence:condition-on-cond-diver}. We now focus on proving \cref{eq:bounding-div-lemma-main-goal} in a similar spirit to the proof given in \cref{sec:BoundingSmoothKL:Warmup}.
	
	In the following, fix $i \in [m]$ and $x_{<i} \in \Supp(\Idl_{X_{<i} \mid C_{\leq i}})$. By data-processing of KL-divergence (\cref{fact:prelim:diver-properties}(\ref{fact:diver-properties:item:data-processing})), it holds that
	\begin{align}\label{eq:bounding-div-lemma-first}
		D(\Idl_{X_i \mid X_{<i}=x_{<i}, T_{\leq i}} || \Rll_{X_i \mid X_{<i}=x_{<i}, J \in \cG_i(X_i)})
		&\leq D(\Idl_{X_i Y_i \mid X_{<i}=x_{<i}, T_{\leq i}} || \Rll'_{X_i Y_i \mid X_{<i}=x_{<i}}),
	\end{align}
	where 
	\begin{align*}
		\Rll'_{X_i Y_i \mid X_{<i} = x_{<i}} 
		&= \Idl_{X_i Y_i \mid X_{<i} = x_{<i}, X_i, Y_{i,J}=1} \circ Q_{J,X_i \mid X_{<i} = x_{<i}, J \in \cG_i(X_i)}\\
		&\equiv \Idl_{X_i Y_i \mid X_{<i} = x_{<i}, X_{i,J}, Y_{i,J}=1} \circ Q_{J,X_{i,J} \mid X_{<i} = x_{<i}, J \in \cs_i, X_{i,J} \in \cX_{i,J}}
	\end{align*}
	Similar calculation to the one in \cref{eq:thm_proof:real-prob} yields that for any fixing of $x_i y_i \in \Supp(\Idl_{X_i Y_i \mid X_{<i}=x_{<i}})$ it holds that
	\begin{align}\label{eq:thm_proof:real-prob-actual}
		\lefteqn{\Rll'_{\bX_i \bY_i \mid \bX_{<i} = \bx_{<i}}(\bx_i \by_i)}\\
		&= \sum_{j \in \cG_i(\bx_{i}) \bigcap 1_{y_i}}{\Rll_{J \mid \bX_{<i} = \bx_{<i}, J \in \cG_i(\bx_{i})}(j) \cdot \frac{\beta_{i,j}(x_{i,j})\cdot \Idl_{\bX_i \bY_i \mid \bX_{<i} = \bx_{<i}}(\bx_i \by_i)}{\tdelta_{i,j}}}\nonumber
	\end{align}
	
	In addition, for any $j \in \cG_i(\bx_{i})$ it holds that
	\begin{align}\label{eq:thm_proof:real-I-prob-actual}
		\Rll_{J \mid \bX_{<i} = \bx_{<i},J \in \cG_i(\bx_{i})}(j) \nonumber
		&= \frac{\Rll_{J \mid \bX_{<i} = \bx_{<i},J \in \cs_{i}}(j)\cdot \indic{x_{i,j} \in \cX_{i,j}}}{\Rll[X_{i,j} \in \cX_{i,j} \mid X_{<i} = x_{<i}, J=j]}\nonumber\\
		&= \frac{\Rll_{J \mid \bX_{<i} = \bx_{<i},J \in \cs_{i}}(j)\cdot \indic{x_{i,j} \in \cX_{i,j}}}{\Unf[X_{i,j} \in \cX_{i,j} \mid X_{<i} = x_{<i}]}\nonumber\\
		&= \frac{\omega_{i,j} \cdot  \indic{x_{i,j} \in \cX_i}}{\sum_{t \in \cs_{i}} \omega_{i,t} \cdot \Unf_{X_{i,j} \mid X_{<i} = x_{<i}}(\cX_{i,j})},
	\end{align}
	where the first equality holds since $\cG_i(x_i) = \set{j \in [n] \colon j \in \cs_i \land x_{i,j} \in \cX_{i,j}}$ and the last equality holds by \cref{eq:thm_proof:real-I-prob}.
	Therefore, by combining \cref{eq:thm_proof:real-prob-actual,eq:thm_proof:real-I-prob-actual} we now can write
	\begin{align}\label{eq:thm_proof:real-prob-full-actual}
	\Rll'_{\bX_i \bY_i \mid \bX_{<i} = \bx_{<i}}(\bx_i \by_i) = 
	\frac{\sum_{j \in \cs_{i} \bigcap 1_{y_i}} \frac{\omega_{i,j} \cdot \beta'_{i,j}(x_{i,j})}{\tdelta_{i,j}}}
	{\sum_{j \in \cs_{i}} \omega_{i,j}\cdot \Unf_{X_{i,j} \mid X_{<i} = x_{<i}}(\cX_{i,j})} \cdot \Idl_{\bX_i \bY_i \mid \bX_{<i} = \bx_{<i}}(\bx_i \by_i)
	\end{align}
	
	Using \cref{eq:bounding-div-lemma-first,eq:thm_proof:real-prob-full-actual}, we deduce that
	
	\begin{align}\label{eq:bounding-div-lemma-long-calc}
	\lefteqn{D(\Idl_{X_i \mid X_{<i}=x_{<i}, T_{\leq i}} || \Rll_{X_i \mid X_{<i}=x_{<i}, J \in \cG_i(X_i)})}\\
	&\leq D(\Idl_{\bX_i \bY_i \mid \bX_{<i} = \bx_{<i}, T_{\leq i}} || \Rll'_{\bX_i \bY_i \mid \bX_{<i} = \bx_{<i}})\nonumber\\
	&= \eex{\bx_i \by_i \sim \Idl_{\bX_i \bY_i \mid \bX_{<i} = \bx_{<i}, T_{\leq i}}}{\log \frac{\Idl_{\bX_i \bY_i \mid \bX_{<i} = \bx_{<i}, T_{\leq i}}(\bx_i \by_i)}{\Rll'_{\bX_i \bY_i \mid \bX_{<i} = \bx_{<i}}(\bx_i \by_i)}}\nonumber\\
	&\leq \eex{\bx_i \by_i \sim \Idl_{\bX_i \bY_i \mid \bX_{<i} = \bx_{<i}, T_{\leq i}}}{\log \frac{\Idl_{\bX_i \bY_i \mid \bX_{<i} = \bx_{<i}}(\bx_i \by_i) / \Idl[T_{\leq i} \mid \bX_{<i} = \bx_{<i}]}{\Rll'_{\bX_i \bY_i \mid \bX_{<i} = \bx_{<i}}(x_i y_i)}}\nonumber\\
	&\leq \eex{\bx_i \by_i \sim \Idl_{\bX_i \bY_i \mid \bX_{<i} = \bx_{<i}, T_{\leq i}}}{\log \frac{\sum_{j \in \cs_i} \omega_{i,j}\cdot \Unf_{X_{i,j} \mid X_{<i} = x_{<i}}(\cX_{i,j})}{\sum_{j \in \cs_i \bigcap 1_{y_i}} \frac{\omega_{i,j} \cdot \beta'_{i,j}(x_{i,j})}{\tdelta_{i,j}}}} +2/n,\nonumber\\
	&= \eex{\bx_i \by_i \sim \Idl_{\bX_i \bY_i \mid \bX_{<i} = \bx_{<i}, T_{\leq i}}}{-\log (1+\gamma_i(x_i y_i))} +2/n\nonumber\\
	&\leq \eex{\bx_i \by_i \sim \Idl_{\bX_i \bY_i \mid \bX_{<i} = \bx_{<i}, T_{\leq i}}}{-\gamma_i(x_i y_i) + \gamma_i(x_i y_i)^2} +2/n\nonumber\\
	&\leq -\eex{\bx_i \by_i \sim \Idl_{\bX_i \bY_i \mid \bX_{<i} = \bx_{<i}}}{\gamma_i(x_i y_i)}/\Idl[T_{\leq i} \mid \bX_{<i} = \bx_{<i}] + \frac2{\delta n} + \eex{\bx_i \by_i \sim \Idl_{\bX_i \bY_i \mid \bX_{<i} = \bx_{<i}, T_{\leq i}}}{\gamma_i(x_i y_i)^2} +2/n\nonumber
	\end{align}
	The third inequality holds by \cref{eq:thm_proof:real-prob-full-actual} and by Assertion \ref{properties-bounding-div-lemma}(\ref{prop:tT_i}) which yields that $\log \frac1{\Idl[T_{\leq i} \mid X_{<i} = x_{<i}]} \leq 2/n$. The one before last inequality holds by the inequality $-\log(1+x) \leq -x + x^2$ for $\size{x} \leq 1/2$ (recall Assertion \ref{properties-bounding-div-lemma}(\ref{prop:gamma})). The last inequality holds since for any random variable $X \leq M$ and any event $T$ it holds that $\ex{-X \mid T} = \frac{-\ex{X} + \ex{X \mid \overline{T}}\cdot \pr{\overline{T}}}{\pr{T}} \leq -\ex{X} / \pr{T} + M\cdot \pr{\overline{T}}$ (recall Assertions \ref{properties-bounding-div-lemma}(\ref{prop:tT_i},\ref{prop:delta_bound_on_gamma})). Note that
	
	\begin{align}\label{eq:thm_proof:exp_gamma_internal-actual}
		\lefteqn{\eex{\bx_i \by_i \sim \Idl_{\bX_i \bY_i \mid \bX_{<i} = \bx_{<i}}}{\sum_{j \in \cs_i \bigcap 1_{y_i}} \frac{\omega_{i,j}\cdot \beta'_{i,j}(x_{i,j})}{\tdelta_{i,j}}}}\\
		&= \sum_{j \in \cs_i} \eex{x_{i,j} y_{i,j} \sim \Idl_{X_{i,j} Y_{i,j} \mid \bX_{<i} = \bx_{<i}}}{y_{i,j}\cdot \frac{\omega_{i,j}\cdot \beta'_{i,j}(x_{i,j})}{\tdelta_{i,j}}}\nonumber\\
		&= \sum_{j \in \cs_i} \omega_{i,j}\cdot \eex{x_{i,j} \sim \Idl_{X_{i,j} \mid \bX_{<i} = \bx_{<i}, Y_{i,j} = 1}}{\beta'_{i,j}(x_{i,j})}\nonumber\\
		&= \sum_{j \in \cs_i} \omega_{i,j}\cdot \eex{x_{i,j} \sim \Idl_{X_{i,j} \mid \bX_{<i} = \bx_{<i}, Y_{i,j} = 1}}{\frac{\Unf_{X_{i,j} \mid \bX_{<i} = \bx_{<i}}(x_{i,j})}{\Idl_{X_{i,j} \mid \bX_{<i} = \bx_{<i}, Y_{i,j} = 1}(x_{i,j})} \cdot \indic{x_{i,j} \in \cX_{i,j}}}\nonumber\\
		&= \sum_{j=1}^{n} \omega_{i,j}\cdot \Unf_{X_{i,j} \mid \bX_{<i}=\bx_{<i}}\paren{\Supp(\Idl_{X_{i,j} \mid \bX_{<i} = \bx_{<i}, Y_{i,j}=1}) \bigcap \cX_{i,j}}
		= \sum_{j=1}^{n} \omega_{i,j} \cdot \Unf_{X_{i,j} \mid \bX_{<i}=\bx_{<i}}\paren{\cX_{i,j}},\nonumber
	\end{align}
	where the second equality holds since $y_{i,j} \in \zo$ and since for all $j \in \cs_i$ it holds that $\Idl_{Y_{i,j} \mid \bX_{<i} = \bx_{<i}}(1) = \tdelta_{i,j}  > 0$, and the last equality holds since $\cX_{i,j} \subseteq \Supp(\Idl_{X_{i,j} \mid \bX_{<i} = \bx_{<i}, Y_{i,j}=1})$. Therefore, we deduce from \cref{eq:thm_proof:exp_gamma_internal-actual} that
	
	\begin{align}\label{eq:thm_proof:exp_gamma-actual}
		\eex{\bx_i \by_i \sim \Idl_{\bX_i \bY_i \mid \bX_{<i} = \bx_{<i}}}{\gamma_i(\bx_i \by_i)} = 0
	\end{align}

	Therefore, by \cref{eq:bounding-div-lemma-long-calc,eq:thm_proof:exp_gamma-actual}, in order to prove \cref{eq:bounding-div-lemma-main-goal} it is left to show that
	\begin{align}\label{eq:thm_proof:exp_gamma_square-actual}
		\eex{\bx_i \by_i \sim \Idl_{\bX_i \bY_i \mid \bX_{<i} = \bx_{<i}, T_{\leq i}}}{\gamma_i(\bx_i\by_i)^2}
		\leq O\paren{\frac1{\delta n}}\cdot \paren{D(\Idl_{\bX_i \bY_i \mid \bX_{<i} = \bx_{<i}} || \Unf_{\bX_i \bY_i \mid \bX_{<i} = \bx_{<i}}) + 1}.
	\end{align}
	
	Let $\Gamma$ be the value of $\gamma_i(\bx_i\by_i)$ when $\bx_i\by_i$ is drawn from either $\Idl' = \Idl_{\bX_i \bY_i \mid \bX_{<i} = \bx_{<i}, T_{\leq i}}$ or $\Idl^{\Pi} = \prod_{j=1}^n \Idl_{X_{i,j} Y_{i,j} \mid \bX_{<i} = \bx_{<i}}$. 
	We now prove that there exists constants $K_1, K_2 > 0$ such that for every $\gamma \in [0,1]$ it holds that
	\begin{align}\label{eq:thm_proof:goal_under_product-actual}
		\Idl^{\Pi}[\size{\Gamma} \geq \gamma] \leq K_2\cdot \exp\paren{-\frac{\gamma^2}{K_1\cdot \sigma^2}},
	\end{align}
	The proof of \cref{eq:thm_proof:exp_gamma_square-actual} then follows since 
	\begin{align}
		\eex{\bx_i \by_i \sim \Idl_{\bX_i \bY_i \mid \bX_{<i} = \bx_{<i}, T_{\leq i}}}{\gamma_i(\bx_i\by_i)^2}
		&= \eex{\Idl'}{\Gamma^2}
		\leq \frac{K_3}{\delta n}\cdot \paren{D(\Idl' || \Idl^{\Pi})}\\
		&= \frac{K_3}{\delta n}\cdot \paren{D(\Idl_{\bX_i \bY_i \mid \bX_{<i} = \bx_{<i}, T_{\leq i}} || \prod_{j=1}^n \Idl_{X_{i,j} Y_{i,j} \mid \bX_{<i} = \bx_{<i}}) + 1}\nonumber\\
		&\leq \frac{K_3}{\delta n}\cdot \paren{D(\Idl_{\bX_i \bY_i \mid \bX_{<i} = \bx_{<i}, T_{\leq i}} || \Unf_{\bX_i \bY_i \mid \bX_{<i} = \bx_{<i}}) + 1}\nonumber\\
		&\leq \frac{K_3}{\delta n}\cdot \paren{\frac{1}{\Idl[T_{\leq j}]}\paren{D(\Idl_{\bX_i \bY_i \mid \bX_{<i} = \bx_{<i}} || \Unf_{\bX_i \bY_i \mid \bX_{<i} = \bx_{<i}}) + 1/e + 1} + 1}\nonumber\\
		&\leq \frac{5 K_3}{\delta n}\cdot \paren{D(\Idl_{\bX_i \bY_i \mid \bX_{<i} = \bx_{<i}} || \Unf_{\bX_i \bY_i \mid \bX_{<i} = \bx_{<i}}) + 1}\nonumber
	\end{align}
	
	where the first inequality holds by \cref{prop:prelim:sub-exp-to-divergence} and the fact that $\size{\Gamma} \leq 1$ under $\Idl'$, the second inequality holds by chain rule of KL-divergence when the right-hand side distribution is product (\cref{fact:prelim:diver-properties}(\ref{fact:diver-properties:item:chain-rule})), the one before last inequality holds by \cref{fact:kl-divergence:conditionP}, and the last one holds since $\Idl[T_{\leq j}] \geq 1/2$. 

	We now prove \cref{eq:thm_proof:goal_under_product-actual}. In the following, given $\bx_i \by_i$ which are drawn from either $\Idl^{\Pi}$ or $\Idl^{\Pi'} = \prod_{j=1}^n \Idl_{Y_{i,j} \mid \bX_{<i} = \bx_{<i}} \cdot \Idl_{X_{i,j} \mid \bX_{<i} = \bx_{<i}, Y_{i,j}=1}$, we define the random variables  $L_j$,$Z_j$,$L$ and $Z$ (in addition to $\Gamma$), where $L_j$ is the value of $\omega_j\cdot \beta_j'(x_{i,j})$, $L = \sum_{j=1}^{n} L_j$, $Z_j = \begin{cases} L_j / \tdelta_j & y_{i,j}=1 \\ 0 & y_{i,j}=0 \end{cases}$ and $Z = \sum_{j=1}^{n} Z_j$, letting $\omega_j = \omega_{i,j}$, $\beta_j(\cdot) = \beta_{i,j}(\cdot)$ and $\tdelta_j = \tdelta_{i,j}$. Note that by definition, $Z = (1 + \Gamma) \mu$ for $\mu = \sum_{j \in \cs_i} \omega_j \cdot \Unf_{X_{i,j} \mid X_{<i}=x_{<i}}(\cX_{i,j})$ (follows from \cref{eq:thm_proof:exp_gamma_internal-actual}). Moreover, by the definition of $\cs_i$ and the fact that $\size{\cs_i} \geq 0.9 n$ (Assertion \ref{properties-bounding-div-lemma}(\ref{prop:G_i})), it holds that $\size{L_{j}} \leq 2$, $\tdelta_j \geq 0.9 \delta$ and $\mu \geq n/2$. Hence,
	
	\begin{align}\label{eq:bounding-div:concet-under-prod}
		\Idl^{\Pi}[\size{\Gamma} \geq \gamma] = \Idl^{\Pi'}[\size{\Gamma} \geq \gamma] \leq 4 \exp\paren{-\frac{\delta n \gamma^2}{100}},
	\end{align}
	where the equality holds since $\Gamma$ has the same distribution under $\Idl^{\Pi}$ and under $\Idl^{\Pi'}$, and the inequality holds
	by  \cref{fact:prelim:L_i_and_Y_i} since under $\Idl^{\Pi'}$ the $L_j$'s are independent random variables with $\eex{\Idl^{\Pi'}}{L} = \mu$ and for all $j \in \cs_i$ we have $Z_j = (L_j / \tdelta_j) \cdot \Bern(\tdelta_j)$. This proves \cref{eq:thm_proof:goal_under_product-actual} and concludes the proof.
		
\end{proof}

	\subsection{Bounding the Probability of Bad Events Under $\Idl$}\label{sec:bad_events}

In this section we prove \cref{lemma:bad-events}. We start with few more notations (in addition to the ones given in \cref{def:ConditionalDits:definitions}), then we prove facts about the distribution $\Idl$ (\cref{sec:bad-events:facts}), and in \cref{sec:bad-events:proof} we present the proof of \cref{lemma:bad-events}.

\begin{notation}[Additional notations]\label{title:additional-definitions}~
\begin{itemize}
	\item $\alpha_{i,j} = \alpha_{i,j}(x_{i,j};x_{<i}) = \frac{\Unf_{X_{i,j} \mid X_{<i}=x_{<i}}(x_{i,j})}{\Idl_{X_{i,j} \mid X_{<i}=x_{<i}}(x_{i,j})} - 1$.
	
	\item $\rho_{i,j} = \rho_{i,j}(x_{<i}) = \frac{\Idl[E_{i,j} \mid X_{<i}=x_{<i}]}{\Unf[E_{i,j} \mid X_{<i}=x_{<i}]} - 1 = \frac{\tdelta_{i,j}}{\delta_{i,j}} - 1$.
	
	\item $\tau_{i,j} = \tau_{i,j}(x_i; x_{<i}) = \frac{\Idl[E_{i,j} \mid X_{\leq i}=x_{ \leq i}]}{\Unf[E_{i,j} \mid X_{ \leq i}=x_{ \leq i}]} - 1 = \frac{\Idl[E_{i,j} \mid X_{\leq i}=x_{ \leq i}]}{\delta_{i,j}} - 1$.
	
	\item $\xi_{i,j} = \xi_{i,j}(x_{i,j}; x_{<i}) = \frac{\Idl[E_{i,j} \mid X_{<i}=x_{<i}, X_{i,j}=x_{i,j}]}{\Unf[E_{i,j} \mid X_{<i}=x_{<i}, X_{i,j}=x_{i,j}]} - 1 = \frac{\Idl[E_{i,j} \mid X_{<i}=x_{<i}, X_{i,j}=x_{i,j}]}{\delta_{i,j}} - 1$.
	
	\item $U_{i,j} = U_{i,j}(x_{i,j}; x_{<i}) = \prod_{s=1}^i \frac{\Idl[E_{s,j} \mid X_{<s}=x_{<s}, X_{s,j}=x_{s,j}]}{\Idl[E_{s,j} \mid X_{<s}=x_{<s}]} = U_{i-1,j} \cdot  \frac{1 + \xi_{i,j}}{1 + \rho_{i,j}}$.
	
	\item $V_{i,j} = V_{i,j}(x_i; x_{<i}) = \prod_{s=1}^i \frac{\Idl[E_{s,j} \mid X_{\leq s}=x_{\leq s}]}{\Idl[E_{s,j} \mid X_{<s}=x_{<s}]} = V_{i-1,j} \cdot \frac{1 + \tau_{i,j}}{1 + \rho_{i,j}}$.
	
	\item $R_{i,j} = R_{i,j}(x_{<i}) = \frac{n \cdot \omega'_{i,j}}{\sum_{t=1}^n \omega'_{i,t}} = n \cdot \zfrac{\paren{\prod_{s=1}^{i-1} \frac{\Unf_{X_{s,j} \mid X_{<s}=x_{<s}}(x_{s,j})}{\Idl_{X_{s,j} \mid X_{<s}=x_{<s}}(x_{s,j})}}}{\paren{\sum_{t=1}^n \prod_{s=1}^{i-1} \frac{\Unf_{X_{s,t} \mid X_{<s}=x_{<s}}(x_{s,t})}{\Idl_{X_{s,t} \mid X_{<s}=x_{<s}}(x_{s,t})}}}$.
	
\end{itemize}

where in all definitions, recall that $\delta_{i,j} = \Unf[E_{i,j} \mid X_{\leq i}]$ for any fixing of $X_{\leq i}$.

\end{notation}

\subsubsection{Facts about $\Idl$}\label{sec:bad-events:facts}

\begin{fact}\label{fact:bad-events:r}
	For all $r \in \set{\rho, \tau, \xi}$ it holds that
	\begin{enumerate}
		
		\item $\eex{\Idl}{\sum_{i=1}^m \sum_{j=1}^n \min\set{\size{r_{i,j}}, r_{i,j}^2}} \leq \frac{4d}{\delta}$.\label{fact:bad-events:r:1}
		
		\item For all $\lambda > 0:$ $\eex{\Idl}{\size{\set{j \in [n] \colon \exists i \in [m]\text{ s.t. } \size{r_{i,j}} \geq \lambda}}} \leq \frac{4d}{\delta \cdot \min\set{\lambda,\lambda^2}}$.\label{fact:bad-events:r:2}
		
	\end{enumerate}
\end{fact}
\begin{proof}
	Assuming \cref{fact:bad-events:r:1} holds, then \cref{fact:bad-events:r:2} holds since
	\begin{align*}
		\eex{\Idl}{\size{\set{j \in [n] \colon \exists i \in [m]\text{ s.t. } \size{r_{i,j}} \geq \lambda}}}
		&\leq \frac1{\min\set{\lambda,\lambda^2}}\cdot \eex{\Idl}{\sum_{i=1}^m \sum_{j=1}^n \min\set{\size{r_{i,j}}, r_{i,j}^2} \cdot \indic{\min\set{\size{r_{i,j}}, r_{i,j}^2} \geq \min\set{\lambda,\lambda^2}}}\\
		&\leq \frac1{\min\set{\lambda,\lambda^2}}\cdot \eex{\Idl}{\sum_{i=1}^m \sum_{j=1}^n \min\set{\size{r_{i,j}}, r_{i,j}^2}}
	\end{align*}
	
	\cref{fact:bad-events:r:1} for $r=\rho$ holds since
	\begin{align*}
		d 
		&= \sum_{i=1}^m \eex{\Idl_{X_{<i}}}{D(\Idl_{X_i Y_i \mid X_{<i}} || \Unf_{X_i Y_i \mid X_{<i}})}
		\geq \sum_{i=1}^m \eex{\Idl_{X_{<i}}}{D(\Idl_{Y_i \mid X_{<i}} || \Unf_{Y_i \mid X_{<i}})}\\
		&\geq \sum_{i=1}^m \sum_{j=1}^n \eex{\Idl_{X_{<i}}}{D(\Idl_{Y_{i,j} \mid X_{<i}} || \Unf_{Y_{i,j} \mid X_{<i}})}
		= \sum_{i=1}^m \sum_{j=1}^n \eex{\Idl_{X_{<i}}}{D((1+\rho_{i,j})\delta_{i,j} || \delta_{i,j})}\\
		&\geq  \sum_{i=1}^m \sum_{j=1}^n \delta_{i,j} \cdot \eex{\Idl_{X_{<i}}}{\min\set{\size{\rho_{i,j}}, \rho_{i,j}^2}} / 4
		\geq \delta \cdot \eex{\Idl}{\sum_{i=1}^m \sum_{j=1}^n \min\set{\size{\rho_{i,j}}, \rho_{i,j}^2}}/4,
	\end{align*}
	where the first inequality holds by data processing of KL divergence (\cref{fact:prelim:diver-properties}(\ref{fact:diver-properties:item:data-processing})), the second one holds by chain rule of KL-divergence when the right-hand side distribution is product (\cref{fact:prelim:diver-properties}(\ref{fact:diver-properties:item:chain-rule})), and the one before last inequality holds by \cref{fact:prelim:bernoulli-div-est}.
	
	For $r=\tau$, \cref{fact:bad-events:r:1} holds since
	\begin{align*}
		d 
		&= \sum_{i=1}^m \eex{\Idl_{X_{<i}}}{D(\Idl_{X_i Y_i \mid X_{<i}} || \Unf_{X_i Y_i \mid X_{<i}})}
		\geq \sum_{i=1}^m \eex{\Idl_{X_{\leq i}}}{D(\Idl_{Y_i \mid X_{\leq i}} || \Unf_{Y_i \mid X_{\leq i}})}\\
		&\geq \sum_{i=1}^m \sum_{j=1}^n \eex{\Idl_{X_{\leq i}}}{D(\Idl_{Y_{i,j} \mid X_{\leq i}} || \Unf_{Y_{i,j} \mid X_{\leq i}})}
		= \sum_{i=1}^m \sum_{j=1}^n \eex{\Idl_{X_{\leq i}}}{D((1+\tau_{i,j})\delta_{i,j} || \delta_{i,j})}\\
		&\geq \sum_{i=1}^m \sum_{j=1}^n \delta_{i,j} \cdot \eex{\Idl_{X_{\leq i}}}{\min\set{\size{\tau_{i,j}}, \tau_{i,j}^2}}/4
		\geq \delta \cdot \eex{\Idl}{\sum_{i=1}^m \sum_{j=1}^n \min\set{\size{\tau_{i,j}}, \tau_{i,j}^2}}/4,
	\end{align*}
	where the first inequality holds by chain rule (\cref{fact:prelim:diver-properties}(\ref{fact:diver-properties:item:chain-rule})) and the second one holds by chain rule when the right-hand side distribution is product (\cref{fact:prelim:diver-properties}(\ref{fact:diver-properties:item:chain-rule})).
	For $r=\xi$, \cref{fact:bad-events:r:1} holds since
	\begin{align*}
		d 
		&= \sum_{i=1}^m \eex{\Idl_{X_{< i}}}{D(\Idl_{X_i Y_i \mid X_{<i}} || \Unf_{X_i Y_i \mid X_{<i}})}
		\geq \sum_{i=1}^m \sum_{j=1}^n \eex{\Idl_{X_{<i}}}{D(\Idl_{X_{i,j} Y_{i,j} \mid X_{<i}} || \Unf_{X_{i,j}  Y_{i,j} \mid X_{<i}})}\\
		&\geq  \sum_{i=1}^m \sum_{j=1}^n \eex{\Idl_{X_{< i}, X_{i,j}}}{(\Idl_{Y_{i,j} \mid X_{<i}, X_{i,j}} || \Unf_{Y_{i,j} \mid X_{<i}, X_{i,j}})}
		= \sum_{i=1}^m \sum_{j=1}^n \eex{\Idl_{X_{< i}, X_{i,j}}}{D((1+\xi_{i,j})\delta_{i,j} || \delta_{i,j})}\\
		&\geq \sum_{i=1}^m \sum_{j=1}^n \delta_{i,j} \cdot \eex{\Idl_{X_{< i}, X_{i,j}}}{\min\set{\size{\xi_{i,j}}, \xi_{i,j}^2}}/4
		\geq \delta \cdot \eex{\Idl}{\sum_{i=1}^m \sum_{j=1}^n \min\set{\size{\xi_{i,j}}, \xi_{i,j}^2}}/4,
	\end{align*}
	where the first inequality holds by chain rule when the right-hand side distribution is product (\cref{fact:prelim:diver-properties}(\ref{fact:diver-properties:item:chain-rule})) and the second one holds by standard chain-rule ok KL-divergence (\cref{fact:prelim:diver-properties}(\ref{fact:diver-properties:item:chain-rule})).

\end{proof}

\begin{fact}\label{fact:bad-events:martingale}
	For all $L \in \set{U,V}$ it holds that 
	\begin{enumerate}
		\item For any $j \in [n]:$ the sequence $\set{L_{i,j}}_{i=1}^m$ is a martingale with respect to $\set{X_i}_{i=1}^m$ which are drawn from $\Idl$ (recall \cref{def:prelim:martingales}).\label{fact:bad-events:martingale:1}
		
		\item For any $\lambda \in (0,\frac14):$ $\eex{\Idl}{\size{\set{j \in [n] \colon \exists i \in [m]\text{ s.t. } \size{L_{i,j}-1} \geq \lambda}}} \leq \frac{c \cdot d}{\delta \lambda^2}$, for some universal constant $c > 0$.\label{fact:bad-events:martingale:2}
	\end{enumerate}
\end{fact}
\begin{proof}
	Note that for any fixing of $x_{<i}$ it holds that
	\begin{align*}
	\eex{\Idl_{X_i \mid X_{<i}=x_{<i}}}{U_{i,j}} = U_{i-1,j} \cdot \frac{\eex{\Idl_{X_i \mid X_{<i}=x_{<i}}}{\Idl[E_{i,j} \mid X_{<i}=x_{<i}, X_{i,j}=x_{i,j}]}}{\Idl[E_{i,j} \mid X_{<i}=x_{<i}]} = U_{i-1,j}
	\end{align*}
	and 
	\begin{align*}
	\eex{\Idl_{X_i \mid X_{<i}=x_{<i}}}{V_{i,j}} = V_{i-1,j} \cdot \frac{\eex{\Idl_{X_i \mid X_{<i}=x_{<i}}}{\Idl[E_{i,j} \mid X_{<i}=x_{<i}, X_{i}=x_{i}]}}{\Idl[E_{i,j} \mid X_{<i}=x_{<i}]} = V_{i-1,j}
	\end{align*}
	This proves \cref{fact:bad-events:martingale:1}. 
	By \cref{prop:prelim:martingale-specific-bound}, there exists a constant $c' > 0$ such that for any $j \in [n]$ and $\lambda \in (0,\frac14)$ it holds that
	\begin{align*}
	\Idl[\exists i \in [m]\text{ s.t. } \size{U_{i,j}-1} \geq \lambda] \leq \frac{c' \cdot \eex{\Idl}{\sum_{i=1}^n \paren{\min\set{\size{\rho_{i,j}},\rho_{i,j}^2} + \min\set{\size{\xi_{i,j}},\xi_{i,j}^2}}}}{\lambda^2}
	\end{align*}
	and that
	\begin{align*}
	\Idl[\exists i \in [m]\text{ s.t. } \size{V_{i,j}-1} \geq \lambda] \leq \frac{c' \cdot \eex{\Idl}{\sum_{i=1}^n \paren{\min\set{\size{\rho_{i,j}},\rho_{i,j}^2} + \min\set{\size{\tau_{i,j}},\tau_{i,j}^2}}}}{\lambda^2}
	\end{align*}
	The proof of \cref{fact:bad-events:martingale:2} now follows from the bounds in \cref{fact:bad-events:r}(\ref{fact:bad-events:r:1}).
	
\end{proof}

\begin{fact}\label{fact:bad-events:Rij}
	For every $\lambda > 0$ it holds that
	\begin{align*}
	\eex{\Idl}{\size{\set{j \in [n] \colon \exists i\in[m]\text{ s.t. } \size{R_{i,j}-1} > \lambda}}} \leq \frac{16 \cdot \dval}{\min\set{\lambda,\lambda^2}}.
	\end{align*}
\end{fact}
\begin{proof}
	We prove that for every $i \in [m]$ it holds that
	\begin{align}\label{eq:bad_events:Rij-goal}
	\eex{\Idl}{\size{\set{j \in [n] \colon \size{R_{i,j}-1} > \lambda}}} \leq \frac{16 \cdot \dval_i}{\min\set{\lambda,\lambda^2}},
	\end{align} 
	The proof of the fact then follows since 
	\begin{align*}
	\eex{\Idl}{\size{\set{j \in [n] \colon \exists i\in[m]\text{ s.t. } \size{R_{i,j}-1} > \lambda}}}
	&\leq \sum_{i=1}^m \eex{\Idl}{\size{\set{j \in [n] \colon \size{R_{i,j}-1} > \lambda}}}
	\end{align*}
	
	In the following, let
	\begin{align*}
	\Rll' = \Rll'_{\bX} =  \prod_{i=1}^{m} \Unf_{X_{i,J} | X_{<i,J}} \Idl_{X_{i,-J} | \bX_{<i}, X_{i,J}}\circ \Rll'_{J},
	\end{align*}
	where $\Rll'_{J} = U_{[n]}$. 
	Applying \cref{eq:thm_proof:real-I-prob} on $\Rll'$ (note that $\Rll'$ is a special case of a skewed distribution $\Rll$ when choosing events $\set{E_{i,j}}$ with $\Unf[E_{i,j}]=1$ for all $i,j$), we obtain for any $i \in [m]$, $x_{<i} \in \Supp(\Idl_{X_{<i}})$ and any $j \in [n]$:
	\begin{align}\label{eq:bad-events-R_i-to_J}
	\Rll'_{J | X_{<i}=x_{<i}}(j) = \paren{\prod_{s=1}^{i-1} \frac{\Unf_{X_{s,j} \mid X_{<s}=x_{<s}}(x_{s,j})}{\Idl_{X_{s,j} \mid X_{<s}=x_{<s}}(x_{s,j})}}/\paren{\sum_{t=1}^n \prod_{s=1}^{i-1} \frac{\Unf_{X_{s,t} \mid X_{<s}=x_{<s}}(x_{s,t})}{\Idl_{X_{s,t} \mid X_{<s}=x_{<s}}(x_{s,t})}} = \frac{R_{i,j}}{n}
	\end{align}
	
	\noindent Moreover, as proven in \cite{ChungP15}, by letting $\Idl_J = U_{[n]}$ (\ie the uniform distribution over $[n]$), it holds that
	\begin{align*}
	\lefteqn{\eex{\Idl_{J}\Idl_{X_{<i}}}{D(\Idl_{X_i \mid X_{<i}} || \Rll'_{X_i \mid J, X_{<i}})}
	= \frac1{n}\cdot  \eex{\Idl_{X_{<i}}}{\sum_{j=1}^n D(\Idl_{X_i \mid X_{<i}} || \Rll'_{X_i \mid J=j, X_{<i}})}}\\
	&= \frac1{n}\cdot \eex{\Idl_{X_{<i}}}{\sum_{j=1}^n D(\Idl_{X_{i,j} \mid X_{<i}} || \Unf_{X_{i,j} \mid X_{<i}})}
	\leq \frac1{n}\cdot \eex{\Idl_{X_{<i}}}{D(\Idl_{X_{i} \mid X_{<i}} || \Unf_{X_{i} \mid X_{<i}})} \leq \frac{d_i}{n},
	\end{align*}
	where inequality holds by chain-rule of KL-divergence where the right-hand side is product. The above yields that
	\begin{align}\label{eq:bad-events:J-in-round-i}
	\eex{\Idl_{X_{<i}}}{D(U_{[n]} || \Rll'_{J \mid X_{<i}})}
	&\leq \eex{\Idl_{X_{<i}}}{D(U_{[n]} \Idl_{X_i \mid X_{<i}} || \Rll'_{J} \Rll'_{X_i \mid J, X_{<i}})}\\
	&= \eex{\Idl_{J} \Idl_{X_{<i}}}{D(\Idl_{X_i \mid X_{<i}} || \Rll'_{X_i \mid J, X_{<i}})}
	\leq \frac{d_i}{n},\nonumber
	\end{align}
	where the first inequality holds by data-processing of KL-divergence, and the equality holds by chain-rule of KL-divergence along with the fact that $\Idl_{J} \equiv \Rll'_{J} \equiv U_{[n]}$. In the following, fix $i \in [m]$ and let
	$\cB_i^+ = \cB_i^+(x_{<i}) = \set{j \in [n] \colon \Rll'_{J|X_{<i}=x_{<i}}(j) > (1+\lambda)/n}$ and let $\cB_i^- = \cB_i^-(x_{<i}) = \set{j \in [n] \colon \Rll'_{J|X_{<i}=x_{<i}}(j) < (1-\lambda)/n}$. By \cref{eq:bad-events:J-in-round-i} along with data-processing of KL-divergence, it holds that
	\begin{align*}
	\eex{\Idl_{X_{<i}}}{D(\frac{\size{\cB_i^+}}{n} || (1 + \lambda ) \frac{\size{\cB_i^+}}{n})} \leq \eex{\Idl_{X_{<i}}}{D(U_{[n]}(\cB_i^+)  ||  \Rll'_{J|X_{<i}}(\cB_i^+))} \leq \dval_i/n
	\end{align*}
	and by \cref{fact:prelim:bernoulli-div-est} we deduce that $\eex{\Idl}{\size{\cB_i^+}} \leq \frac{8 \cdot \dval_i}{\min\set{\size{\lambda},\lambda^2}}$. Similarly it holds that $\eex{\Idl}{\size{\cB_i^-}} \leq \frac{8 \cdot \dval}{\min\set{\size{\lambda},\lambda^2}}$. The proof of \cref{eq:bad_events:Rij-goal} now follows since for any $x_{<i} \in \Supp(\Idl_{X_{<i}})$ and any $j \notin \cB_i^+(x_{<i}) \bigcup \cB_i^-(x_{<i})$ it holds that $\frac{R_{i,j}(x_{<i})}{n} = \Rll'_{J|X_{<i}=x_{<i}}(j) \in (1 \pm \lambda)/n$ (the equality holds by \cref{eq:bad-events-R_i-to_J}).
\end{proof}

\begin{fact}\label{fact:bad-events:omega}
	For all $\lambda \in (0,\frac14)$ it holds that $$\eex{\Idl}{\size{\set{j \in [n] \colon \exists i \in [m]\text{ s.t. } \size{\omega_{i,j} - 1} \geq \lambda}}} \leq \frac{c\cdot d}{\delta \cdot \min\set{\lambda,\lambda^2}},$$ for some universal constant $c > 0$.
\end{fact}
\begin{proof}
	Note that
	\begin{align*}
	\omega_{i,j} 
	= \frac{R_{i,j}\cdot V_{i-1,j}}{U_{i-1,j}}
	\end{align*}
	Therefore, we deduce that
	\begin{align*}
	\lefteqn{\eex{\Idl}{\size{\set{j \in [n] \colon \exists i \in [m]\text{ s.t. } \size{\omega_{i,j} - 1} \geq \lambda}}}}\\
	&\leq \eex{\Idl}{\size{\set{j \in [n] \colon \exists i \in [m]\text{ s.t. } \paren{\size{U_{i-1,j}-1} > \lambda/10} \lor \paren{\size{V_{i-1,j}-1} > \lambda/10} \lor \paren{\size{R_{i,j}-1} > \lambda/10}}}}\\
	&\leq 100(c_1 + c_2 + c_3) \cdot d/ \delta,
	\end{align*}
	where $c_1$, $c_2$ and $c_3$ are the constants from \cref{fact:bad-events:r}(\ref{fact:bad-events:r:2}), \cref{fact:bad-events:Rij} and \cref{fact:bad-events:martingale}(\ref{fact:bad-events:martingale:2}), respectively.
\end{proof}

\begin{fact}\label{fact:bad-events:alpha}
	For every $\lambda \in (0,\frac12)$ it holds that	
	\begin{align*}
		\eex{\Idl}{\size{\set{j \in [n] \colon \exists i\in[m]\text{ s.t. } \alpha_{i,j} > \lambda}}} \leq \frac{4 \cdot \dval}{\lambda^2},
	\end{align*}
	 for some constant $c > 0$.
\end{fact}
\begin{proof}
	We prove that for every $i \in [m]$ it holds that
	\begin{align}\label{eq:bad_events:alpha_ij-goal}
	\eex{\Idl}{\size{\set{j \in [n] \colon \alpha_{i,j} > \lambda}}} \leq \frac{4 \cdot \dval_i}{\min\set{\lambda,\lambda^2}},
	\end{align} 
	The proof of the fact then follows since 
	\begin{align*}
	\eex{\Idl}{\size{\set{j \in [n] \colon \exists i\in[m]\text{ s.t. } \alpha_{i,j} > \lambda}}}
	&\leq \sum_{i=1}^m \eex{\Idl}{\size{\set{j \in [n] \colon \alpha_{i,j} > \lambda}}}
	\end{align*}
	In the following, fix $i \in [m]$ and compute
	\begin{align*}
		d_i
		&\geq \eex{\Idl_{X_{<i}}}{D(\Idl_{X_i \mid X_{<i}} || \Unf_{X_i \mid X_{<i}})}
		\geq \sum_{j=1}^n \eex{\Idl_{X_{<i}}}{D(\Idl_{X_{i,j} \mid X_{<i}} || \Unf_{X_{i,j} \mid X_{<i}})}\\
		&\geq \sum_{j=1}^n \eex{\Idl_{X_{<i}}}{D(\Idl_{X_{i,j} \mid X_{<i}}[\alpha_{i,j} > \lambda] || \Unf_{X_{i,j} \mid X_{<i}}[\alpha_{i,j} > \lambda])}\\
		&\geq \sum_{j=1}^n \eex{\Idl_{X_{<i}}}{D(\Idl_{X_{i,j} \mid X_{<i}}[\alpha_{i,j} > \lambda] || (1+\lambda)\cdot \Idl_{X_{i,j} \mid X_{<i}}[\alpha_{i,j} > \lambda])}\\
		&\geq \sum_{j=1}^n \eex{\Idl_{X_{<i}}}{\frac12\cdot \paren{\frac{\lambda}{1+\lambda}}^2 \cdot (1+\lambda)\cdot \Idl_{X_{i,j} \mid X_{<i}}[\alpha_{i,j} > \lambda]}\\
		&\geq \frac{\lambda^2}{4}\cdot \sum_{j=1}^n \eex{\Idl_{X_{<i}}}{ \Idl_{X_{i,j} \mid X_{<i}}[\alpha_{i,j} > \lambda]}
		= \frac{\lambda^2}{4}\cdot \eex{\Idl}{\sum_{j=1}^n \indic{\alpha_{i,j} > \lambda}}\\
		&= \frac{\lambda^2}{4}\cdot \eex{\Idl}{\size{\set{j \in [n] \colon \alpha_{i,j} > \lambda}}}.
	\end{align*}
	Which concludes the proof of \cref{eq:bad_events:alpha_ij-goal}. The second inequality holds by data-processing of KL-divergence when the right-hand side distribution is product. The third inequality holds by data-processing of KL-divergence. The fourth inequality holds since for any $x_{i,j}$ with $\alpha_{i,j}(x_{i,j}) > \lambda$, it holds that $\Unf_{X_{i,j} \mid X_{<i}}(x_{i,j}) \geq (1+\lambda) \Idl_{X_{i,j} \mid X_{<i}}(x_{i,j})$. The fifth inequality holds by \cref{fact:prelim:bernoulli-div-est}(\ref{fact:prelim:bernoulli-div-est:minus}).
\end{proof}

\begin{fact}\label{fact:bad-events:beta}
	There exist constants $c,c' >0$ such that for all $\lambda > 0$ it holds that
	\begin{enumerate}
		\item $\eex{\Idl}{\size{\set{j \in [n] \colon \exists i \in [m]\text{ s.t. } \beta_{i,j} \geq 1 + \lambda}}} \leq \frac{c\cdot d}{\delta}$.\label{fact:bad-events:beta:1}
		
		\item $\eex{x \sim \Idl_X}{\sum_{i=1}^m \sum_{j=1}^n \Unf_{X_{i,j} \mid X_{<i} = x_{<i}}(\overline{\cX_{i,j}}) \cdot \indic{\rho_{i,j} \geq -0.5}} \leq \frac{c \cdot d}{\delta}$.\label{fact:bad-events:beta:exp}
	
		\item $\eex{x \sim \Idl_X}{\size{\set{j \in [n] \colon \exists i \in [m]\text{ s.t. } \Unf_{X_{i,j} | X_{<i}=x_{<i}}(\cX_{i,j}) < 0.9}}} \leq \frac{c' \cdot d}{\delta}$.\label{fact:bad-events:beta:2}
	\end{enumerate}
\end{fact}
\begin{proof}
	Note that by definition, $\beta_{i,j} = \frac{1 + \alpha_{i,j}}{1 + \xi_{i,j}}$. Therefore, $\beta_{i,j} \geq 1+\lambda \implies \paren{\alpha_{i,j} > 0.01} \lor \paren{\size{\xi_{i,j}} > 0.01}$. The proof of \cref{fact:bad-events:beta:1} then follows by \cref{fact:bad-events:r}(\ref{fact:bad-events:r:2}) and \cref{fact:bad-events:alpha}. Moreover, note that the proof of \cref{fact:bad-events:beta:2} follows by \cref{fact:bad-events:beta:exp} and \cref{fact:bad-events:r}(\ref{fact:bad-events:r:2}) (for $r=\rho$ and $\lambda = 1/2$). Therefore, it is left to prove \cref{fact:bad-events:beta:exp}. Note that
	\begin{align}\label{eq:bad-event:beta-calc}
		d 
		&\geq \eex{x \sim \Idl_X}{\sum_{i=1}^m \sum_{j=1}^n D(\Idl_{X_{i,j} Y_{i,j} \mid X_{<i}=x_{<i}} || \Unf_{X_{i,j} Y_{i,j} \mid X_{<i}=x_{<i}})}\\
		&\geq \eex{x \sim \Idl_X}{\sum_{i=1}^m \sum_{j=1}^n \Idl_{Y_{i,j} \mid X_{<i}}(1) \cdot D(\Idl_{X_{i,j} \mid X_{<i}=x_{<i}, Y_{i,j}=1} || \Unf_{X_{i,j} \mid X_{<i}=x_{<i}, Y_{i,j}=1})}\nonumber\\
		&= \eex{x \sim \Idl_X}{\sum_{i=1}^m \sum_{j=1}^n (1 + \rho_{i,j})\cdot \delta_{i,j} \cdot D(\Idl_{X_{i,j} \mid X_{<i}=x_{<i},E_{i,j}} || \Unf_{X_{i,j} \mid X_{<i}=x_{<i}})}\nonumber\\
		&\geq \eex{x \sim \Idl_X}{\sum_{i=1}^m \sum_{j=1}^n (1 + \rho_{i,j})\cdot \delta_{i,j} \cdot D(\Idl_{X_{i,j} \mid X_{<i}=x_{<i},E_{i,j}}(\overline{\cX_{i,j}}) || \Unf_{X_{i,j} \mid X_{<i}=x_{<i}}(\overline{\cX_{i,j}}))}\nonumber\\
		&\geq \eex{x \sim \Idl_X}{\sum_{i=1}^m \sum_{j=1}^n (1 + \rho_{i,j})\cdot \delta_{i,j} \cdot \Unf_{X_{i,j} \mid X_{<i} = x_{<i}}(\overline{\cX_{i,j}})}/400,\nonumber\\
		&\geq \eex{x \sim \Idl_X}{\sum_{i=1}^m \sum_{j=1}^n \delta \cdot \Unf_{X_{i,j} \mid X_{<i} = x_{<i}}(\overline{\cX_{i,j}})\cdot \indic{\rho_{i,j} \geq -0.5}}/800,\nonumber
	\end{align}
	which concludes the proof.
	Note that the one before last inequality holds by \cref{fact:prelim:bernoulli-div-est} since (recall that) $$\cX_{i,j} = \set{x_{i,j} \in \Supp(\Unf_{X_{i,j} \mid X_{<i}=x_{<i}}) \colon \Unf_{X_{i,j} \mid \bX_{<i} = \bx_{<i}}(x_{i,j}) / \Idl_{X_{i,j} \mid \bX_{<i} = \bx_{<i}, E_{i,j}}(x_{i,j}) \leq 1.1}$$
	and the equality holds since for any $x_{i,j}$ it holds that 
	\begin{align*}
		\Unf[X_{i,j}=x_{i,j} \mid X_{<i}=x_{<i}, Y_{i,j}=1] 
		&= \frac{\Unf[E_{i,j} \mid X_{\leq i}=x_{\leq i}]\cdot \Unf[X_{i,j}=x_{i,j} \mid X_{<i}=x_{<i}]}{\Unf[E_{i,j}]}\\
		&= \Unf[X_{i,j}=x_{i,j}\mid X_{<i}=x_{<i}],
	\end{align*}
	(recall that by assumption, $\Unf[E_{i,j} \mid X_{\leq i}=x_{\leq i}] = \Unf[E_{i,j}]$ for any fixing of $x_{\leq i}$).
\end{proof}

	\subsubsection{Proving \cref{lemma:bad-events}}\label{sec:bad-events:proof}

We now ready to prove \cref{lemma:bad-events}, restated for convenience below.

\begin{lemma}[Restatement of \cref{lemma:bad-events}]\label{lemma:bad-events-rest}
	\badEventsLemma
\end{lemma}
\begin{proof}
	The proof is divided into three parts. We prove that
	\begin{enumerate}
		\item $\Idl[G_m] = \Idl[G_1 \land \ldots \land G_m] \geq 1 - c \cdot d/\delta n$.\label{item:bad_events:proof:S}
		\item $\Idl[\tB_1 \land \ldots \land \tB_m \mid G_m] \geq 1 - c' \cdot d/\delta n$.\label{item:bad_events:proof:H}
		\item $\sum_{i=1}^m \Idl[T_i \land \pT_i \mid G_i] \geq  1 - 2/n$.\label{item:bad_events:proof:T}
	\end{enumerate}

\paragraph{Proving Part \ref{item:bad_events:proof:S}}
	Note that
	
	\begin{align*}
		\Idl[G_m] = \Idl[\size{\cs_m} \geq 0.9 n] \geq 1 - \Idl[\size{\cB} > 0.1 n]
	\end{align*}
	
	\noindent where $\cB = \cB(x) = \biguplus_{i=1}^n \paren{\cB_i \setminus \cB_{i-1}}$, letting $\cB_0 = \emptyset$ and 
	\begin{align*}
		\cB_i = \cB_i(x)=
		\set{j \in [n] \colon \paren{\size{\rho_{i,j}} > 0.1} \lor \paren{\size{\omega_{i,j}-1} > 0.1} \lor \paren{\beta_{i,j} > 1.1} \lor \Unf_{X_{i,j} \mid X_{<i}=x_{<i}}(\cX_{i,j}) < 0.9}
	\end{align*} 
	for $i \in [m]$. 
	By \cref{fact:bad-events:r}(\ref{fact:bad-events:r:2}), \cref{fact:bad-events:omega} and \cref{fact:bad-events:beta}(\ref{fact:bad-events:beta:1},\ref{fact:bad-events:beta:2}) it holds that
	\begin{align}\label{eq:bad-events:exp-size-B-small}
		\eex{\Idl}{\size{\cB}} \leq c\cdot d/\delta
	\end{align}
	for some universal constant $c > 0$. 
	Therefore, by Markov inequality we deduce that
	\begin{align}\label{eq:goal-S-events}
		\Idl[\size{\cB} > 0.1 n] \leq \frac{10c\cdot d}{\delta n}.
	\end{align}
	which ends the proof of Part~\ref{item:bad_events:proof:S}.
	
\paragraph{Proving Part \ref{item:bad_events:proof:H}}	
	
	By definition of $\tB_i$ it holds that 
	\begin{align}\label{eq:bad-events:tB_i-main}
	\Idl[\neg \tB_i \mid G_m] 
	&= \eex{x_{<i} \sim \Idl_{X_{<i} \mid G_m}}{\Rll[\neg{B_i} \mid B_{<i}, X_{<i}=x_{<i}]}\\
	&= \eex{x_{<i} \sim \Idl_{X_{<i} | G_m}}{\Rll[J \notin \cG_{i}(X_i) \mid J \in \cG_{i-1}, X_{<i}=x_{<i}]}\nonumber\\
	&= \eex{x_{<i} \sim \Idl_{X_{<i} \mid G_m}}{\Rll[\paren{J \notin \cJ_i} \lor \paren{X_{i,J} \notin \cX_{i,J}} \mid J \in \cG_{i-1}, X_{<i}=x_{<i}]}\nonumber\\
	&= \eex{x_{<i} \sim \Idl_{X_{<i} \mid G_m}}{\Rll[J \notin \cJ_i \mid J \in \cG_{i-1}, X_{<i}=x_{<i}] + \Rll[X_{i,J} \notin \cX_{i,J} \mid J \in \cs_{i}, X_{<i}=x_{<i}]}\nonumber
	\end{align}
	
	where in the last equality recall that $\cs_i = \cG_{i-1} \bigcap \cJ_i$.
	In the following, fix $x_{<i} \in \Supp(\Idl_{X_{<i} \mid G_m})$. We first bound the left-hand side term \wrt $x_{<i}$.  Note that by definition, for all $j \in \cG_{i-1}$ it holds that $\omega_{i-1,j}, \beta_{i-1,j}, \paren{1 + \rho_{i-1,j}} \in 1 \pm 0.1$ which yields that
	\begin{align*}
		\omega_{i,j} = \omega_{i-1,j} \cdot \beta_{i-1,j} \cdot \frac{1 + \tau_{i-1,j}}{1 + \rho_{i-1,j}} \leq 2(1 + \tau_{i-1,j})
	\end{align*}
	Moreover, by the event $G_m$ it holds that $\size{\cs_i} \geq  0.9n$ and note that by definition of $\cs_i$ it holds that $\cs_i \subseteq \cG_{i-1}$ and that $\omega_{i,j} \geq 0.9$ for all $j \in \cs_i$.	
	We deduce that
	
	\begin{align}\label{eq:bad-events:tB_i-left}
	\lefteqn{\Rll[J \notin \cJ_i \mid J \in \cG_{i-1}, X_{<i}=x_{<i}]}\nonumber\\
	&= \frac{\sum_{j \in \cG_{i-1} \setminus \cJ_i} \omega_{i,j}}{\sum_{j \in \cG_{i-1}} \omega_{i,j}}
	\leq \frac{\sum_{j \in \cG_{i-1} \setminus \cJ_i} \omega_{i,j}}{\sum_{j \in \cs_{i}} \omega_{i,j}}
	\leq \frac{2\cdot \sum_{j \in \cG_{i-1} \setminus \cJ_i} (1 + \tau_{i,j})}{0.9 \cdot \size{\cs_{i}}}\nonumber\\
	&\leq \frac3{n} \cdot \sum_{j \in \cG_{i-1} \setminus \cJ_i} (1 + \tau_{i,j})
	\leq \frac6{n} \cdot \paren{\size{\cG_{i-1} \setminus \cJ_i} + \sum_{j \in \cG_{i-1} \setminus \cJ_i} \tau_{i,j} \cdot \indic{\tau_{i,j} > 1}}\nonumber\\
	&\leq \frac6{n} \cdot \paren{\size{\cB_i} + \sum_{j=1}^n \min\set{\size{\tau_{i,j}}, \tau_{i,j}^2}}
	\end{align}
	
	We now bound the right-hand side term in \cref{eq:bad-events:tB_i-main} \wrt $x_{<i}$. Compute
	
	\begin{align}\label{eq:bad-events:tB_i-right}
	\lefteqn{\Rll[X_{i,J} \notin \cX_{i,J} \mid J \in \cs_{i}, X_{<i}=x_{<i}]}\nonumber\\
	&\leq \frac{2}{\size{\cs_{i}}}\cdot \sum_{j \in \cs_{i}} \Rll[X_{i,j} \notin \cX_{i,j} \mid J = j, X_{<i}=x_{<i}]\nonumber\\
	&= \frac{2}{\size{\cs_{i}}}\cdot \sum_{j \in \cs_{i}} \Unf_{X_{i,j} \mid X_{<i}=x_{<i}}(\neg{\cX_{i,j}}) 
	= \frac{2}{\size{\cs_{i}}}\cdot \sum_{j \in \cs_{i}} \Unf_{X_{i,j} \mid X_{<i}=x_{<i}}(\neg{\cX_{i,j}})\cdot \indic{\rho_{i,j} > -0.5}\nonumber\\
	&\leq \frac{4}{n}\cdot \sum_{j =1}^n \Unf_{X_{i,j} \mid X_{<i}=x_{<i}}(\neg{\cX_{i,j}})\cdot \indic{\rho_{i,j} > -0.5}
	\end{align}
	The first inequality holds since given $X_{<i}$ and given $J \in \cs_{i}$, then by definition $J$ is distributed (almost) uniformly over $\cs_{i}$ (\ie has high min entropy). The last equality holds since, by definition, for all $j \in \cs_i$ it holds that $\rho_{i,j} > -0.5$. The last inequality holds since the event $G_m$ implies that $\size{\cs_i} \geq 0.9n$. 
	By combining \cref{eq:bad-events:tB_i-main,eq:bad-events:tB_i-left,eq:bad-events:tB_i-right} we deduce that
	\begin{align*}
		\sum_{i=1}^m \Idl[\neg{\tB_i} \mid G_m]
		&\leq \eex{x \sim \Idl_{X | G_m}}{\sum_{i=1}^m \Rll[J \notin \cJ_i \mid J \in \cG_{i-1}, X_{<i}=x_{<i}] + \Rll[X_{i,J} \notin \cX_{i,J} \mid J \in \cG_{i}, X_{<i}=x_{<i}]}\\
		&\leq \frac6{n} \cdot \eex{x \sim \Idl_{X | G_m}}{\sum_{i=1}^n \paren{\size{\cB_i} + \sum_{j=1}^n \min\set{\size{\tau_{i,j}}, \tau_{i,j}^2} + \sum_{j =1}^n \Unf_{X_{i,j} \mid X_{<i}=x_{<i}}(\neg{\cX_{i,j}}) \cdot \indic{\rho_{i,j} > -0.5}}}\\
	\end{align*}
	and the proof of Part \ref{item:bad_events:proof:H} follows by Part \ref{item:bad_events:proof:S}, \cref{fact:bad-events:r}(\ref{fact:bad-events:r:1}), \cref{fact:bad-events:beta}(\ref{fact:bad-events:beta:exp}) and \cref{eq:bad-events:exp-size-B-small} (recall that $\size{\cB} = \sum_{i=1}^n \size{\cB_i}$).
	
\paragraph{Proving Part \ref{item:bad_events:proof:T}}	
	
	Assume (towards a contradiction) that $\exists i \in [m]$ with $\Idl[T_i \mid G_i] \geq \frac1{n^3} \geq \frac{2}{\delta n^4}$ (recall that $n \geq c\cdot m/\delta$ for a large constant $c$ of our choice) and let $\Idl'_{X_{\leq i} Y_i} = \Idl_{X_{<i}} \prod_{j=1}^n \Idl_{X_{i,j} Y_{i,j} \mid X_{<j}}$ (namely, $\Idl'$ behaves as $\Idl$ in the first $i-1$ rows, and in row $i$ it becomes the product of the marginals of $\Idl$ given $X_{<i}$). It holds that
	
	\begin{align}
		\dval \geq D(\Idl_{X_{\leq i} Y_{i}} || \Unf_{X_{\leq i} Y_{i}}) 
		&\geq D(\Idl_{X_{\leq i} Y_{i}} || \Idl'_{X_{\leq i} Y_{i}}) \geq D(\Idl[G_i \land \neg{T_i}] || \Idl'[G_i \land \neg{T_i}])\nonumber\\
		&\geq D\paren{\frac{1}{\delta n^4} || \Idl'[\size{\gamma_i} > 1/2 \mid G_i]}\nonumber\\
		&\geq D\paren{\frac{1}{\delta n^4} || 4 \cdot \exp\paren{-\delta n/400}}
		\geq \frac{\delta n}{500}
	\end{align}
	where the first inequality holds by chain rule and data processing of KL-divergence (recall that $\dval = \sum_{i=1}^{m} D(\Idl_{\bX_i \bY_i} || \Unf_{\bX_i \bY_i} | \Idl_{\bX_{<i}})$), the second one holds by the product case of chain rule, the third one holds by data-processing (indicator to the event $G_i \land \neg{T_i}$) and the fourth one holds by assumption (recall that $\Idl[G_i] \geq  1- O(d/\delta n) \geq 1/2$). The one before last inequality holds by \cref{eq:bounding-div:concet-under-prod} (under product, when $G_i$ occurs, there is a strong concentration), and last inequality holds since $n$ is large enough.  This contradicts the assumption on $\dval$ (by setting the constant there to be larger than $500$). Therefore, we deduce that for all $i \in [m]:$
	\begin{align}\label{eq:bad-events:T}
		\Idl[\neg{T_i} \mid G_i] \leq 1/n^3
	\end{align}
	Moreover, by definition of $\pT_i$ (recall that $T'_i$ is the event that $\Idl[T_i \mid X_{<i}] \geq 1 - 1/n$), it holds that
	\begin{align}\label{eq:bad-events:T'}
		\Idl[\neg{\pT_i} \mid G_i] \leq \frac{\Idl[\neg{T_i} \mid G_i]}{\Idl[\neg{T_i} \mid \neg{\pT_i} \land G_i]} \leq \paren{1/n^3}/\paren{1/n} = 1/n^2.
	\end{align}
	The proof now immediately follows by \cref{eq:bad-events:T,eq:bad-events:T'}.
\remove{
	By Part \cref{item:bad_events:proof:S} and by \cref{eq:bad-events:T,eq:bad-events:T'} we conclude the proof since
	\begin{align*}
		\Idl[T_{\leq m} \land \pT_{\leq m}] 
		&\geq 1 - \sum_{i=1}^m (\Idl[\neg{T_{i}}]+\Idl[\neg{\pT_{i}}])\\
		&\geq  1 - \sum_{i=1}^m (\Idl[\neg{T_{i}} \mid G_i]+\Idl[\neg{\pT_{i}} \mid G_i] + 2\Idl[\neg{G_i}])\\
		&\geq 1 - 2/n - 2c\cdot d/\delta n.
	\end{align*}
}
\end{proof}

\newcommand{\CS}{CS}

\section{Lower Bound}\label{sec:LowerBound}

In this section we formally state and prove \cref{thm:intro:LowerBound}, showing that \cref{thm:PR} is tight for partially prefix-simulatable interactive arguments. In \cref{sec:counter-example:why-rand-ter} we start by showing how random termination helps to beat  \cite{BellareIN97}'s counterexample, and in \cref{sec:counter-example:proving-lowerbound} we restate and prove \cref{thm:intro:LowerBound} using a variant of \cite{BellareIN97}'s protocol.

\subsection{Random Termination Beats Counterexample of \texorpdfstring{\cite{BellareIN97}}{BellareIN97}}\label{sec:counter-example:why-rand-ter}
In this section we exemplify the power of random termination,  showing  that the counterexample of \cite{BellareIN97} does not apply to random-terminating verifiers.  We do so by presenting \cite{BellareIN97}'s counterexample against $n$ repetitions and see how random termination helps in this case. The protocol is described below.

\begin{protocol}[\cite{BellareIN97}'s Protocol $\pi = (\P,\V)$]
	\item[Common input:] Security parameter $1^{\kappa}$ and public key $pk$ .
	\item[Prover's private input:] Secret key $sk$.
	\item[Operation:]~
	\begin{enumerate}
		\item Round $1$:
		\begin{enumerate}
			\item $\V$ uniformly samples  $b \la \zo$ and $r \la \zo^{\kappa}$, and sends  $B = \Enc_{pk}(b,r)$ to $\P$.
			\item $\P$ computes $(b,r) = \Dec_{sk}(B)$ and for any $i \in [n-1]$, it uniformly samples  $b'_i \in \zo$ and $r'_i \in \zo^{\kappa}$ conditioned on  $b = \oplus_{i=1}^{n-1} b'_i$. Then it computes $C_i = \Enc_{pk}(b'_i,r'_i)$, and sends $(C_1,\ldots, C_{n-1})$ to $\V$.
		\end{enumerate}
		\item Round $2$:
		\begin{enumerate}
			\item $\V$ sends  $(b,r)$ to $\P$.
			\item $\P$ sends $\left((b'_1,r'_1),\ldots,(b'_{n-1},r'_{n-1})\right)$ to $\V$.
		\end{enumerate}
		\item At the end: $\V$ accepts iff $b = \oplus_{i=1}^{n-1} b'_i$, and for any $i \in [n-1]$: $C_i = \Enc_{pk}(b'_i,r'_i)$ and $B \neq C_i$.
	\end{enumerate}
\end{protocol}

Intuitively, assuming the cryptosystem is CCA2-secure, if a single instance of the protocol is run, then a prover without
access to $sk$ can only convince the honest verifier with probability $1/2$, since it
must commit itself to a guess $\oplus_{i=1}^{n-1} b'_i$ of $b$ before receiving $(b, r)$. On the other hand, if $n$ instances of the protocol are run in parallel, then a cheating prover can send the tuple $(C_1,\ldots,C_{n-1}) = (B_1,\ldots,B_{i-1},B_{i+1},\ldots,B_n)$ to $V_i$ and then either all verifier instances accept or all verifier instances fail, the first event occurring with probability at least $1/2$.

Let's look now on a $n$ instances that run in parallel of the protocol $\pi = (\P,\tV)$, where $\tV$ is the random-terminating variant of $\V$ (note that this protocol has only two rounds, and therefore, a random terminating bit takes one with probability $1/2$). First, we expect that $\approx \zfrac{n}{2}$ of the verifiers abort at the first round, and with high probability at least $n/4$ of the verifiers remain active (assume that $n$ is large enough). For a cheating prover, aborting at the first round is not an issue since it can completely simulate the aborted verifiers. However, even if a single verifier $V_i$ aborts at the second round, then the attack presented above 
completely fail since the prover has no way to reveal $(b_i,r_i)$, needed for the other verifiers. Note that the attack do succeed in case non of the verifiers abort at the second round, but the probability of this to happen is at most $2^{-n/4}$.

\subsection{Proving \texorpdfstring{\cref{thm:intro:LowerBound}}{Our Counterexample}}\label{sec:counter-example:proving-lowerbound}

We now restate and prove \cref{thm:intro:LowerBound}.

\begin{theorem}[lower bound, retstment of \cref{thm:intro:LowerBound}.]\label{thm:LowerBound}
	Assume the existence of CCA2-secure public-key cryptosystem. Then for every $m = m(\kappa) \in [2,\poly(\kappa)]$ and $\eps = \eps(\kappa) \in [\zfrac1{\poly(\kappa)},\zfrac13]$ and $n  = n(\kappa) \in [\zfrac{m}{\eps}, \poly(\kappa)]$, there exists an $m$-round interactive argument $(\P,\V)$ with soundness error $1-\eps$ such that $(\P^n,\tV^n)$ has soundness error of at least  $(1-\eps)^{c\cdot n/m}$ for some universal constant $c>0$, where $\tV$ is the $\zfrac1{m}$-random-terminating variant of $\V$ (according to \cref{def:prelim:RT}) and  $(\P^n,\tV^n)$ denotes the $n$-parallel repetition of $(\P,\tV)$.\footnote{Assuming the existence of collision-free family of hash functions and CCA2-secure cryptosystem with respect to superpolynomial adversaries, one can adopt the techniques used in \cite{PietrzakW12} for constructing a single protocol $(\P,\V)$ such that for any polynomial bounded $n$, $(\P^n,\tV^n)$ has soundness error of at least $(1-\eps)^{c\cdot n/m}$. This, however, is beyond the scope of this paper.}
\end{theorem}

Fix large enough $\kappa$ and fix $m,\eps,n$ as in the theorem statements, and let $\CS = (\Gen,\Enc,\Dec)$ be a CCA2-secure public-key cryptosystem. Consider the following $m$-round variant $(\P,\V)$ of \cite{BellareIN97}'s protocol:

\begin{protocol}[The counterexample protocol $\pi = (\P,\V)$]\label{prot:LowerBound}
	\item[Common input:] Security parameter $1^{\kappa}$ and public key $pk$ .
	\item[Prover's private input:] Secret key $sk$.
	\item[Operation:]~
	\begin{enumerate}
		\item Round $1$:
		\begin{enumerate}
			\item $\V$ flips a coin that takes one with probability $1-3\eps$ and zero otherwise. \label{step:counter:bit}

			If the coin outcome is one, $\V$ sends $\perp$ to $\P$, accepts and the protocol terminates.

			Else,   $\V$ uniformly samples  $b \la \zo$ and $r \la \zo^{\kappa}$, and sends  $B = \Enc_{pk}(b,r)$ to $\P$.\label{step:first}
			
			\item $\P$ computes $(b,r) = \Dec_{sk}(B)$ and for any $i \in [n-1]$, it uniformly samples  $b'_i \in \zo$ and $r'_i \in \zo^{\kappa}$ conditioned on  $b = \oplus_{i=1}^{n-1} b'_i$. Then it computes $C_i = \Enc_{pk}(b'_i,r'_i)$, and sends $(C_1,\ldots, C_{n-1})$ to $\V$.
			
		\end{enumerate}

		\item Round $2$:
		\begin{enumerate}
			\item $\V$ sends  $(b,r)$ to $\P$.\label{step:ver-rev}

			\item $\P$ sends $\left((b'_1,r'_1),\ldots,(b'_{n-1},r'_{n-1})\right)$ to $\V$.\label{step:prov-rev}
		\end{enumerate}

		\item Rounds $3$ to $m$: parties exchange  dummy messages.\label{step:dummy-rounds}

		\item At the end: $\V$ accepts iff $b = \oplus_{i=1}^{n-1} b'_i$, and for every $i \in [n-1]$: $C_i = \Enc_{pk}(b'_i,r'_i)$ and $B \neq C_i$.
	\end{enumerate}
\end{protocol}

Namely, \cref{prot:LowerBound} first transforms \cite{BellareIN97}'s
two-rounds protocol, of soundness error $\zfrac12 + \negl(\kappa)$, into an $m$-round
protocol with soundness error $1-\eps$, by flipping a coin at
\stepref{step:first} (for increasing the soundness error) and adding dummy rounds at the end for increasing the number of rounds (\stepref{step:dummy-rounds}).\footnote{As in \cite{BellareIN97,PietrzakW12}, the soundness error holds with respect to a prover without access to $sk$.}\remove{We slightly changed \cite{BellareIN97}'s protocol by letting the prover to choose $t$: the number of encrypted elements that the verifier expect to receive. While the honest prover always sets this value to $1$, the attacker we present next sets it to a different value.}

We first note that soundness error of $\pi$ is indeed low.
\begin{claim}
	The soundness error of $\pi(1^{\kappa})$ is at most $1-\eps$.
\end{claim}
\begin{proof}
	Let $\sP$ be some efficient cheating prover and let $T$ be the event over a random execution of $(\sP, \V)$ that the outcome of the $(1-3\eps,3\eps)$ bit (flipped by $\V$ at Step~\ref{step:counter:bit}) is $0$ (\ie $\V$ does not abort). Conditioned on $T$,  $\sP$ must commit itself to a guess $\oplus_{i=1}^{n-1} b'_i$ before receiving $(b,r)$.  Since the encryption scheme is CCA2-secure (which implies non-malleability), we obtain that
	\begin{align*}
		\ppr{(pk,sk) \la \Gen(1^{\kappa})}{(\sP, \V)(1^{\kappa}, pk)=1 \mid T} \leq \zfrac12 + \negl(\kappa),
	\end{align*}
	and hence
	\begin{align*}
		\ppr{(pk,sk) \la \Gen(1^{\kappa})}{(\sP, \V)(1^{\kappa})=1}
		&\leq \pr{\neg T} + \pr{T} \cdot \ppr{(pk,sk) \la \Gen(1^{\kappa})}{(\sP, \V)(1^{\kappa},pk)=1 \mid T}\\
		&\leq 1-3\eps + 3\eps \cdot ( \zfrac12 + \negl(\kappa))\\
		&\leq 1 - \eps.
	\end{align*}
\end{proof}

So it is left to show that the soundness error of the $n$ parallel repetition of the random terminating variant of $\pi$ is high. Let $\tV$ and  $(\P^n,\tV^n)$  be as in the theorem statement with respect to $(\P,\V)$ (\cref{prot:LowerBound}) and assume \wlg that $\tV$ sends $\perp$ to the prover right after flipping a termination coin with outcome one. Consider the following cheating prover $\sP$:

\begin{algorithm}[Cheating prover $\nsP$]\label{alg:attacker:LowerBound}
	\item[Input:] Security parameter $1^{\kappa}$.
	\item[Operation:]~
	\begin{enumerate}
		\item Upon receiving a $n$-tuple $(a_1,\ldots,a_n)$ from $\ntV  = (\tV_1,\ldots,\tV_n)$, 		
		let $\cs = \set{i \in [n] \colon a_i \neq \perp}$ (the set of active verifiers) and for $i \notin \cs$ sample uniformly $b_i \la \zo$ and $r_i \la \zo^{\kappa}$. Then for any $i \in \cs$ send $(a'_1,\ldots,a'_{i-1},a'_{i+1},\ldots,a'_n)$ to $\tV_i$, where $a'_{j} = \begin{cases} a_{j} & j \in \cs \\ \Enc_{pk}(b_j,r_j) & o.w\end{cases}$.
		
		\item If at least one verifier in $\cs$ sends $\perp$ (after aborting at the second round), fail and abort. Otherwise, upon receiving $(b_{i},r_{i})$ for all $i \in \cs$, send  the tuple $\paren{(b_{1},r_{1}),\ldots,(b_{i-1},r_{i-1}),(b_{i+1},r_{i+1}),\ldots,(b_{n},r_{n})}$ to $\tV_{i}$.
		
		\remove{
		\item Upon receiving a $n$-tuple $(a_1,\ldots,a_n)$ from $\ntV = (\tV_1,\ldots,\tV_n)$, set $t = \ell - 1$ for $\ell = \size{\set{i \in [k] \colon a_i \neq \perp}}$ (the number of active verifiers after the first round).  For each $j \in [\ell]$, send the tuple $(t, a_{i_1},\ldots,a_{i_{j-1}},a_{i_{j+1}},\ldots,a_{i_{\ell}})$ to $\tV_{i_j}$, for $i_j$ being the \jth active verifier.\label{step:first-step-attack}

		\item If one of the $\ell$ active verifiers sends $\perp$ (after aborting at the second round), fail and abort. Otherwise, upon receiving $(b_{i_j},r_{i_j})$ for all $j \in [\ell]$, send  the tuple $\paren{(b_{i_1},r_{i_1}),\ldots,(b_{i_{j-1}},r_{i_{j-1}}),(b_{i_{j+1}},r_{i_{j+1}}),\ldots,(b_{i_{\ell}},r_{i_{\ell}})}$ to $\tV_{i_j}$.
	    }
	\end{enumerate}
\end{algorithm}

Namely, $\nsP$ performs \cite{BellareIN97}'s attack on the verifiers that remain active after the first round. The attack, however, can only be performed if none of these active verifiers abort in the second round.  Yet, we show that the probability for this to happen is high enough. The following claim conclude the proof of  \cref{thm:LowerBound}.

\begin{claim}\label{claim:LowerBound}
	Let $\eps, m, n$ as in the theorem statement, let $(\P,\V)$ be \cref{prot:LowerBound} and let $\nsP$ be the cheating prover described in \cref{alg:attacker:LowerBound} (with respect to $n$). Then
	\begin{align*}
		\ppr{(pk,sk) \la \Gen(1^{\kappa})}{(\nsP,\ntV)(1^{\kappa}, pk) = 1} \geq (1-\eps)^{14\cdot \zfrac{n}{m}}.
	\end{align*}
\end{claim}

\begin{proof}
	Fix $pk$ and let $L$ be the random variable that denotes the value of $\size{\cs}$ (the number of active verifiers after the first round) in a random execution of $(\nsP,\ntV)(1^{\kappa}, pk)$. Note that each verifier aborts with probability greater than $1-3\eps$ at the first round (it can abort by the $(1-3\eps,3\eps)$ coin or by the $(\zfrac1m,1-\zfrac1m)$ random-terminating coin). Therefore, $\ex{L} \leq 3 \eps n$ and we obtain by Markov's inequality that $\pr{L \leq  6 \eps n} \geq \zfrac12$. Let $G$ be the event that none of the verifiers abort at the second round. Note that
	\begin{align}\label{eq:LowerBound:G}
		\pr{G}
		&\geq \pr{L \leq  6 \eps n}\cdot \pr{G | L \leq  6 \eps n}\\
		&\geq \zfrac12 \cdot (1 - \zfrac1{m})^{6 \eps n}\nonumber\\
		&\geq \zfrac12 \cdot \exp\paren{-\zfrac{12 \eps n}{m}}.\nonumber
	\end{align}
	The second inequality holds since $1-x \geq e^{-2x}$ for $x \in [0,\zfrac12]$.
	In addition, observe that
	\begin{align}\label{eq:LowerBound:cond-G}
		\pr{(\nsP,\ntV)(1^{\kappa}, pk) = 1 \mid G}
		&\geq \ppr{(b_1,\ldots,b_n) \la \zo^n}{\oplus_{i=1}^n b_{i} = 0} - \negl(\kappa)\\
		&= \zfrac12 -  \negl(\kappa)\nonumber
	\end{align}

	and we conclude by \cref{eq:LowerBound:G,eq:LowerBound:cond-G} that

	\begin{align*}
		\pr{(\nsP,\ntV)(1^{\kappa}, pk) = 1}
		&\geq \pr{G} \cdot \pr{(\nsP,\ntV)(1^{\kappa}, pk) = 1 \mid G}\\
		&\geq \zfrac12 \cdot \exp\paren{-\zfrac{12 \eps n}{m}} \cdot \left(\zfrac12 -  \negl(\kappa)\right)\\
		&\geq \exp\paren{-\zfrac{14 \eps n}{m}}\\
		&\geq (1-\eps)^{\zfrac{14 n}{m}}.
	\end{align*}
	The penultimate inequality holds since we assumed that
        $n \geq \zfrac{m}{\eps}$, and the last one since $1+x \leq e^x$ for any
        $x \in \mathbb{R}$.
\end{proof}

\paragraph{Putting it together.}
\begin{proof}[Proof of \cref{thm:LowerBound}]
	Immediate by \cref{claim:LowerBound}.
\end{proof}

\subsection*{Acknowledgment}
We thank Chris Brzuska, Or Ordentlich and Yury Polyanskiy for very useful
discussions.

\printbibliography

\section{Missing Proofs}\label{sec:missinProofs}

\subsection{Proof of \texorpdfstring{\cref{prop:prelim:smooth-div}}{Small Events Map To Small Events}}\label{sec:appendix:smooth-div}
\begin{proposition}[Restatement of \cref{prop:prelim:smooth-div}]
	\PropSmallToSmallEvents
\end{proposition}
\begin{proof}
	
	We  assume that $\max\set{\alpha + P[E], 4\beta} \le 1/2$, as otherwise the proof holds trivially.  The definition of smooth KL-divergence yields the existence of  randomized function $F_P,F_Q$ satisfying
	\begin{enumerate}[a.]
		
		\item $\ppr{x\sim P}{F_P(x) \neq x} \leq \alpha$, \label{item:f-property:2}

		\item $\re{F_P(P)}{F_Q(Q)} < \beta$, and \label{item:f-property:1}

		\item $\forall x\in\Uni$:\  $\Supp(F_P(x)) \cap \Uni \subseteq{\set{x}}$ and $\Supp(F_Q(x)) \cap \Uni \subseteq{\set{x}}$.\label{item:f-property:3}
	\end{enumerate}
	Let $E' = E \cup \paren{\Supp(F_P(\Uni)) \cup \Supp(F_Q(\Uni) )\setminus  \Uni}$. 
	By  \cref{item:f-property:2}  and data processing of (standard) KL-divergence,
	\begin{align}
	\re{\indic{F_P(P)\in  E'}}{\indic{ F_Q(Q)\in E'}} <  \beta
	\end{align}
	By  \cref{item:f-property:1,item:f-property:3}, 
	\begin{align}
	\F_P(P)[ E' ]  \le \ppr{x\sim P}{F_P(x) \neq x}  + P[E] \le \alpha + P[E] 
	\end{align}	
	Assume toward a contradiction that $\F_Q(Q)[E'] \ge 2\cdot \max\set{\alpha + P[E], 4\beta}$, then by the above equations
	
	\begin{align}\label{eq:prelim:low-events}
	\re{\alpha + P[E] }{ 2\cdot \max\set{\alpha + P[E], 4\beta}} <  \beta
	\end{align}
	If $\alpha + P[E] > 4\beta$, then \cref{eq:prelim:low-events} yields that $\re{\alpha + P[E] }{ 2(\alpha + P[E])} <  \beta$. Otherwise, \cref{eq:prelim:low-events} yields that $\re{4\beta }{ 8\beta} <  \beta$. In both cases we get a contradiction to \cref{fact:prelim:bernoulli-div-est}(\ref{fact:prelim:bernoulli-div-est:minus}). 
	Since  by \cref{item:f-property:3} it holds that $Q[ E ]  \le \F_Q(Q)[ E' ]$, we conclude that $Q[ E ]  < 2\cdot \max\set{\alpha + P[E], 4\beta}$.
\end{proof}

\subsection{Proof of \texorpdfstring{\cref{prop:prelim:smooth-DP}}{Data Processing Smooth KL-Divergence}}\label{sec:appendix:prop:prelim:smooth-DP}

\begin{proposition}[Restatement of \cref{prop:prelim:smooth-DP}]
	\DataProcessSmoothKL
\end{proposition}
\begin{proof} Let $(F_P,F_Q)$ be a pair  of functions such that
	\begin{enumerate}
		\item $\Pr_{x\sim P}[ F_P(x) \neq x]\leq \alpha$, and
		
		\item $\forall x\in\Uni$:  $\Supp(F_P(x)) \cap \Uni \subseteq{\set{x}}$ and  $\Supp(F_Q(x)) \cap \Uni \subseteq{\set{x}}$.
	\end{enumerate}
	We assume \wlg that   for both $T\in \set{P,Q}$:
	\begin{align}\label{eq:SmoothDP:1}
	\forall x\in\Uni: \ \Supp(F_T(x)) \cap \Supp(H(x)) \subseteq{\set{x}}.
	\end{align}
	Indeed, since $F_T(x) \neq x$ implies $F_T(x) \notin \Uni$,  one can add a fixed prefix to the value of $F_T(x)$ when $F_T(x) \neq x$ (same prefix for both $T\in \set{P,Q}$) such that \cref{eq:SmoothDP:1} holds. Such a change neither  effect the properties of $F_P$ and $F_Q$ stated above, nor  the value of $D(F_P(P) || F_Q(Q))$.
	
	For $T\in \set{P,Q}$, let $G_T(y)$ be the randomized function defined by the
	following process:
	\begin{enumerate}[a.]
		\item Sample $x\sim T_{X|H(X)=y}$.
		\item Sample $z \sim F_T(x)$.\label{step:smooth-kl:apply-FT}
		\item If $z = x$, output $y$.
		
		Else, output $z$.
	\end{enumerate}
	\remove{
		, $R_T$  and $O_T$, be the values of $x$ and $r$, and the final output, receptively, in a random execution of  $G_T(Y_T)$. It is clear that $(X_T,R_T) \sim (T,R)$.}
	
	\noindent By construction and \cref{eq:SmoothDP:1}, for both  $T\in \set{P,Q}$:
	\begin{align}\label{eq:SmoothDP:G_T-second-prop}
	\forall y\in H(\Uni):  \Supp(G_T(y)) \cap \H(\Uni) \subseteq{\set{y}}.
	\end{align}
	
	Let $Y_T = H(T)$ and let $X_T$ be the value of $x$ in a random execution of
	$G_T(Y_T)$. It is clear that $X_T \sim T$. We note that
	\begin{align}\label{eq:SmoothDP:G_T-first-prop}
	\pr{G_P(Y_P) \neq Y_P} &= \pr{F_P(X_P) \neq X_P}\\
	&=Pr_{x\sim P}[F_P(x) \neq x]\nonumber\\
	& \leq \alpha.\nonumber
	\end{align}
	The inequality is by the assumption about $F_P$.
	
	Consider the randomized function $K(z)$ that outputs $H(z)$ if
	$z\in\Uni$, and otherwise outputs $z$. It holds that
	\begin{align*}
	\Pr[K(F_T(T))=z]&=\Pr[F_T(T)\in\Uni]\cdot\Pr[H(F_T(T))=z|F_T(T)\in\Uni]\\
	&\quad+\Pr[F_T(T)\notin\Uni]\cdot\Pr[F_T(T)=z|F_T(T)\notin\Uni]\\
	&=\Pr[F_T(T) = T]\cdot\Pr[H(T)=z|F_T(T) = T]\\
	&\quad+\Pr[F_T(T) \neq T]\cdot\Pr[F_T(T)=z|F_T(T) \neq T],
	\end{align*}
	where the second inequality follows from the second property of $(F_P,F_Q)$;
	namely, $F_T(T)\in\Uni\Longleftrightarrow F_T(T)=T$. Similarly,
	\begin{align*}
	\Pr[G_T(H(T))=z]&=\Pr[F_T(X_T)=X_T]\cdot\Pr[H(X_T)=z|F_T(X_T)=X_T]\\
	&\quad+\Pr[F_T(X_T)=X_T]\cdot\Pr[F_T(X_T)=z|F_T(X_T)\neq X_T].\\
	&=\Pr[F_T(T) = T]\cdot\Pr[H(T)=z|F_T(T) = T]\\
	&\quad+\Pr[F_T(T) \neq T]\cdot\Pr[F_T(T)=z|F_T(T) \neq T],
	\end{align*}
	where the second inequality holds since $X_T \sim T$. Hence, we have $G_T(H(T)) \equiv K(F_T(T))$. Thus, the data-processing inequality for
	(standard) KL-divergence implies that
	\begin{align}\label{eq:SmoothDP:G_T-ineq}
	\diver(F_P(P) || F_Q(Q)) &\geq \diver(K(F_P(P)) || K(F_Q(Q)))\\
	&=\diver(G_P(H(P)) || G_Q(H(Q))).\nonumber
	\end{align}
	The proof now follows by Properties (\ref{eq:SmoothDP:G_T-second-prop}), (\ref{eq:SmoothDP:G_T-first-prop}), (\ref{eq:SmoothDP:G_T-ineq}) of $G_P$ and $G_Q$.
\end{proof}

\subsection{Proof of \texorpdfstring{\cref{prop:prelim:sub-exp-to-divergence}}{sub exponential to divergnce proposition}}\label{sec:appendix:sub-exp-to-divergence}

\begin{proposition}[Restatement of \cref{prop:prelim:sub-exp-to-divergence}]
	\propSubExpToDiv
\end{proposition}
Note that for $\sigma\geq 1$, the statement is trivial, and thus not
interesting. We would use this proposition when $\sigma \ll 1$.

\begin{proof}
  Assume that $\sigma^2\leq 1$ and that $D(P||Q)<\infty$, since otherwise the
  statement is trivial.  We use the following two fundamental theorems. The
  first theorem gives a variational characterization for divergence that is
  useful for bounding expected values of random variables.
  \begin{theorem}[Donsker-Varadhan; {cf. \cite[Theorem 3.5]{PW13}}]
    \label{thm:Donsker-Varadhan}
    Let $P$ and $Q$ be probability measures on $\mathcal{X}$ and let
    $\mathcal{C}$ denote the set of functions $f\colon\mathcal{X}\to\mathbb{R}$ such
    that $\Ex_{Q}[\exp(f(X))]<\infty$. If $D(P||Q)<\infty$, then
    \begin{align*}
      D(P||Q) = \sup_{f\in\mathcal{C}}\Ex_{P}[f(X)]-\log\Ex_{Q}[\exp(f(X))].
    \end{align*}
    In particular, for every $f\in\mathcal{C}$, it holds that
    \begin{align*}
      \Ex_{P}[f(X)] \leq \log\Ex_{Q}[\exp(f(X))] + D(P||Q).
    \end{align*}
  \end{theorem}
  The second theorem is the super-exponential moment characterization condition
  for sub-Gaussianity.

  \begin{theorem}[Sub-Gaussian characterization; {cf. \cite[Theorem 3.10]{Duchi13}}\footnote{While the statement of \cite[Theorem
    3.10]{Duchi13} explicitly take $K_2'=2$ and require that $X$'s mean is zero,
    it is easy to see how to modify the proof to work with any constant $K_2'$
    and that the proof of this part does not actually use that $X$ has a zero
    mean. For example, see \cite[Lemma 5.5]{2010arXiv1011.3027V} that uses
    $K_2'=e$ and does not assume that $X$ has zero mean.}]
    \label{thm:sub-gauss-equiv}
    Let $X$ be a random variable and $\sigma^2>0$ be a constant. Assume that
    there exist $K_1',K_2'>0$ such that
    \begin{align*}
      \Pr[\abs{X}\geq t] \leq K_2'\cdot \exp\paren{-\frac{t^2}{K_1'\sigma^2}}
      \quad\text{for all $t \geq 0$}.
    \end{align*}
    Then, there exists $K_3'=K_3'(K_1',K_2')$ such that
    \begin{align*}
      \Ex\left[\exp\paren{\frac{X^2}{K_3'\sigma^2}}\right]\leq e.
    \end{align*}
  \end{theorem}

  We would like to apply the above theorems to derive the proof. However, under
  the $Q$ distribution $X$ is not sub-Gaussian, since its concentration bound
  apply only for $0\leq t \leq 1$. Instead, we let $\mathcal{W}=[0,1]$,
  $K_2'=K_2/(1-\eps)$ and observe that
  \begin{align*}
    \Pr_Q[\abs{X}\geq t \mid \abs{X}\in \mathcal{W}] \leq K_2'\cdot \exp\paren{-\frac{t^2}{K_1\sigma^2}}
      \quad\text{for all $t \geq 0$}.
  \end{align*}
  Indeed, for $t>1$ this inequality holds trivially. For $0\leq t\leq 1$, it holds that
  \begin{align*}
    \Pr_Q[\abs{X}\geq t \mid \abs{X}\in \mathcal{W}] &\leq \frac{\Pr_Q[\abs{X}\geq t]}{\Pr_Q[\abs{X}\in \mathcal{W}]}\\
                                          &\leq \frac{\Pr_Q[\abs{X}\geq t]}{1-\eps}\\
                                          &\leq K_2'\cdot \exp\paren{-\frac{t^2}{K_1\sigma^2}},
  \end{align*}
  where the second inequality follows from the assumption of the proposition and
  since $\sigma^2\leq 1$, and the third inequality again follows from the
  assumption of the proposition.

  Let $K_3=K_3'(K_1,K_2')$ from the statement of
  \cref{thm:sub-gauss-equiv}. Furthermore, note that
  $D(P_X||Q_{X|(\abs{X}\in\mathcal{W})})<\infty$, since $D(P_X||Q_X)<\infty$ and
  $\abs{X}\in\mathcal{W}$ under $P$ almost surely. Using
  \cref{thm:sub-gauss-equiv,thm:Donsker-Varadhan}, it follows that
  \begin{align*}
    \frac{1}{K_2\sigma^2}\Ex_P[X^2] &\leq \log\Ex_Q[\exp(X^2/(K_2\sigma^2))|\abs{X}\in\mathcal{W}] + D(P_X||Q_{X|(\abs{X}\in\mathcal{W})})\\
                                    &\leq \log e + D(P_X||Q_{X|(\abs{X}\in\mathcal{W})}).
  \end{align*}

  Finally, the proposition follows since
  \begin{align*}
    D(P_X||Q_{X|(\abs{X}\in\mathcal{W})}) &= \Ex_{x\sim P_X}\log\frac{P_X(x)}{Q_X(x)/\Pr_Q[\abs{X}\in\mathcal{W}]}\\
                                          &= D(P_X||Q_X) + \log(\Pr_Q[\abs{X}\in\mathcal{W}])\\
                                          &\leq D(P_X||Q_X),
  \end{align*}
  where in the first equality we again used that $|x|\in\mathcal{W}$ for every $x\in\Supp(P_X)$, so $\Pr_Q[X=x\land |X|\in\mathcal{W}]=Q_X(x)$ for any
  such $x$.
\end{proof}

\subsection{Proof of \texorpdfstring{\cref{lemma:prelim:Martingales-new-bound}}{Martingales Lemma}}\label{sec:appendix:Martingales-new-bound}

\begin{proposition}[Restatement of \cref{lemma:prelim:Martingales-new-bound}]
	\MartingalesLemma
\end{proposition}

We use the following fact.
\begin{fact}[{\cite[Theorem 14.9]{DasGupta}}]\label{fact:prelim:martingales:week_concent}
	Let $Y_0 = 0,Y_1,\ldots, Y_n$ be a martingale sequence \wrt $X_0,X_1,\ldots,X_n$, and assume that $\ex{Y_i^2} < \infty$ for all $i \in [n]$. Then for every $\lambda > 0$, it holds that
	\begin{align*}
	\pr{\max_{i\in[n]}\size{Y_i} \geq \lambda} \leq \frac{\ex{\sum_{i=1}^n D_i^2}}{\lambda^2},
	\end{align*}
	for $D_i = Y_i - Y_{i-1}$.
\end{fact}

\begin{proof}[Proof of \cref{lemma:prelim:Martingales-new-bound}.]

	Let $\mu = \ex{\sum_{i=1}^n \min\set{\size{R_i}, R_i^2}}$ and assume \wlg that $\mu \leq 0.1$ (otherwise the proof holds trivially). For $i\in[n]$ let $\Delta_i = \ex{R_i \cdot \indic{R_i > 1} \mid X_0,\ldots,X_{i-1}}$,  let  
	\begin{align*}
	\hR_i = 
	\begin{cases} R_i\cdot \indic{{\size{R_i}} \leq 1} + \Delta_i & \Delta_i \leq 1 \\
	0 & otherwise,
	\end{cases}
	\end{align*}
	and let $\hS_i = \sum_{j=1}^i \hR_j$. Note that for any $i \in [n]$ and a fixing of $X_0,\ldots,X_{i-1}$ such that $\Delta_i \leq 1$, it holds that
	\begin{align*}
	\ex{\hS_i \mid X_0,\ldots,X_{i-1}} - \hS_{i-1}
	&= \ex{\hR_i \mid X_0,\ldots,X_{i-1}}\\
	&= \ex{R_i\cdot \indic{\size{R_i} \leq 1} + \Delta_i \mid X_0,\ldots,X_{i-1}}\\
	&= \pr{\size{R_i} \leq 1 \mid X_0,\ldots,X_{i-1}}\cdot \ex{R_i \mid X_0,\ldots,X_{i-1}, \paren{\size{R_i} \leq 1}}\\
	&\quad + \pr{R_i > 1 \mid X_0,\ldots,X_{i-1}}\cdot \ex{R_i \mid X_0,\ldots,X_{i-1}, \paren{R_i>1}}\\
	&= \ex{R_i \mid X_0,\ldots,X_{i-1}} = 0.
	\end{align*}
	The penultimate  equality holds since $R_i \geq -1$.
	By definition,  for any fixing of $X_0,\ldots,X_{i-1}$ such that $\Delta_i > 1$, it holds that   $\hR_i = 0$. Hence, $\ex{\hS_i \mid X_0,\ldots,X_{i-1}} = \hS_{i-1}$ also for any such fixing. Thus, the sequence $\hS_1,\ldots,\hS_n$ is a martingale \wrt $X_1,\ldots,X_n$ for any fixing of $X_0$.

	By \cref{fact:prelim:martingales:week_concent},
	\begin{align}\label{eq:prelim:prob_max_size_hSi}
	\forall \beta > 0:\text{  }\pr{\max_{i\in[n]}\ssize{\hS_i} \geq \beta}
	\leq \frac{\ex{\sum_{i=1}^n \hR_i^2}}{\beta^2}
	\end{align}
	In addition, note that
	\begin{align}\label{eq:prelim:exp-sum-deltas}
	\ex{\sum_{i=1}^n \Delta_i}
	&= \sum_{i=1}^n \eex{X_0,\ldots,X_{i-1}}{\ex{R_i \cdot \indic{R_i > 1} \mid X_0,\ldots, X_{i-1}}}\\
	&\leq \sum_{i=1}^n \eex{X_0,\ldots,X_{i-1}}{\ex{\min\set{\size{R_i}, R_i^2} \mid X_0,\ldots, X_{i-1}}}\nonumber\\
	&= \mu.\nonumber
	\end{align}
	\remove{
		and by Markov inequality we get
		\begin{align}\label{eq:prelim:prop-sum-deltas}
		\forall \beta > 0:\pr{\sum_{i=1}^n \Delta_i \geq \beta} 
		\leq \frac{\mu}{\beta}.
		\end{align}
	}
	Therefore,
	\begin{align}\label{eq:prelim:exp-sum-hRs-squares}
	\ex{\sum_{i=1}^n \hR_i^2}
	&\leq \ex{\sum_{i=1}^n \paren{R_i\cdot \indic{\size{R_i} \leq 1} + \Delta_i \cdot \indic{\Delta_i \leq 1}}^2}\\
	&\leq 2 \cdot \ex{\sum_{i=1}^n R_i^2\cdot \indic{\size{R_i} \leq 1}} + 2 \cdot \ex{\sum_{i=1}^n \Delta_i^2 \cdot \indic{\Delta_i \leq 1}}\nonumber\\
	&\leq 2 \mu + 2 \cdot \ex{\sum_{i=1}^n \Delta_i \cdot \indic{\Delta_i \leq 1}}
	\leq 2 \mu + 2 \cdot \ex{\sum_{i=1}^n \Delta_i} \nonumber\\
	&\leq 4 \mu.\nonumber
	\end{align}
	The last inequality holds by \cref{eq:prelim:exp-sum-deltas}. Combining \cref{eq:prelim:exp-sum-hRs-squares,eq:prelim:prob_max_size_hSi} yields that
	\begin{align}\label{eq:prelim:prob_max_size_hSi_strong}
	\forall \beta > 0:\text{  }\pr{\max_{i\in[n]}\ssize{\hS_i} \geq \beta}
	\leq {4 \mu}/{\beta^2}
	\end{align}
	Let $S_i = \sum_{j=1}^i R_j$. Note that  for any $i \in [n]$:
	\begin{align*}
	\ex{\max_{i \in [n]}\set{\ssize{S_i - \hS_i}}}
	\leq \ex{\sum_{i=1}^n R_i\cdot \indic{R_i > 1}} + \ex{\sum_{i=1}^n \Delta_i}
	\leq 2 \mu,
	\end{align*}
	the last inequality holds by \cref{eq:prelim:exp-sum-deltas}. Hence, by Markov inequality
	\begin{align}\label{eq:prelim:prob-dif-S-hS}
	\forall \beta > 0:~\pr{\max_{i \in [n]}\set{\ssize{S_i - \hS_i}} \geq \beta} \leq {2 \mu}/{\beta}
	\end{align}
	A Markov inequality yields that  for any $i \in [n]$:
	\begin{align}\label{eq:prelim:prop_size_Ri}
	\pr{\size{R_i} > \frac12} = \pr{\min\set{\size{R_i},R_i^2} > \frac14} \leq 4\cdot \ex{\min\set{\size{R_i},R_i^2}}
	\end{align}
	Let $E$ be the event that $\size{R_i} \leq \frac12$ for all $i \in [n]$. By \cref{eq:prelim:prop_size_Ri} and a union bound:
	\begin{align}\label{eq:prelim:prop_E}
	\pr{E} \geq 1-4\mu
	\end{align}
	Since $Y_i = \prod_{j=1}^i (1+R_j)$, then conditioned on $E$, we can use the inequality $e^{x-x^2}\leq 1+x \leq e^x$ for $x \in (0,\frac12)$ to deduce that 
	\begin{align}
	e^{\hS_i - \size{S_i - \hS_i} - \sum_{i=1}^n R_i^2} \leq e^{S_i - \sum_{i=1}^n R_i^2} \leq Y_i \leq e^{S_i} \leq e^{\hS_i + \size{S_i - \hS_i}}
	\end{align}
	Note that if $E$ happens, and $\max_{i \in [n]}\set{\ssize{\hS_i}} \leq \frac12 \ln \frac1{1-\lambda}$, and $\max_{i \in [n]}\set{\size{S_i - \hS_i}} \leq \frac14 \ln \frac1{1-\lambda}$ and $\sum_{i=1}^n R_i^2 \leq \frac14 \ln \frac1{1-\lambda}$, then for every  $i \in [n]$:
	\begin{align}
	1-\lambda = e^{- \ln \frac1{1-\lambda}} \leq Y_i \leq e^{\frac34 \cdot \ln \frac1{1-\lambda}} = \frac1{(1-\lambda)^{3/4}} \leq 1+\lambda,
	\end{align}
	the last inequality holds since $\lambda \in (0,\frac14]$.
	The proof  follows since the probability that one of the conditions above does not happen is at most
	\begin{align*}
	\pr{\neg E} 
	&+ \pr{\sum_{i=1}^n R_i^2 > \frac14 \ln \frac1{1-\lambda} \mid E} + \pr{\max_{i \in [n]}\set{\size{S_i}} > \frac12 \ln \frac1{1-\lambda}} + \pr{\max_{i \in [n]}\set{\size{S_i - \hS_i}} > \frac14 \ln \frac1{1-\lambda}}\\
	&\leq 4\mu +\frac{\pr{\sum_{i=1}^n \min\set{\size{R_i},R_i^2} > \frac14 \ln \frac1{1-\lambda}}}{\pr{E}} + \frac{16\mu}{\ln^2 \frac1{1-\lambda}} + \frac{8\mu}{\ln \frac1{1-\lambda}}\\
	&\leq 4\mu +\frac{4\mu}{(1-4\mu)\cdot \ln \frac1{1-\lambda}} + \frac{16\mu}{\ln^2 \frac1{1-\lambda}} + \frac{8\mu}{\ln \frac1{1-\lambda}}\\
	&\leq \frac{\mu}{4 \lambda^2} + \frac{2\mu}{\lambda^2} + \frac{16 \mu}{\lambda^2} + \frac{4 \mu}{\lambda^2}
	\leq \frac{23 \mu}{\lambda^2}.
	\end{align*}
	The first inequality holds by \cref{eq:prelim:prop_E,eq:prelim:prob-dif-S-hS,eq:prelim:prob_max_size_hSi_strong}, the second one by \cref{eq:prelim:prop_E} and by Markov inequality, and the third one holds since $\mu \leq 0.1$, $\lambda \in (0,\frac14]$ and since $\ln \frac1{1-\lambda} \geq \lambda$ and  $\ln^2 \frac1{1-\lambda} \geq \lambda^2$ for $\lambda \in (0,\frac14]$.
\end{proof}

\subsection{Proof of \texorpdfstring{\cref{prop:prelim:martingale-specific-bound}}{Martingales Proposition}}\label{sec:prelim:martingale-specific-bound}

\begin{proposition}[Restatement of \cref{prop:prelim:martingale-specific-bound}]
	\MartingalesProp
\end{proposition}
\begin{proof}
	Let $\tY_0, \tY_1, \ldots, \tY_n$ be the random variables such that for all $i \in [n]$, $\tY_i = Y_{j_{\min}-1}$ where $j_{\min} \in [i]$ is the value with $\size{T_1},\ldots,\size{T_{j-1}} \leq 0.1, \size{T_{j_{\max}}} > 0.1$, letting $j_{\min} = i+1$ (\ie $\tY_i = Y_i$) in case $\size{T_1},\ldots,\size{T_i} \leq 0.1$. Since $T_i$ is a deterministic function of $X_0,\ldots,X_{i-1}$, then $\ex{\tY_i \mid X_0,\ldots, X_{i-1}} = \tY_{i-1}$ (namely, $\tY_0, \tY_1, \ldots, \tY_n$ are martingale w.r.t $X_0,\ldots,X_{n}$). \cref{lemma:prelim:Martingales-new-bound} yields that
	\begin{align}\label{eq:martingale:bound-using-tR_i}
	\pr{\exists i \in [n]\text{ s.t. } \size{\tY_i - 1} \geq \lambda} \leq \frac{23\cdot \ex{\sum_{i=1}^n \min\set{\size{\tR_i}, \tR_i^2}}}{\lambda^2},
	\end{align}
	where $\tR_i = \tY_{i}/\tY_{i-1} - 1$. By definition, for any fixing of $X_0,\ldots,X_{n}$ with $j_{\min} \in [i]$ it holds that $\tR_i = 0$, and for any fixing with $j_{\min} = i+1$ it holds that $\tR_i = \paren{Z_i - T_i}/\paren{1 + T_i}$. In the latter case, it holds that
	\begin{align*}
	\size{\tR_i} \leq \paren{\size{Z_i} + \size{T_i}}/\paren{\size{1 + T_i}} \leq \paren{\size{Z_i} + \size{T_i}}/0.9 \leq 2 \paren{\size{Z_i} + \size{T_i}},
	\end{align*} 
	and that
	\begin{align*}
	\size{\tR_i^2} \leq \paren{2 Z_i^2 + 2 T_i^2}/\paren{1 + T_i}^2 \leq \paren{2 Z_i^2 + 2 T_i^2}/0.9^2 \leq 3 \paren{Z_i^2 + T_i^2},
	\end{align*} 
	where in the first inequality \remove{in \cref{eq:R_i-square} }we used the fact that $(a+b)^2 \leq 2a^2 + 2b^2$.
	Overall, we deduce that for all $i \in [n]$:
	\begin{align}\label{eq:martingale:bounding-tR_i-tR_i-square}
	\min\set{\size{\tR_i}, \tR_i^2} \leq 3 \min\set{\size{Z_i} + \size{T_i}, Z_i^2 + T_i^2} \leq 6 \paren{\min\set{\size{Z_i}, Z_i^2} + \min\set{\size{T_i}, T_i^2}}, 
	\end{align}
	where in the last inequality we use the fact that $\min\set{x+y,x^2+y^2} \leq 2 \min\set{x,x^2} + 2 \min\set{y,y^2}$ for any $x,y \geq 0$.
	By \cref{eq:martingale:bound-using-tR_i,eq:martingale:bounding-tR_i-tR_i-square} we deduce that
	\begin{align}\label{eq:martingale:bound-on-tY_i}
	\pr{\exists i \in [n]\text{ s.t. } \size{\tY_i - 1} \geq \lambda} \leq \frac{138\cdot \ex{\sum_{i=1}^n \paren{\min\set{\size{Z_i}, Z_i^2} + \min\set{\size{T_i}, T_i^2}}}}{\lambda^2},
	\end{align}
	By Markov inequality, for any $i \in [n]$ it holds that
	\begin{align*}
	\pr{\size{T_i} \geq 0.1} = \pr{\min\set{\size{T_i},T_i^2} \geq 0.01} \leq 100 \cdot \ex{\min\set{\size{T_i},T_i^2}}
	\end{align*}
	and therefore
	\begin{align}\label{eq:martingale:tY_i-to_Y_i}
	\pr{\exists i \in [n]\text{ s.t. } \tY_i \neq Y_i} 
	&\leq \pr{\exists i \in [n]\text{ s.t. } \size{T_i} \geq 0.1}\nonumber\\
	&\leq 100\cdot \ex{\sum_{i=1}^n \min\set{\size{T_i},T_i^2}} \leq \frac{7 \cdot \ex{\min\set{\size{T_i},T_i^2}}}{\lambda^2}
	\end{align}
	The proof now follows by \cref{eq:martingale:bound-on-tY_i,eq:martingale:tY_i-to_Y_i}.
\end{proof}

\remove{
\begin{proof}
	We discretized the real non-negative line as follows: let $r\in\naturals$ such
	that $(r-1)\sqrt{\sigma} < 1 \leq r\sqrt{\sigma}$. For an integer $i < r$
	let $\mathcal{S}_i = [(i-1)\sqrt{\sigma},i\sqrt{\sigma})$, let
	$\mathcal{S}_r=[(r-1)\sqrt{\sigma},1]$, and finally let
	$\mathcal{S}_{r+1} = (1,\infty)$. That is, $\mathcal{S}_1$ up to
	$\mathcal{S}_{r-1}$ are disjoint intervals of length $\sqrt{\sigma}$ starting
	from $0$; $\mathcal{S}_r$ is the interval that ends at $1$ of length at
	most $\sqrt{\sigma}$ such that $\mathcal{S}_1,\ldots,\mathcal{S}_{r}$ is a
	discretization of the interval $[0,1]$; and finally, $\mathcal{S}_{r+1}$
	is the remainder of the real non-negative line.

	Let $P_I(i) = \Pr_P[\abs{X}\in \mathcal{S}_i]$ and $Q_I(i) = \Pr_Q[\abs{X}\in
	\mathcal{S}_i]$. By assumption, it holds that
	\begin{align}\label{eq:boundQ}
	Q_I(i) \leq K_2\cdot \exp\paren{-\frac{(i-1)^2}{K_1}} \quad\text{for all $i\leq r$}.
	\end{align}
	Since $\abs{X}\leq 1$ almost surely under the $P$ distribution, it holds that
	$\Ex_P[X^2] \leq \sigma\cdot\sum_{i=1}^rP_I(i)\cdot i^2$. Furthermore, the
	data-processing property of divergence
	(\cref{fact:prelim:diver-properties}(\ref{fact:diver-properties:item:data-processing})) implies that
	$D(P_I||Q_I)\leq D(P||Q)$. Thus, it suffices to show that there exists
	$K_3=K_3(K_1,K_2)>0$ such that
	\begin{align}\label{eq:new-bound}
	\sum_{i=1}^rP_I(i)\cdot i^2 \leq K_3\cdot (D(P_I||Q_I) + 1).
	\end{align}

	The proof relies on the following claim.
	\begin{claim}
		\label{claim:square-manipulations}
		For every integer $\sqrt{2K_1} +1 \leq i\leq r$, it holds that
		\begin{align*}
		P_I(i)\cdot i^2 \leq 8K_1\cdot P_I(i)\log\frac{P_I(i)}{Q_I(i)} + 4K_2\cdot(i-1)^2\cdot\exp\paren{-\frac{(i-1)^2}{2K_1}}.
		\end{align*}
	\end{claim}
	Before proving the claim, we use it to derive the proposition.
	Let $\ell= \lfloor \sqrt{2K_1} +1 \rfloor$. We write
	\begin{align*}
	\sum_{i=1}^{r} P_I(i)\cdot i^2 &= \sum_{i=1}^{\ell} P_I(i)\cdot i^2 +
	\sum_{i=\ell+1}^{r} P_I(i)\cdot i^2\\
	&\leq \ell^3 + 8K_1\sum_{i=\ell+1}^{r} P_I(i)\log\frac{P_I(i)}{Q_I(i)} +
	4K_2\sum_{i=\ell+1}^{r} \cdot(i-1)^2\cdot\exp\paren{-\frac{(i-1)^2}{2K_1}}.
	\end{align*}
	It is easy to verify that the sum
	$\sum_{i=\ell+1}^{r} \cdot(i-1)^2\cdot\exp\paren{-\frac{(i-1)^2}{2K_1}}$ and
	$\ell^3$ are bounded from above by some function of $K_1$ ($\ell^3$ is simply
	bounded by $(\sqrt{2K_1}+1)^2$ and the sum can be bounded by, for example,
	comparing it to the expected value of an exponential random
	variable). Thus, there exists $K_4=K_4(K_1,K_2)>0$ such that
	\begin{align}
	\sum_{i=1}^{r} P_I(i)\cdot i^2 &\leq 8K_1\sum_{i=\ell+1}^{r} P_I(i)\log\frac{P_I(i)}{Q_I(i)} + K_4\nonumber\\
	&= 8K_1\sum_{i=\ell+1}^{r}P_I(i)\log\frac{P_I(i)}{Q_I(i)} +
	8K_1\cdot\sum_{i=1}^{\ell}P_I(i)\log\frac{P_I(i)}{Q_I(i)} + 8K_1\cdot
	P_I(r+1)\log\frac{P_I(r+1)}{Q_I(r+1)} \label{eq:eqdeiver}\\
	&\quad - 8K_1\cdot\sum_{i=1}^{\ell}P_I(i)\log\frac{P_I(i)}{Q_I(i)} - 8K_1\cdot
	P_I(r+1)\log\frac{P_I(r+1)}{Q_I(r+1)}+ K_4.\label{eq:what-isleft}
	\end{align}
	The terms in \cref{eq:eqdeiver} are equal, by definition, to $8K_1\cdot D(P_I||Q_I)$. We now
	want to bound the terms in \cref{eq:what-isleft}:
	\begin{align*}
	&- 8K_1\cdot\sum_{i=1}^{\ell}P_I(i)\log\frac{P_I(i)}{Q_I(i)} - 8K_1\cdot
	P_I(r+1)\log\frac{P_I(r+1)}{Q_I(r+1)} +K_4\\
	& = -8K_1\cdot\paren{\sum_{i=1}^{\ell}Q_I(i)\frac{P_I(i)}{Q_I(i)}\log\frac{P_I(i)}{Q_I(i)}
		+ Q_I(r+1)\frac{P_I(r+1)}{Q_I(r+1)}\log\frac{P_I(r+1)}{Q_I(r+1)}}+K_4\\
	&\leq -8K_1\cdot\paren{\sum_{i=1}^{\ell}Q_I(i)(-e^{-1}) +
		Q_I(r+1)(-e^{-1})}+K_4\\
	&\leq  8K_1\cdot e^{-1} + K_4,
	\end{align*}
	where the first inequality follows from the fact that the function $x\log
	x\geq -e^{-1}$ for all $x>0$. Combining the above two equations and taking
	$K_3=8K_1 + K_4$, we have that
	\begin{align*}
	\sum_{i=1}^{r} P_I(i)\cdot i^2 \leq K_3(D(P_I||Q_I) +  1).
	\end{align*}
	This establishes \cref{eq:new-bound} and completes the proof of the
	proposition, assuming the above claim.

	All that is left is to prove \cref{claim:square-manipulations}, which we do next.
	\begin{proof}[Proof of \cref{claim:square-manipulations}]
		The proof splits to the following complement cases.

		If $P_I(i)\geq \sqrt{K_2Q_I(i)}$, then
		\begin{align*}
		8K_1\cdot P_I(i)\log\frac{P_I(i)}{Q_I(i)} &\geq 8K_1\cdot P_I(i)\log\sqrt{\frac{K_2}{Q_I(i)}}\\
		& \geq 8K_1\cdot P_I(i)\cdot \frac{(i-1)^2}{2K_1}\\
		&= 4\cdot P_I(i)\cdot (i-1)^2\\
		&\geq P_I(i)\cdot i^2,
		\end{align*}
		where the second inequality follows from \cref{eq:boundQ} and the last
		inequality holds since $4(x-1)^2\geq x^2$ for all $x\geq 2$ and since the
		assumptions of the claim imply that $i\geq 2$.

		For the complement case, if $P_I(i)< \sqrt{K_2Q_I(i)}$, then
		\begin{align}\label{eq:boundP}
		\frac{P_I(i)}{K_2} < \sqrt{\frac{Q_I(i)}{K_2}} \leq
		\exp\paren{-\frac{(i-1)^2}{2K_1}} \leq e^{-1},
		\end{align}
		where the second inequality follows from \cref{eq:boundQ} and the last
		inequality follows from the assumption of claim that $(i-1)^2\geq 2K_1$. It
		holds that
		\begin{align*}
		8K_1\cdot P_I(i)\log\frac{P_I(i)}{Q_I(i)} &=
		8K_1\cdot P_I(i)\log\frac{P_I(i)}{K_2} + 8K_1\cdot P_I(i)\log\frac{K_2}{Q_I(i)}\\
		& \geq 8K_1K_2\frac{P_I(i)}{K_2}\log\frac{P_I(i)}{K_2} + 8K_1\cdot P_I(i)\cdot \frac{(i-1)^2}{K_1}\\
		& \geq 8K_1K_2 \cdot \exp\paren{-\frac{(i-1)^2}{2K_1}} \cdot \paren{-\frac{(i-1)^2}{2K_1}} + 4\cdot P_I(i)\cdot (i-1)^2\\
		& \geq -4K_2\cdot(i-1)^2 \cdot \exp\paren{-\frac{(i-1)^2}{2K_1}} + P_I(i)\cdot i^2,
		\end{align*}
		where the first inequality follows from \cref{eq:boundQ}, the second
		inequality follows from \cref{eq:boundP} and since the function
		$x\mapsto x\log x$ is decreasing in $(0,e^{-1})$, and the last inequality
		again holds since $4(x-1)^2\geq x^2$ for all $x\geq 2$ and since the
		assumptions of the claim imply that $i\geq 2$.
	\end{proof}
	Having proved \cref{claim:square-manipulations}, the proof of
	\cref{prop:prelim:sub-exp-to-divergence} is now complete.
\end{proof}
} 


\end{document}